\documentclass[twocolumn,prx,hyperref,footinbib,longbibliography]{revtex4-1} 

\usepackage{graphicx,color,mathbbol,amsthm,amsmath,amsfonts,amssymb,bm,dsfont,footmisc,mathrsfs,multirow,latexsym,times}

\usepackage[usenames,dvipsnames,reprint,aps]{xcolor}

\usepackage[utf8]{inputenc}
\usepackage[english]{babel}
\usepackage[T1]{fontenc}
\usepackage{amsmath,amsfonts,amsthm}
\usepackage{graphicx,color,soul}
\usepackage{hyperref} 
\hypersetup{
    colorlinks=true,
    linkcolor=red,
    filecolor=blue,      
    urlcolor=blue,
    citecolor=blue
}

\usepackage{stmaryrd}
\usepackage{bm}
\usepackage{booktabs}
\usepackage{hyperref}

\newtheorem{theorem}{Theorem}
\newtheorem{prop}{Proposition}

\newtheorem{conclusion}{Conclusion}[theorem]
\newtheorem{definition}{Definition}[theorem]
\DeclareMathOperator{\Tr}{Tr}
\newcommand{\expect}[1]{\langle #1 \rangle}

\usepackage{mathtools}
\DeclarePairedDelimiter{\ceil}{\lceil}{\rceil} 
\DeclarePairedDelimiter{\floor}{\lfloor}{\rfloor}
\DeclareMathOperator*{\argmax}{arg\,max}  

\begin{document}

\title{Efficient learning and optimizing non-Gaussian correlated noise in digitally controlled qubit systems}


\author{Wenzheng Dong}
\email{dongwz@vt.edu}
\email{wenzheng.dong.quantum@gmail.com}
\affiliation{Department of Physics, Virginia Tech, Blacksburg, Virginia 24061, USA}
\affiliation{Shenzhen International Quantum Academy, Futian District, Shenzhen, P. R. China}
\affiliation{Shenzhen Institute for Quantum Science and Engineering,
Southern University of Science and Technology, Shenzhen, P. R. China}

\author{Yuanlong Wang} 		        
\affiliation{Key Laboratory of Systems and Control, Academy of Mathematics and Systems Science, Chinese Academy of Sciences, Beijing 100190, People’s Republic of China}
\author{Muhammad Qasim Khan} 
\affiliation{Department of Physics and Astronomy, Dartmouth College, Hanover, New Hampshire 03755, USA}

\begin{abstract}
Precise qubit control in the presence of spatio-temporally correlated noise is pivotal for transitioning to fault-tolerant quantum computing. Generically, such noise can also have non-Gaussian statistics,
which hampers existing non-Markovian noise spectroscopy protocols. By utilizing frame-based characterization and a symmetry analysis, we show how to achieve higher-order spectral estimation for noise-optimized circuit design. Remarkably, we find that the digitally driven qubit dynamics can be solely determined by the complexity of the applied control, rather than the non-perturbative nature of the non-Gaussian environment. This enables us to address certain non-perturbative qubit dynamics more simply. We delineate several complexity bounds for learning such high-complexity noise and demonstrate our single and two-qubit digital characterization and control using a series of numerical simulations. Our results not only provide insights into the exact solvability of (small-sized) open quantum dynamics but also highlight a resource-efficient approach for optimal control and possible error reduction techniques for current qubit devices.
 \end{abstract}

\date{\today}

\maketitle

\section{Introduction}
Fault-tolerant quantum processors demand control precision over physical qubits beyond the threshold theorem~\cite{Aharonov_2008,Knill_science_1998,Gottesman_QEC_FT}.
However, these systems are delicate, often undermined by interactions with complex environments that introduce spatially and temporally correlated —or ``non-Markovian''~\cite{non_markovian_comment} —noise~\cite{Sung2019,Uwe2020,Dykman_2022,Wudarski_PRL_2023,Maloney_PRA_2022},  which can degrade  processor performance significantly~\cite{Paladino_PRL_2002, Lidar2019Feb, Breuer_book, Breuer2016Apr, Li2018Oct, Paz_PRL_2014, deVega2017Jan, Clerk2010Apr, Breuer2016Apr}. The ``colored'' nature of such noise usually degrades the plug-and-play dynamical decouplings and breaks the basic assumption of stochastic error models behind the error correction code~\cite{Gottesman_class}, necessitating active noise diagnosis and correction~\cite{yan2013rotating,Sung2019,Qasim_SPAM_paper,Frey_NatCom_2017,Uwe2020,Zhou_crosstalk_PRL,Romero2021,white2023unifying,Berk_npjQI_2023, Sankar2008,Burgelman_PRQ_2025}.  Optimization at the circuit level~\cite{Harper_NatPhy,White2020,Vezvaee2024Jul,Ruiz2024Feb, Zeng2023May} requires additional efforts beyond standard gate-level control methods~\cite{Ezzell2023Dec, Stefanatos2021Mar, Daems2013Jul,  Barnes_curve_QST, Tang_curve_PRApp, Nelson_curve_PRA, Calderon_Vargas2017Apr, Kestner2013Apr, Kanaar2024Apr,Wang2012Aug}, as characterizing and controlling (C\&C) non-Markovian noise~\cite{Teerawat_PRXQ_Frame,Paz_PRL_2014,Norris_PRL_2016,Frey_NatCom_2017,White2020, white2023unifying, Wang_RBL_PRL_2019, Wu_RBL_PRL_2024, Shen_RBL_PRL_2023} presents substantial challenges. As quantum processors scale up and circuit depths increase, the complexity of large-scale, non-Markovian dynamics intensifies, demanding novel C\&C strategies~\cite{Aloisio_PRQ_2023, Gambetta2002Jul, Rebentrost_2009_PRL, Trivedi2022Apr, Harper_NatPhy, Berk_npjQI_2023, Dykman_2022}.

Quantum noise spectroscopy (QNS)~\cite{Alvarez_Suter_RPL_2011,Norris_PRL_2016,Youssry_npj_2020,Ferrie_2018,Paz_multiaxis_PRA_2019,Multilevel,Vezvaee_2022,Paola_Walsh,Szankowski_2017,Frey_PRApp_2020,Frey_NatCom_2017,yan2013rotating,Sung2019,Yuge_PRL_2011,Norris_PRA_2018,Uwe2020} has become an essential tool for probing correlated noise in quantum systems. Existing QNS protocols~\cite{Paz_mutilqubit_PRA_2017,Szankowski_2017,LukeRTN,DongWZ_APL} typically truncate the qubit dynamics at the  few lowest orders (Gaussian or leading non-Gaussian) of the parameter space, focusing on learning the noise within these truncated regions. However, extending QNS to account for strong non-Gaussian noise requires higher-order truncation, which becomes increasingly complex as the dimensionality of the parameter space grows. Despite recent advancements, several challenges remain in C\&C of non-Gaussian, non-Markovian noise. Probing non-Gaussian noise~\cite{Norris_PRL_2016,Ramon_trispectrum_PRB_2019, Wang_RBL_PRL_2019,Wu_RBL_PRL_2024} often requires a sophisticated control repertoire, which is beyond the capabilities of many current quantum devices. More importantly, there is no established rule of thumb for selecting the optimal truncation order, and no formal theory exists to quantify the control resources required to learn and mitigate non-Gaussian noise. Intuitively, higher truncation orders yield better results; however, this comes at the cost of pushing the control resource overhead within a limit of the curse of dimensionality in non-perturbative noise. One recent work by some of us~\cite{Teerawat_PRXQ_Frame} sheds light on solving this by introducing a concept of control adaptation, where only specific components of the full noise information, relevant to qubit dynamics, are considered for each fixed-dimensional noise space.  This emphasizes the need for a comprehensive understanding of sampling complexity in learning non-Gaussian noise.

In this work, we address these challenges in a specific control setting --- quantum systems under \textit{digital control}. Here, digital control refers to the application of control pulses that are effectively instantaneous, meaning the time separation between consecutive pulses is significantly longer than the duration of each pulse. Our formalism is achieved by integrating frame-based C\&C with a novel control-based symmetry analysis. For digitally driven quantum systems subject to general non-Gaussian \textit{dephasing} noise, we show that resource requirements for non-Gaussian noise C\&C only depend on the complexity of the applied digital control (i.e., the size of the quantum circuit), rather than the intrinsic complexity of noise itself (i.e., the truncation order of non-Gaussianity). This insight distinguishes our approach from \textit{all} traditional QNS methods and provides a pathway for the effective C\&C of non-perturbative noise.  Based on successful noise spectral estimation from our fundamental digital QNS, we numerically showcase our method's capability to accurately predict qubit dynamics and optimize a two-qubit quantum circuit under spatially and temporally correlated noise. Our analysis offers new perspectives on the minimal control resources required for managing high-complexity non-Gaussian noise in digital quantum circuits. Although our approach is not generally scalable, it provides significant value for open-loop control-based fault tolerance, particularly in settings with minimal noise distribution assumptions. We refer readers primarily interested in applications to Fig.~\ref{fig:teaser} for a schematic illustration of the noise C\&C in our work.

The content is organized as follows. 
Sect.~\ref{sec:Hamiltonian} introduces the Hamiltonian-level setup of open quantum dynamics under non-Markovian dephasing noise and illustrates Dyson series of reduced system dynamics. Sect.~\ref{sec:NonGaussian_complexity} details  how the simultaneous presence of temporal correlations and non-Gaussianity exponentially escalates C\&C resource demands. 
In Sect.~\ref{sec:Frames}, through the lens of the \textit{frame-based filter-function formalism},  we review the concept of window frames and control-adapted spectra. Crucially, we illustrate why model reduction and resource efficiency are crucial in the C\&C of non-Gaussian noise.
Sect.~\ref{sec:Symmetry_analysis} presents novel insights into control symmetries and their impact on the noise spectra. This symmetry analysis is a groundbreaking addition to conventional methods, significantly altering the QNS sample complexity.
In Sect.~\ref{sec:Fundamental_QNS}, we present several interesting conclusions and numerical simulations on exact qubits C$\&$C, culminating in the finding that the full complexity of open quantum circuit dynamics is governed by circuit size rather than noise complexity. This conclusion empowers us to manage C\&C for single- and two-qubit dynamics that exhibit non-perturbative characteristics, which existing QNS protocols cannot address.
We conclude in Sect.~\ref{sec:Discussion}, discussing potential developments and scalability limitations. Several appendices include technical details supporting our approach and examples.

\begin{figure}
    \centering
    \includegraphics[width=1.0\linewidth]{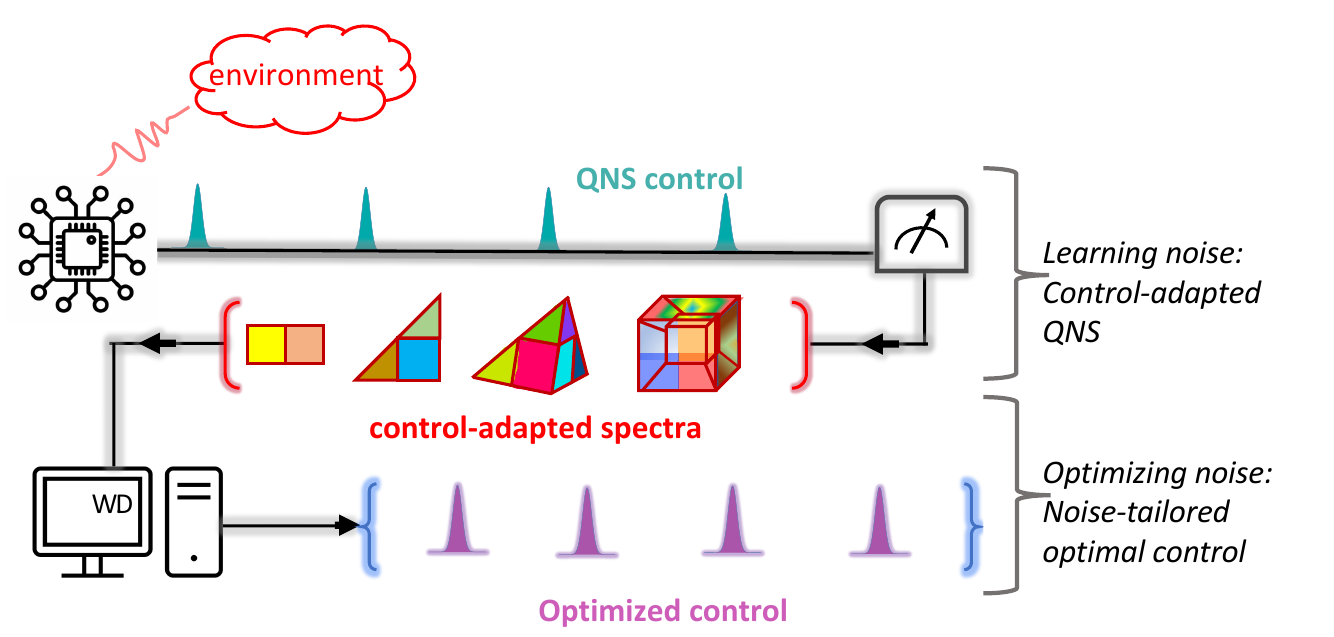}
\caption{\textbf{Illustration of the practical highlight of this work.} This figure highlights two key aspects: \textbf{Learning noise} —-- employing control-adapted quantum noise spectroscopy (QNS) to extract a model-reduced form of spectra (control-adapted spectra) relevant to the applied digital control; and \textbf{Optimizing noise} —-- designing noise-tailored optimal control based on the learned control-adapted spectra to physically mitigate their noise effects in quantum circuits.}
    \label{fig:teaser}
\end{figure}

\section{Controlled multi-qubit dynamics under general dephasing noise}
\label{sec:Hamiltonian}
\subsection{Open-system model and control assumptions}
We consider a controlled multi-qubit system $Q$ interacting with a general dephasing environment $E$, because it is more dominant in various platforms. In a rotating frame defined by the free evolution of the system and environment, the total Hamiltonian can be expressed as
\begin{equation}
    H(t)=H_{QE}(t) + H_{\text{ctrl}}(t),  \qquad 
    H_{QE}(t) = \sum_{q\in Q} \sigma_z^{[q]} \otimes B_q(t).
\end{equation}
Our analysis focuses on the interplay between the system $Q$-only control Hamiltonian $H_{\text{ctrl}}(t)$ and the system-environment interaction $H_{QE}(t)$. In $H_{QE}(t)$, the one-hot Pauli $\sigma_z^{[q]}$ denotes  $\sigma_z$ acting on qubit $q$, with the identity on the remaining qubits $q'\in Q \setminus \{q\}$; such a form arises from the fact that evolution of $Q$ preserves the dephasing. $B_q(t)$ represents the environmental operator coupling to  qubit $q$ in $Q$, whereby its strength ($||B(t)||_2$) implicitly indicating the  $QE$ coupling strength (e.g., a coupling constant $g$). Note that our setup allows different $B_q(t)$ to act on the same environment subsystem; consequently, tracing out environment degrees of freedom results in effective noisy many-body interactions within $Q$ (e.g., $ZZ$-crosstalk). 

For such a Hamiltonian, the error propagation can be perturbatively calculated in the toggling frame, which corresponds to the interaction picture with respect to $H_{\text{ctrl}}(t).$ Defining the control propagator as $U_0(t) = \mathcal{T}_+ e^{-i \int^t_0 H_{\text{ctrl}}(s)ds}$,  the effective error Hamiltonian is expressed as
\begin{equation}
\begin{aligned}
    \tilde{H}(t) &= U_0^{\dagger}(t) H_{QE}(t) U_0(t)  \\
    &\equiv \sum_{q \in Q} \tilde{H}_q(t) \equiv \sum_{q \in Q} \tilde{h}_q(t) \otimes B_q(t)|_{\tilde{h}_q(t) = U^{\dagger}_0(t) \sigma_z^{[q]}U_0(t)}  \\
    &=\sum_{q \in Q} \sum_{u=1}^{4^{|Q|}-1}y_{[q],u}(t) \Lambda_u \otimes B_q(t), 
    \label{eq:eq_1}
\end{aligned}
\end{equation}
where the switching function $y_{[q],u}(t) =\Tr[U_0^{\dagger}(t)\sigma_z^{[q]}  U_0(t) \Lambda_u]/2^{|Q|}$, and $\Lambda_u$ denotes a generator in Pauli group $\mathcal{P}_{|Q|}$, excluding identity and phase factors. That said, $\Lambda_u=\sigma_u$ ($u=\{x,y,z\}$) for $|Q|=1$, and $\Lambda_u=\{ \Lambda_1,\cdots,\Lambda_{15}\} = \{\sigma_0\otimes \sigma_x, \cdots,\sigma_z\otimes \sigma_z \}$ for $|Q|=2$. 
We focus on invertible Hermitian system observable ($O=O^{\dagger}$ and $OO^{\dagger}=\mathbb{I}_Q$) throughout this work, e.g. $O=\sigma_z\otimes\sigma_z$.  Following an approach used in \cite{Teerawat_PRXQ_Frame} and \cite{DongWZ_APL}, we define an auxiliary Hamiltonian
\begin{equation*}
\begin{aligned}
       &{H}_O(t) = 
    \begin{cases}
    \bar{H}(T-t), ~& 0\leq t \leq T,\\
   \widetilde{H}(T+t),  ~& -T\leq t <0,\\ 
    \end{cases}\\
    &\bar{H}(t)\equiv -\widetilde{O}^{-1}(T) \widetilde{H}(t) \widetilde{O}(T),
\end{aligned}
\end{equation*}
where $\tilde{O}(T)=U_0^{\dagger}(T) O U_0(T)$. 
The expectation value of $O$ after system evolution for time $T$ is
\begin{equation}
    \mathbb{E}[O(T)]_{\rho_Q\otimes \rho_E}=\Tr_Q\Big [\expect{\mathcal{T}_{+} e^{-i \int^T_{-T}{H}_O(t)dt }} \rho_Q \tilde{O}(T)\Big],
\end{equation}
where $\rho_Q\otimes \rho_E$ is the joint factorizable initial state and $\expect{\cdot}\equiv\expect{\Tr_E[\cdot \rho_E]}_{\text{C}}$ is a classical-quantum average over $E$.

\subsection{Perturbative representation of the reduced-system dynamics}

The above quantity can be computed using the Dyson series as $ \mathbb{E}[O(T)]= \Tr[\sum_{k=0} ^{\infty}   \mathcal{D}^{(k)}_O(T)/k!  \rho_Q  \tilde{O}(T) ] $.  Inspired by established techniques in \cite{Teerawat_PRXQ_Frame} [Eq.\,(A3) therein] and in \cite{DongWZ_APL} [Eq.(3) therein], the detailed Dyson term reads
\begin{widetext}
\begin{equation}
\begin{aligned}
      \mathcal{D}^{(k)}_O(T)/k! &=(-i)^k \sum_{\vec{q}_{[k]}\in Q^{\otimes k}}\sum\limits^{k}_{l=0}\sum\limits_{\pi\in \Pi_{l;k}} \int^T_{0}d_>\vec{t}_{[k]} \Big\langle \prod\limits^l_{j=1}\bar{H}_{q_j}(t_{\pi(j)}) \prod\limits^k_{j'=l+1} \tilde{H}_{q_{j'}}(t_{\pi(j')})\Big\rangle\\
      =(-i)^k&  \sum_{\substack{\vec{q}_{[k]}; ~l; \\ \pi; ~\vec{u}; ~\vec{c}}}   \int^T_0d_>\vec{t}_{[k]} 
      \Big[  \prod_{j=1}^l\prod_{j'=l+1}^k f^{c_j}_{u_{j}}y_{[q_j],c_j}(t_{\pi(j)})y_{[q_{j'}],u_{j'}}(t_{\pi(j')})\Lambda_{u_j}\Lambda_{u_{j'}}  \Big]  \\
      & \quad \quad \quad \quad \quad \times  \sum_{\vec{\mu}_{[k-1]}\in \{0,1\}^{\otimes k-1}}  (-1)^{\bar{f}^{(k)}_{\pi}(\vec{\mu}_{[k-1]})}  \expect{\mathcal{B}^{\vec{\mu}_{[k-1]}}_{ \vec{q}_{[k]}} (t_1, ..., t_k)}.
\end{aligned}
\label{eq:dyson}
\end{equation}
\end{widetext}
Here $d_{>}\vec{t}_{[k]}$ indicates time-ordered integral such that $t_1\geq \cdots \geq t_k$, and $\Pi_{l;k}$  denotes the set of all admissible permutations in the expansion. Specifically, $\pi\in \Pi_{l;k}$ imposes that $t_{\pi(j)}$ in $l$-length $\bar{H}$ string must be nondecreasing, while $t_{\pi(j')}$ in $(k-l)$-length $\tilde{H}$ string must be nonincreasing, adhering to time ordering. A conjugation factor $f^{c}_u \equiv -\frac{1}{2^{|Q|}} \Tr[\tilde{O}^{\dagger}(T) \Lambda_{u} \tilde{O}(T) \Lambda_{c}] $ is introduced to help organize the general Pauli product $\prod_{j} \Lambda_{u_j}$.
Regarding $E$ , the noise correlator is expressed in a  \textit{nested bracket} representation~\cite{Wang_RBL_PRL_2019,Gasbarri_PRA_2018}:
\begin{equation}
\begin{aligned}
    \mathcal{B}^{\vec{\mu}_{[k-1]}}_{ \vec{q}_{[k]}} (\vec{t})\equiv     \frac{1}{2^{k-1}} \widehat{P}_{\vec{q}}\big( & \llbracket ...\llbracket B_{q_1}(t_1), B_{q_2}(t_2) \rrbracket_{\mu_1}, \\
    & B_{q_3}(t_3) \rrbracket_{\mu_2}, ..,  B_{q_k}(t_k) \rrbracket_{\mu_{k-1}} \big),
\end{aligned}
\end{equation} 
where  each bracket marked by a sign $\mu\in\{0,1\}$ gives either a commutator or an anti-commutator $\llbracket X,Y \rrbracket_{\mu} = XY+(-1)^{\mu} YX$. The operator $\widehat{P}_{\vec{q}}$ is a $q$ ordering operator to ensure that the expanded operator string is consistently indexed by $\{q_1,...,q_k\}$. 
However, using nested bracket representation, which aligns all times in order, introduces a complexity: all noise correlators enter the full dynamics with a binary phase factor of $(-1)^{\bar{f}^{(k)}_{\pi}(\vec{\mu})}$ determined by $\pi$ and $\vec{\mu}$.
The above discussion of the algebraic details in Eq.~\ref{eq:dyson} is indeed high level, and we refer  readers to the Appendix~\ref{app:dyson} for  elaboration.   

In summary, the Dyson expression detailed above constitutes a superoperator that accurately captures the reduced dynamics of the system $Q$. This occurs through the convolution of control specific to the system and noise correlators from the environment. Gaining insight into and manipulating this reduced dynamics is crucial for the progress of quantum technologies, especially when the environment cannot be controlled~\cite{Khodjasteh_PRL_2009,Khodjasteh_PRL_2010,Glaser2015Dec,Wang2012Aug, Khodjasteh_PRL_2005, Kestner2013Apr, Kabytayev_PRA_2014, Barnes_curve_QST}. 

\section{Challenges in non-Gaussian non-Markovian noise Characterization}
\label{sec:NonGaussian_complexity}

Dynamical decoupling has been successfully demonstrated for noise suppression in qubit devices~\cite{Viola_PRA_1998, Viola_PRL_1999, Biercuk2009Apr, Uys2009Jul, Biercuk2011, Green2015Mar, Malinowski_NatTech_2017,Pokharel_PRL_2018, Ezzell2023Dec}. However, its efficiency diminishes when noise deviates from static and correlations become significant. In this situation, it is believed that understanding such ``colored'' noise~\cite{Paladino_PRL_2002, Paladino_RMP_2014, Alvarez_Suter_RPL_2011, Norris_PRL_2016, Paz_mutilqubit_PRA_2017} and designing optimal control strategies are essential to mitigating its impact on qubit dynamics.

\subsection{Growing complexity in non-Gaussian noise characterization}

In quantum control, gate-level optimization often assumes that the aforementioned reduced system dynamics is Markovian, which corresponds to the leading-order term ($k=1$) in the Dyson series. This assumption, based on the rapid gate speed ($T\rightarrow 0$) or a very weak coupling regime ($g\rightarrow 0$),  simplifies control design by rendering the reduced system dynamics temporally uncorrelated and memoryless. This simplification facilitates numerous exactly solvable models for robust gate design~\cite{Barnes_curve_QST, Walelign2024Dec, xiuhao_PRL}. However, as soon as the dynamics exhibit nontrivial $k=2$ contributions, the dynamics become non-Markovian, introducing temporal correlations that complicate the behavior of the qubit system. 
These temporal correlations are captured by two-point noise correlators: $\expect{\mathcal{B}^{0}_{(q_1,q_2)} (\vec{t})} = (\expect{B_{q_1}(t_1)B_{q_2}(t_2)} + \expect{B_{q_1}(t_2)B_{q_2}(t_1)})/2$ and $\expect{\mathcal{B}^{1}_{(q_1,q_2)} (\vec{t})} = (\expect{B_{q_1}(t_1)B_{q_2}(t_2)} - \expect{B_{q_1}(t_2)B_{q_2}(t_1)})/2$, which correspond to classical and quantum noise correlators, respectively~\cite{Paz_PRL_2014, Paz_mutilqubit_PRA_2017, Teerawat_PRXQ_Frame, YuanlongWang_Q_noise}. For $|Q|>1$,  we shall label the correlator as ``self'' and ``cross'' in situations of $q_1=q_2$ or $q_1\neq q_2$, respectively.   The naming of classical and quantum is evident in a single-qubit case, where noise is fully described by a classical stochastic process, say, $B(t)=\beta(t) \mathbb{I}_E$ with $\mathbb{I}_E$  the identity on $E$. In such a case, only the classical correlator $\expect{\mathcal{B}^{0} (\vec{t})}=\expect{\beta(t_1)\beta(t_2)\rho_E(t)}=\expect{\beta(t_1)\beta(t_2)}_{\text{C}}$ is non-zero; 
and its Fourier transform gives the well-known \textit{power spectral density} (also called noise spectrum) $S_1(\omega) \equiv \int^{+\infty} _{-\infty} d\tau e^{-i\omega \tau} \expect{\beta(\tau) \beta(0)}_{\text{C}}$. Similarly, the Fourier transform of $\expect{\mathcal{B}^{1} (\vec{t})}$ gives the quantum noise spectrum in the frequency domain. The aforementioned classical noise process has vanishing quantum noise spectra under regular commutation. 
If $B(t)$ at different times does not commute,  both classical and quantum noise correlators, as well as their corresponding noise spectra, remain non-zero. 

Accurate estimation of noise power spectra, \textit{\`a la} quantum noise spectroscopy (QNS), serves as a fundamental diagnostic for analyzing and mitigating noise in quantum systems, which has been a mainstay in noise C\&C~\cite{Alvarez_Suter_RPL_2011, yan2013rotating, Paz_mutilqubit_PRA_2017, Szankowski_2017, Frey_NatCom_2017, Uwe2020, Teerawat_PRXQ_Frame, Qasim_SPAM_paper}. 
In QNS protocols, a repertoire of control sequences is applied to learn the noise information from system-only readouts. The QNS \textit{sample complexity}, representing the overhead of control sequences, is equivalent to the number of spectra to be estimated. 
Recent advances in learning two-qubit \textit{cross-power spectral density} (the Fourier transformed spectrum of $\expect{\mathcal{B}^{0}_{(A,B)} (\vec{t})})$ have highlighted the environment-mediated qubit crosstalk as a significant barrier to the scalability of qubit devices~\cite{Uwe2020,tarucha_NatPhy, tarucha_PRApp}.

In a strong coupling regime and at sufficiently long evolution times, $k=3$ or higher non-Gaussian noise contributions emerge (provided noise statistics are not strictly Gaussian; otherwise, $k=2$ cumulant expansion suffices), reflecting increasingly complicated environmental correlations. At this point, higher-order multiple-time and multi-location correlators (or spectra) must be considered, leading to a richer yet more complicated representation of noise. For example, the so-called \textit{bispectrum} can be defined for the Fourier transform of $\expect{\mathcal{B}^{(0,0)}(t_1,t_2,t_3)}$~\cite{Mendel1991Mar, Norris_PRL_2016, Ramon_trispectrum_PRB_2019}.  However, the cross- and quantum versions of non-Gaussian noise remain largely unexplored because of their high complexity.

As $k$ increases further, the QNS sample complexity grows exponentially. While Gaussian noise can be systematically addressed through standard QNS, learning non-Gaussian noise requires advanced techniques that often involve higher sampling complexity and specially designed control protocols. For general noise statistics where $B_q(t)$ at different $t$ don't vanish under commutators, there is just one classical noise term ($\vec{\mu}\equiv\vec{0}$) and $2^{k-1}-1$ quantum noise terms ($\vec{\mu}\neq\vec{0}$)~\cite{PingWang_PRL_2019}. In addition, there are $|Q|$  self-spectra and  $|Q|^k-|Q|$ cross-spectra. Altogether, there are $2^{k-1}|Q|^k$ possible combinations when accounting for all locations and signs (of brackets) at $k$-th order~\cite{Wang_RBL_PRL_2019}. 
When temporal correlation details are incorporated, the dimension of the noise space, as well as the QNS sample complexity balloons significantly. In practice, QNS protocols are truncation-based, whereby only noise correlators of orders up to $k=K$ are considered, neglecting all higher-order ($k>K$) contributions. In fact, learning strong non-Gaussian noises is caught between a rock and a hard place: Higher-order truncation promises better spectral reconstruction at a cost of large throughput data to process, while low-order truncation is troubled by more severe spectral reconstruction error. The $k>K$ order spectra of strong non-Gaussian noise, if carelessly neglected, entered the reduced dynamics undesirably and the  spectral reconstruction produced  deviates by those higher-order spectra on the small learning space.  
Due to these challenges, existing QNS protocols typically consider the leading orders (up to $K=4$)~\cite{Norris_PRL_2016, Ramon_trispectrum_PRB_2019, DongWZ_APL,Wu_RBL_PRL_2024} of non-Gaussian noise.
However, truncating the perturbative expansion at $K=4$ is not always sufficient. Much higher-order non-Gaussian signatures have been observed in realistic quantum systems~\cite{McCourt_DD_learn_noise,Neder_NatPhy_2007,Galperin_PRB_2007,Kotler_PRL_2013}. For example, Ref.~\cite{McCourt_DD_learn_noise} observed non-perturbative features in non-Gaussian contributions~\cite{Santos_PRA_2005} ($k\gg 10$), which standard QNS protocols struggle to work. In summary, the shift from Gaussian to non-Gaussian noise represents a central challenge in the C\&C of general open quantum systems.

\subsection{Overview of frequency-based noise spectral estimation}
Most QNS protocols work in the frequency domain for spectral estimation. We now briefly review frequency-based noise characterization to highlight its limitations. Take the example of a single qubit dephased by stationary zero-mean classical Gaussian noise that is described by a stochastic process $\beta(t)$. In this scenario, the noise correlator is $\expect{\beta(t_1-t_2) \beta(0)}_{\text{C}}$. The frequency-based noise characterization in a standard picture (SP) starts with focusing on the noise correlator.
The power spectral density can be calculated as
\begin{equation}
S_1(\omega) \equiv \int^{+\infty} _{-\infty} d\tau e^{-i\omega \tau} \expect{\beta(\tau) \beta(0)}_{\text{C}}.
\end{equation}
In frequency or comb-based QNS, the qubit control $\tilde{\mathcal{C}}_{\text{comb}} $  consists of long sequences of dynamical decoupling pulses; the fast-flipped switching function is Fourier transformed to form the filter function $F(\omega, T)$~\cite{Alvarez_Suter_RPL_2011}. In this case, the qubit dynamics is sufficiently determined by a function $ \int^{+\infty} _{-\infty} d \omega |F(\omega,T)|^2 S_1(\omega)/4\pi $.  
When estimating non-Gaussian noise spectra, hyper-comb-based QNS~\cite{Norris_PRL_2016} is employed to probe function $\int^{+\infty} _{-\infty} d \vec{\omega} \\ \Big[ \prod^k_{j=1}  F(\omega_j,T) \Big]S_{k-1}(\vec{\omega}),$ where $S_{k-1}(\vec{\omega}) =$ $$  \int^{+\infty}_{-\infty}  d\vec{\tau}_{[k-1]}  e^{-i\vec{\omega} \cdot\vec{\tau}}\expect{\beta(\tau_{k-1})\beta({\tau_{k-2}})..\beta(\tau_1)\beta(0)}_{\text{C}}$$ represents  frequency-based \textit{polyspectrum} (bispectrum for $S_{2}(\vec{\omega})$ and trispectrum for $S_{3}(\vec{\omega})$). 

Frequency-based QNS faces two fundamental challenges.
First, the frequency comb needs a discrete sampling of a series of harmonic points in the frequency space. Using $\tilde{L}_{\omega}$ to represent how many points at which frequency is sampled between the frequency cutoffs, we can express the scaling of the sample complexity of the comb-based QNS as $N^{(k)}_{\tilde{\mathcal{C}}_{\text{comb}}}\sim \mathcal{O}(\tilde{L}_{\omega})$, where $k$ is the order of the expansion of the perturbation. 
The resource inefficiency, i.e., large sample complexity, of comb-based QNS becomes more pronounced for stronger non-Gaussianity since $\mathcal{O}(\tilde{L}_{\omega}^{k-1})$  grows significantly for large $k$, especially when $\tilde{L}_{\omega} \gg 1$ ($\tilde{L}_{\omega} \gg 1$ is essential to ensure accurate interpolation).  
Second, as we denote the comb control as $\tilde{\mathcal{C}}_{\text{comb}}$, such a control normally needs very long sequences of unequally spaced, instantaneous pulses. This typically lies outside the set of high-fidelity constrained control, denoted by $\mathcal{C}$.  We express such a control misalignment as $\tilde{\mathcal{C}}_{\text{comb}} \notin \mathcal{C}$. For example, for single-qubit gate fidelity $99.99\%$, employing a frequency comb with 100 pulses gives $10\%$ infidelity. 

\section{Resource-efficient control-adapted C\&C}
\label{sec:Frames}
In this section, we review the frame-based formalism and the control-adapted (CA) QNS to illustrate the resource efficiency of CA QNS in noise C\&C. 
We refer the reader to~\cite{Teerawat_PRXQ_Frame} for more in-depth mathematical details about the notion of frames in qubit dynamics and CA QNS. 

\subsection{Two useful structures in non-Gaussian open quantum dynamics}
Although the high-order convolution introduced in the Dyson series for non-Gaussian cases may seem so complicated to handle, two key structures can be identified to alleviate this challenge. In the following, we present two structures that serve as foundational elements for the two main pillars of our work: control-adapted QNS (discussed in this Section) and control-based symmetry analysis (discussed in Sec.~\ref{sec:Symmetry_analysis}). 

First, the time-ordered integral of all time-dependent elements in Dyson series identifies a \textit{dynamical integral} $\mathcal{I}^{(k)}_{[\vec{q}], \vec{u}}(T) = \int^T_0 d_>t_{[k]}  \big[ \prod_{j=1}^k y_{[q_j],u_j}(t_j)  \big] \expect{\mathcal{B}^{\vec{\mu}_{[k-1]}}_{ \vec{q}_{[k]}} (\vec{t})} $~\cite{Paz_PRL_2014,Teerawat_PRXQ_Frame}. The dynamical integral encapsulates the temporal ``convolution'' between control and noise, reflecting an inherent difficulty in addressing non-Markovian non-Gaussian $(k\geq2)$ open qubit dynamics. 

However, it is evident that the Pauli string, which plays a crucial role in $Q$'s dynamics, is absent from the dynamical integral structure described above. As the algebra of Pauli strings is essential in subsequent symmetry analysis, we thus define the full $Q$ relevant components in the temporal integrand in Eq.~\ref{eq:dyson}, which include both the switching function and the Pauli string, as the \textit{control tensor}.  It is expressed as 
\begin{widetext}
\begin{equation} 
\begin{aligned}
    \mathbf{T}^{(k)}_{\vec{q}; \vec{\mu}}(\vec{t}) &\equiv \lambda^{(k)}[\sum_{\vec{u}} \prod_{i=1}^k y_{[q_j],u_i}(t_i) \Lambda_{{u_i}}] = \sum_{l,\pi,\vec{u},\vec{c}}(-1)^{\bar{f}^{(k)}_{\pi}(\vec{\mu})} 
        \prod_{j=1}^l\prod_{j'=l+1}^k f^{c_j}_{u_{j}}y_{[q_j],c_j}(t_{\pi(j)})y_{[q_{j'}],u_{j'}}(t_{\pi(j')})\Lambda_{u_j}\Lambda_{u_{j'}} 
\label{eq:control_tensor}
\end{aligned}
\end{equation} 
\end{widetext}
where  $\lambda^{(k)}[...]$ is a control independent linear map as detailed above.  In the single-qubit case, the control tensor is denoted by $\mathbf{T}^{(k)}_{\vec{\mu}}(\vec{t}) $, omitting the index $q$. 
It is precisely the control tensor that convolutes and filters the noise correlators in qubit dynamics. 
Based on such a definition, we can recompose the Dyson term as 
$$\mathcal{D}^{(k)}_O(T)/k! = (-i)^k\sum_{\vec{q},\vec{\mu}} \int^T_{0} d_{>}\vec{t}_{[k]} \mathbf{T}^{(k)}_{\vec{q};\vec{\mu}}(\vec{t})  \mathcal{B}^{\vec{\mu}}_{ \vec{q}} (\vec{t}). $$
While the control tensor may reduce to a matrix of the dimension of $\Lambda$,  we emphasize that, for this study, it is treated as a tensor indexed by $\vec{q}, \vec{\mu},\vec{t}$. These three indices define a noise correlator.  For a fixed $\vec{q}$ and $\vec{\mu}$, $\mathbf{T}^{(k)}_{\vec{q}; \vec{\mu}}(\vec{t})$ is a $k$-dimensional tensor with each dimension $t$ supported continuously in $[0,T]$.  

\begin{figure}[!htbp]
    \centering
    \includegraphics[width =0.4 \textwidth]{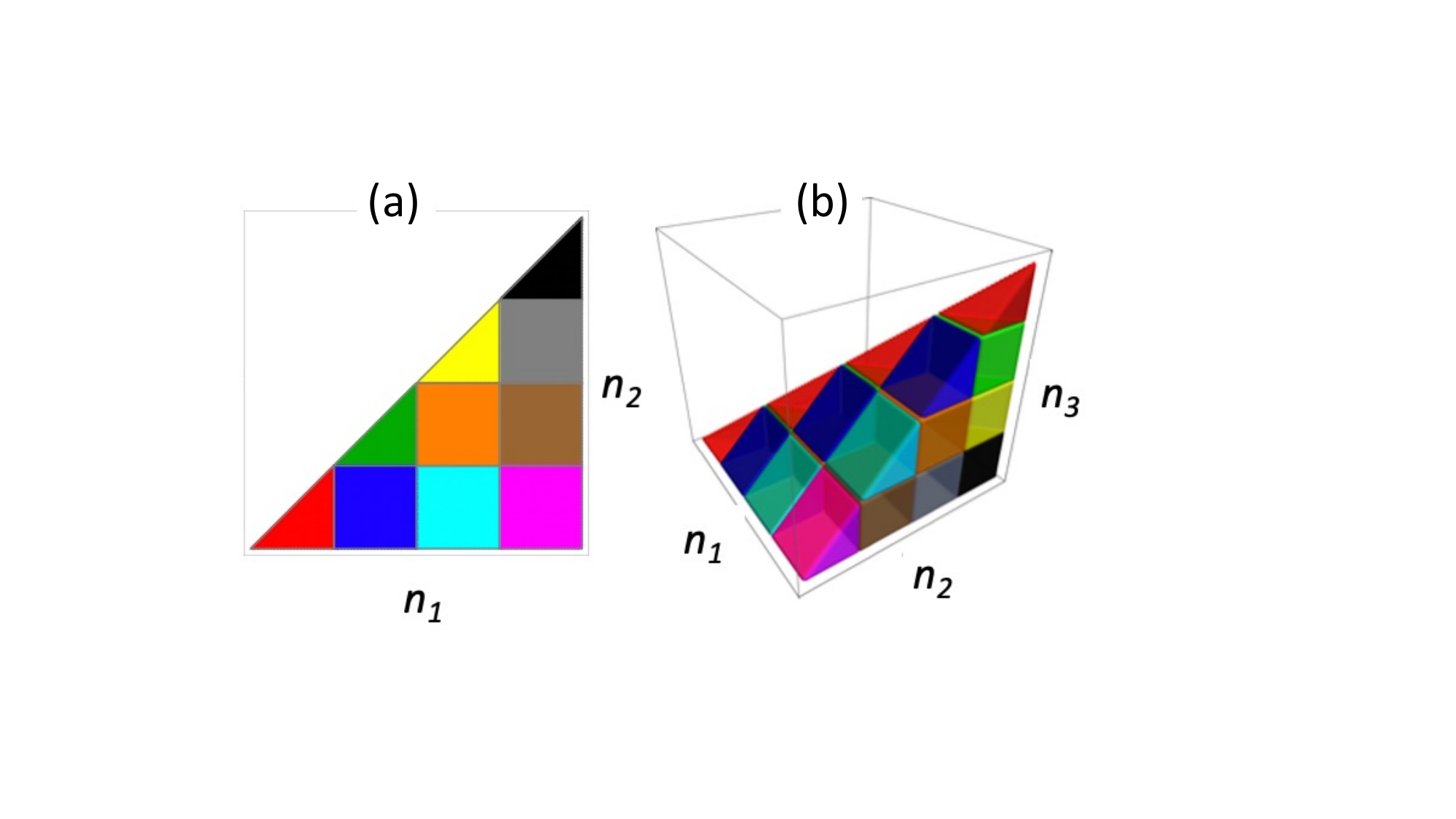}
    \caption{ \textbf{A sketch of digital CA spectra} 
     (a) Gaussian [$k=2$] CA digital spectra $\bar{S}(n_1,n_2)$ with four windows, where $L=4\geq n_1\geq n_2\geq1$.  
     (b)  Non-Gaussian [$k=3$] CA digital spectra with four windows where the noise is stationary. This feature is shown as translational invariance along diagonal direction $\bar{S}(n_1,n_2,n_3) = \bar{S}(n_1-m,n_2-m,n_3-m)$, where $L=4\geq n_1\geq \cdots n_L\geq1$. }
    \label{fig:CA_spectra_demo}
\end{figure}

\subsection{Control-adapted C\&C with frames}

Recall that in SP, Fourier functions are introduced to target the noise correlators. In the CA framework, basis or frame representations (a generalization of bases that include non-orthogonal functions) are defined in terms of the control variables rather than on the noise correlator. This focus on a control approach, as opposed to noise, is driven by the recognition that actual control operations are typically constrained by strong limitations. These constraints, which include restrictions on control waveform profiles and pulse timing, are generally well-understood and are described by some form of time-dependent functions.

We verbalize this control adaptation abstractly and then elaborate on it using detailed equations in the following paragraphs. To accommodate these constraints, we select appropriate frames (using the nomenclature from Ref~\cite{Teerawat_PRXQ_Frame}) that allow any controlled operation to be expressed using a finite and preferably minimal set of frame elements. That said, we can decompose the constrained time-dependent control (we use $\mathcal{C}$ hereafter to denote the control that respects realistic constraints and can be well implemented) as $$\text{control}(t) =\sum_n \text{coefficient}_n\times \text{frame function}_n(t),$$ where the frame functions match the constraints such that the upper bound of $n$ is small. Notice that the coefficients are time independent and encode customizable control parameters (denoted $P$ hereafter) that satisfy control constraints. For example, $P$ can include the control amplitude in $[0,2\pi]$ or the phase factor in $[0,\pi]$ along the $x-y$ plane.
Consequently, the constrained control $\mathcal{C}$ is transferred to the noise correlator as a temporal integral. 
This process implies that only certain noise components, determined by the control frames, are relevant to the restricted qubit dynamics. These are the noise components that should be identified and characterized through QNS. These components of noise that influence the qubit dynamics are compactly represented as ``control-adapted spectra,'' which are easier to analyze using fewer, or even minimal, control resources.

In this paper, we illustrate these concepts by \textit{only} focusing on digital controls where the frames correspond to an orthogonal window (or Walsh) basis. 

\paragraph{Digital Control}
We consider a scenario in which qubit control comprises a sequence of $L$ instantaneous, equidistant, perfect gates over a time interval $[0, T]$. This setup assumes that the inter-pulse time delay significantly exceeds the gate execution time.  As progress in reducing the two-qubit gate time (e.g., $40$ ns for two-qubit gate versus $20$ns for single-qubit gate in~\cite{Google_surface_code_Nature}), instantaneous control assumption would be a reasonable assumption in some, if not all, part of real-device quantum circuits. 
This assumption also indicates that the control unitary and noise impact are separated in the noisy quantum channel, which is widely presumed in many channel-based error characterization and mitigation protocols \cite{White2020, SamsonWang, Harper_NatPhy}. Under our assumption, we have $\mathcal{C}_{\text{W}}\in \mathcal{C}$.

\paragraph{Frame-based Filter Function Formalism}
In digital control, the switching function $y_{[q], u}(t)$ is a piecewise constant. The constraints and symmetries underlying digital control help identify the window function as the appropriate frames (essentially a basis set) for $Q$. For digital window frames $\mathcal{C}_{\text{W}} \equiv \{W_n(t)\}$ and $\mathcal{C}_{\text{W}} \in \mathcal{C}$, we define:
\begin{equation}
    W_n(t) = 
    \begin{cases}
        1, & \text{if } (n-1)\tau \leq t < n\tau,\\
        0, & \text{otherwise},
    \end{cases}
\end{equation}
where $\tau \equiv T/L$ and $n = \{1, \dots, L\}$. The switching function can be expressed using these window frames as $ y_{[q],u}(t) = \sum_{n=1}^L F^{(1)}_{[q],u}(n) W_n(t)$, where the expansion coefficient
\begin{equation}
    F^{(1)}_{[q],u}(n) = \frac{1}{\tau} \int_0^T ds \, y_{[q],u}(s) W_n(s)
\end{equation}
is known as the the \textit{frame-based filter function}. The way the frame-based filter function is defined guarantees that all allowed tunable degrees of freedom $P$, e.g., the amplitude or phase of each pulse, naturally respect constraints $\mathcal{C}_{\text{W}}$ because they are only encoded in $F^{(1)}_{[q],u}(n)$ rather than in $ W_n(t)$. 

This approach prioritizes constrained control and significantly diverges from Fourier-based SP methods; in the latter case, control is associated with noise-oriented bases or frames, which are potentially misaligned with actual control constraints. 
In the frame expansion approach, the noise correlator is passively associated with the control frames in the temporal convolution, resulting in the so-called CA spectra, which are easier to analyze than their complete counterparts. These CA spectra are defined as:
\begin{equation}
\begin{aligned}
   \bar{S}^{\vec{\mu}_{[k-1]}}_{\vec{q}}(\vec{n})& \equiv \int^T_0 d_>\vec{t}_{[k]} \expect{\mathcal{B}^{\vec{\mu}_{[k-1]}}_{\vec{q}} (\vec{t})}   \prod\limits^k_{j=1} W_{n_j}(t_j), \\
\end{aligned}
\end{equation}
where the $n$-string $(n_1,\cdots,n_k)$ satisfies $1\leq n_k \leq \cdots n_1 \leq L$. Any $n$-string that violates the time ordering results in a vanishing CA spectrum.   By definition, the full knowledge of the noise is temporarily convoluted with the control frames that are ``back-actioned'' onto it.  We refer the reader to a geometric illustration of the digital CA spectra for $L=4$ in Fig.~\ref{fig:CA_spectra_demo}, and note that only these ``coarse-grained'' spectra contribute to digitally controlled qubit dynamics.  The CA spectra are involved in the dynamical integral as 
\begin{equation}
    \mathcal{I}^{(k)}_{{[q],\vec{u}}}(T)|_{\text{CA}}=\sum\limits_{\vec{n}}\prod\limits_{j=1}^k F^{(1)}_{[q_j], u_j}(n_j)\bar{S}^{\vec{\mu}_{[k-1]}}_{\vec{q}}(\vec{n}).
\end{equation}
The CA spectra represent a ``model-reduced'' description of the noise, and in fact, the model reduction allows one to design QNS protocols for spectral reconstruction more easily than protocols in SP. 

For $|Q|>1$, $\bar{S}^{\mu}_{q,q}(n_1,n_2)$ represents Gaussian \textit{CA self-spectra} while $\bar{S}^{\mu}_{q,\bar{q}}(n_1,n_2)$ for $q\neq \bar{q}$ represents  Gaussian \textit{CA cross-spectra}, the latter representing a truly model-reduced depiction of spatio-temporally correlated noise.  As non-Gaussianity emerges, the complexity escalates, making computational efforts untenable without model reduction. Model reduction significantly reduces the sample complexity compared to traditional harmonic sampling in comb-based QNS, thus establishing the resource efficiency of the C\&C method with finite $L$. On the one hand, the number of single-qubit CA spectra at $k$-th order scales $N^{(k)}_{\mathcal{C}_{\text{W}}}(L)=\mathcal{O}(L^k 2^{k-1})$ [and $\mathcal{O}(L^k)$ for solely classical noise], which corresponds to the QNS {sample complexity}. For classical noise, the resource advantage of CA QNS over SP QNS can be roughly expressed by 
\begin{equation}
    N^{(k)}_{\mathcal{C}_{\text{W}}}(L)\big|_{\text{CA}} \bigg/ N^{(k)}_{\tilde{\mathcal{C}}_{\text{comb}}}(\tilde{L}_{\omega})\big|_{\text{SP}} = (L/\tilde{L}_{\omega})^k.
\end{equation}
When $L<\tilde{L}_{\omega}$, C$\&$C lead to significant model reduction, with higher $k$ yielding more resource efficiency~\cite{DongWZ_APL}. 
Additionally, $\prod_{j=1}^k F^{(1)}_{[q_j]; u_j}(n_j)$  is a product form and inherently incorporates control constraints, which allows one to design the appropriate control details $P$ by properly choosing $F^{(1)}(n)$, simplifying  C$\&$C design.

\paragraph{Dyson series in digital frames} 
Following a symbolic structure similar to the above CA dynamical integral, one can obtain a similar form for the CA Dyson expression as:
\begin{equation}
\begin{aligned}
    \mathcal{D}_{O}^{(k)}(T)/k! |_{\text{CA}} &= (-i)^k\sum\limits_{\vec{n},\vec{\mu},\vec{q}} \lambda^{(k)}\big[\sum_{\vec{u}} \bm{F}_{[\vec{q}], \vec{u}}(\vec{n})\Lambda_{u_1}...\Lambda_{u_k} \big] \\
    & \quad \quad \quad \times \bar{S}^{\vec{\mu}_{[k-1]}}_{\vec{q}}(\vec{n})\\
    &=(-i)^k\sum_{\vec{n},\vec{q},\vec{\mu}}  \mathbf{T}^{(k)}_{\vec{q};\vec{\mu}}(\vec{n})  \bar{{S}}^{\vec{\mu}_{[k-1]}}_{ \vec{q}} (\vec{n}),
    \end{aligned}
\end{equation}
where $\bm{F}_{[\vec{q}], \vec{u}}(\vec{u})\equiv \prod^k_{j=1}F^{(1)}_{[q_j],u_j}(n_j)$. It is evident that the frame transformation applies exclusively to $y(t)$ within the control tensor, as this is the only time-dependent component. The CA representation of the Dyson series also allows us to express the observable $\mathbb{E}[O(T)]_{\rho_Q\otimes \rho_E}$ in the CA picture for all $O$ and $\rho_Q$, which is detailed in the Appendix.~\ref{app:ugly_D1_D4}.

\paragraph{Control-adapted QNS} The CA QNS protocols work on the level of $F^{(1)}_{[q],u}(n)$ to reconstruct CA spectra.  They follow the conventional steps of other QNS methods: configuring various control setups (including pulse parameters,  measurement bases, and initial states) and using the measured expectation values of observables $\mathbb{E}[O(T)]$ to deduce the CA spectra. By establishing the appropriate CA QNS controls, the CA spectra are subsequently derived by solving a linear system. 

We note that for non-Gaussian dynamics, the noise spectra exhibit higher nonlinearity, making designing a QNS protocol much harder than for non-Gaussian noise. This complication is discussed in the Appendix~\ref{app:qns_ctrl_design}.

\paragraph{Control-adapted optimal control design}
Once the control-adapted spectra are identified via CA QNS, this information helps to design appropriate frame-based filter functions tailored to specific objectives.
For example, one can devise an optimal noise-tailored control $P_{\text{opt}}$ by optimizing $F^{(1)}_{[q],u}(n)$ which is constrained within the persistent control symmetries $\mathcal{C}_{\text{W}}$. 

\section{Digital Symmetry analysis}
\label{sec:Symmetry_analysis}

As demonstrated in the preceding section, control adaptation enables resource-efficient C\&C of the CA spectra. In this section, we extend the concept of control adaptation by introducing a symmetry analysis on digital control, which in turn induces symmetry in the CA spectra. This induced symmetry ultimately makes C\&C even more resource-efficient, and eventually changes the CA digital QNS sample complexity.

Symmetry analysis traditionally applies to noise spectra in frequency-based QNS, especially to simplify the spectral sampling of high-dimensional non-Gaussian noise~\cite{Chandran_1994Jan}. In the standard picture (SP) Fourier QNS, all classical frequency-based polyspectra can be inferred from a limited region known as the \textit{principal domain}, thereby reducing sampling overhead. This reduction stems from the inherent properties of classical stationary noise, which allow spectral symmetries to minimize sampling density. However, while the principal domain concept reduces sampling overhead, it typically does not alter the scaling of sample complexity.

In contrast,  CA spectra are organized in time-ordered structures that are more compact and less redundant than SP spectra. In this sense, CA spectra can be regarded as existing within a principal domain for spectra in the CA framework. Hence, conventional spectra symmetry analysis is no longer applicable to CA spectra.
The forthcoming symmetry analysis is disparate as it is derived directly from the digital control tensor. This control-based symmetry translates into symmetries within the CA spectra,  fundamentally rescaling the complexity of the learning process in digital CA QNS.

\subsection{Digital control tensor contractions}

The control tensors of $Q$ under digital control $\mathcal{C}_{\text{W}}$ contract to lower-dimensional tensors whenever certain rules in the $n$-strings are satisfied. This contraction is universal $\forall P$ within fixed $\mathcal{C}_{\text{W}}$. We illustrate these rules for contraction in the following propositions. 
\begin{prop} 
In single-qubit dynamics, the window-framed $k$ dimensional control tensor $\bm{T}^{(k)}_{\vec{\mu}}(\vec{n}) =\lambda^{(k)}[\sum_{\vec{u}} \bm{F}_{\vec{u}}(\vec{n})\sigma_{u_1}\sigma_{u_2}...\sigma_{u_k}]$ is contractible to a $(k-2)$ dimensional control tensor, whenever a ``3-streak'' exists in the $n$-string ($n_{i-1}=n_i=n_{i+1}$ in $\vec{n}=(n_1,\cdots,n_{i-1},n_i,n_{i+1},\cdots,n_k)$).
\label{prop-1}
\end{prop}
\begin{proof}
Define the effective error Hamiltonian in Eq.~\ref{eq:eq_1} as $\tilde{H}(t) =\tilde{h}(t) \otimes B(t)$, then in the Dyson series of the error propagator,  one has
\begin{widetext}
\begin{equation}
\begin{aligned}
      \mathcal{D}^{(k)}_O(T)/k! =(-i)^k&\sum\limits^{k}_{l=0}\sum\limits_{\pi\in \Pi_{l;k}} \int^T_0d_>\vec{t}_{[k]} 
      \Big[  \prod_{j=1}^l \bar{h}(t_{\pi(j)}) \prod_{j'=l+1}^k \tilde{h}(t_{\pi(j')})  \Big] \times   \sum_{\vec{\mu}}  (-1)^{\bar{f}^{(k)}_{\pi}(\vec{\mu})}  \expect{\mathcal{B}^{\vec{\mu}} (\vec{t}) }.
\end{aligned}
\end{equation}
\end{widetext}
A key point to note is that times are always ordered (reversely-ordered) in the control string $\prod_{j'} \tilde{h}(t_{\pi(j')})$ ($\prod_j \bar{h}(t_{\pi(j)})$); see Appendix.~\ref{app:dyson}. Because the control is digital,  the above expression in the window frames becomes:
\begin{widetext}
    \begin{equation}
\begin{aligned}
      \mathcal{D}^{(k)}_O(T)/k!|_{\text{CA}}& =(-i)^k \sum_{\vec{n}}\sum\limits_{l; ~\pi} 
      \Big[  \prod_{j=1}^l \sum\limits_{u_j} F^{(1)}_{u_j}(n_{\pi(j)})\bar{\sigma}_{u_j} \Big]    \Big[    \prod_{j'=l+1}^{k} \sum\limits_{u_j'} F^{(1)}_{u_j'}(n_{\pi(j')}) \sigma_{u_j'} \Big] \times   \sum_{\vec{\mu}}  (-1)^{\bar{f}^{(k)}_{\pi}(\vec{\mu})}  \bar{S}^{\vec{\mu}}(\vec{n}) \\
      & =(-i)^k \sum\limits_{\vec{n},\vec{\mu}} \lambda^{(k)}\big[\sum_{\vec{u}} \bm{F}_{ \vec{u}}(\vec{n})\sigma_{u_1}...\sigma_{u_k} \big] \bar{S}^{\vec{\mu}}(\vec{n}) =(-i)^k\sum_{\vec{n}, \vec{\mu}}  \mathbf{T}^{(k)}_{ \vec{\mu}}(\vec{n})  \bar{{S}}^{\vec{\mu}} (\vec{n}),
  \end{aligned}
\end{equation}
\end{widetext}
where $\bar{\sigma}_u = -\tilde{O}^{-1}(T)\sigma_{u}\tilde{O}(T) $.  A similar notice to give is that the $n$-strings are always ordered (reversely-ordered) in string $\prod_{j'} F^{(1)}_{u_{j'}}(n_{\pi(j')})$ ($\prod_j F^{(1)}_{u_j}(n_{\pi(j)})$). 
For a 3-streak ($n_{i-1}=n_i=n_{i+1}=\hat{n}$) with $L\geq n_1 \geq ..n_{i-1} \geq n_i \geq n_{i+1}..\geq 1$,   at least two of these three  must appear as neighbors in either $\tilde{h}$ or $\bar{h}$ string, which, in these case, gives 
\begin{widetext}
    \begin{equation}
    \begin{aligned}
 &[\sum_u F^{(1)}_u  (\hat{n})\sigma_u ]   [\sum_v F^{(1)}_v  (\hat{n})\sigma_v ] =\Big[\sum_u \frac{1}{\tau} \int^T_0dsy_u(s)W_{\hat{n}}(s)\sigma_u \Big]^2 = \Big[ U_0^{\dagger}(\hat{n}\tau)\sigma_z U_0(\hat{n}\tau) \Big]^2\Big|_{(\hat{n}-1)\tau\leq s<\hat{n}\tau}= \mathbb{I},\\
 &  [\sum_u F^{(1)}_u  (\hat{n})\bar{\sigma}_u ]   [\sum_v F^{(1)}_v  (\hat{n}) \bar{\sigma}_v ] 
  =-\tilde{O}^{-1}(T)\Big[ U_0^{\dagger}(\hat{n}\tau)\sigma_z U_0(\hat{n}\tau) \Big]^2 \tilde{O}(T)\Big|_{(\hat{n}-1)\tau\leq s<\hat{n}\tau} = \mathbb{I}.
    \end{aligned}
\end{equation}
\end{widetext}
This pattern shows that such a high-dimensional control tensor factors out a trivial identity. The remainder is easy to verify as a $(k-2)$-dimensional control tensor, up to a factor:
\begin{equation}
        \bm{T}^{(k)}_{\vec{\mu}}(\vec{n}) =  c_{\vec{\mu}}\bm{T}^{(k-2)}_{\vec{ \mu}'}(\vec{n}'),
    \label{eq:contraction}
\end{equation}
where  trivial byproduct factor $c \neq 0$,  string $\vec{n}'$ is the remaining string by deleting $n_{i-1},n_i$ in $\vec{n}$, and $\vec{\mu}'$ is the remaining string by deleting $(i-1)$-th and $i$-th element in $\vec{\mu}$. This operation is also vividly shown in Fig.~\ref{fig:symmetry_demo}(a). The exact value of $c$ will be clarified later.
Insomuch, the 3-streak spectrum $\bar{S}^{\vec{\mu}_{[k-1]}}(n_1,\cdots,n_{i-1} = n_i =n_{i+1},\cdots,n_k)$ is \textit{filtered} by the same control tensor as a  $\bar{S}^{\vec{\mu}_{[k-3]}}$ does, regardless of the value of $F^{(1)}_u(n)\in P$ , $\forall n\leq L$.  
\end{proof}

\begin{prop}
In the case of two qubits, the window-framed $k$-dimensional two-qubit control tensor is defined as $\bm{T}^{(k)}_{\vec{q},\vec{\mu}}(\vec{n})= \lambda^{(k)}[ \sum_{\vec{u}} \bm{F}_{[\vec{q}],\vec{u}}(\vec{n})\Lambda_{u_1}...\Lambda_{u_k}]$, where $\Lambda$ denotes a two-Pauli operator and $q\in\{A,B\}$. Whenever a ``5-streak'' exists in the $n$-string ($n_{i-2}=n_{i-1}=n_i=n_{i+1}=n_{i+2}$), such a control tensor contracts to a lower-dimensional control tensor.
\label{prop-2}
\end{prop}
\begin{proof}
The proof here is based on Proposition 1. Denote $\tilde{H}(t) = \sum_{q=\{A,B\}} \tilde{h}_{q}(t) \otimes B_q(t)$. When projected into a window frame, the Dyson series becomes
\begin{widetext}
    \begin{equation}
\begin{aligned}
 \mathcal{D}^{(k)}_O(T)/k!|_{\text{CA}} =(-i)^k \sum_{\vec{n},\vec{q}}\sum\limits^{k}_{l=0}\sum\limits_{\pi\in \Pi_{l;k}} 
       &\Big[  \prod_{j=1}^l \sum\limits_{u_j} F^{(1)}_{q_j,u_j}(n_{\pi(j)})\bar{\Lambda}_{u_j} \Big]       \Big[    \prod_{j'=l+1}^{k} \sum\limits_{u_j'} F^{(1)}_{q_{j'},u_{j'}}(n_{\pi(j')}) \Lambda_{u_{j'}} \Big]  \times   \sum_{\vec{\mu}}  (-1)^{\bar{f}^{(k)}_{\pi}(\vec{\mu})}  \bar{S}^{\vec{\mu}}_{ \vec{q}_{[k]}} (\vec{n}), 
       \end{aligned}
\end{equation}
\end{widetext}
where $\bar{\Lambda}_u =-\tilde{O}^{-1}(T)\Lambda_{u}\tilde{O}(T)$ and the two $[\dots]$ in the above expression correspond to the $\bar{h}_{q}(n)$-string and $\tilde{h}_{q}(n)$-string. 
If there is a 3-streak in the window and in the $q$-string at the same indices -- $(n_{i-1} \equiv n_i \equiv n_{i+1}) \wedge (q_{i-1} \equiv q_i \equiv q_{i+1})$ -- then the problem simplifies to the single-qubit case. The nontrivial case is when $q_{i-1}\equiv q_i \equiv q_{i+1}$ is false. Let us focus on the $\tilde{h}_{q}(n)$-string in the second $[\cdots]$.  Now we relate  $\tilde{h}_{q}(n) = \sum_u F^{(1)}_{q,u}(n){\Lambda}_u$ to   $\tilde{h}_{q}(n) = U^{\dagger}_0 (n\tau) \sigma_z^{[q]} U_0(n\tau)$, where $U_0(n\tau)$ is the control propagator at time $t=n\tau$. For a 2-streak in the window but no streak in $q$, we have $\tilde{h}_A(n)\tilde{h}_B(n) =\tilde{h}_B(n)\tilde{h}_A(n)= U_0^{\dagger}(n\tau) \sigma_z^{[A]} \sigma_z^{[B]}  U_0(n\tau)$. However, when this product is associated with a third $\tilde{h}_{q_3}(n)$ that expands the window streak from 2 to 3, regardless of the value of $q_3$, we have $\tilde{h}_A(n)\tilde{h}_B(n)\tilde{h}_{q_3}(n) =  U_0^{\dagger}(n\tau) \sigma_z^{[\bar{q_3}]}  U_0(n\tau)$ (where $q_3 \neq \bar{q_3} \in \{A,B\}$), which can be verified to form a lower-order control tensor.  Thus, in the two-qubit case, a 5-streak in windows is required to ensure that at least a 3-streak exists in either the bar or tilde part to form a tensor contraction. In such a case, the contraction reads $\bm{T}^{(k)}_{\vec{q};\vec{\mu}}(\vec{n}) =  c_{\vec{q},\vec{\mu}}\bm{T}^{(k-2)}_{\vec{q}';\vec{ \mu}'}(\vec{n}')$.
\end{proof}

\begin{figure} [h]
    \centering
    \includegraphics[width =0.5 \textwidth]{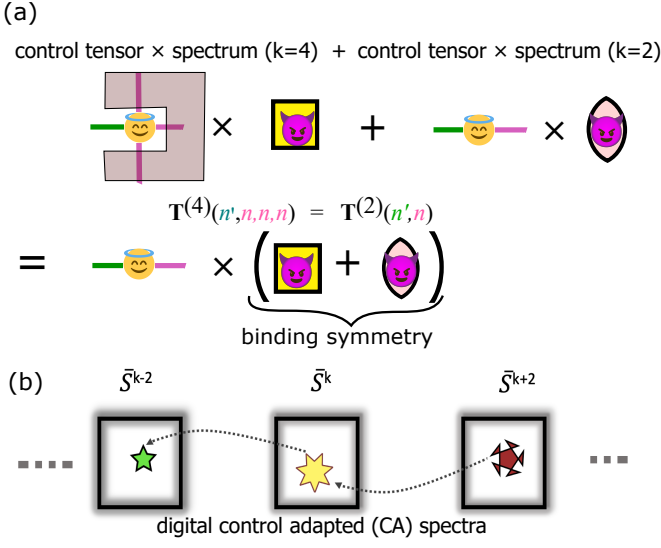}
    \caption{\textbf{Control tensor contraction induced spectral binding symmetry} (a) A control tensor $\bm{T}^{(4)}$, whenever it has a 3-streak (3 same colored legs in the shades) in the window $n$-string, is contractible to a tensor $\bm{T}^{(2)}$, giving $\bm{T}^{(4)}(n',n,n,n)= \bm{T}^{(2)}(n',n)$ [we set $c=1$ in Eq.~\ref{eq:contraction} for simplicity].  Such symmetry is back-actioned on the corresponding CA spectrum, making it bound to a lower-dimensional spectrum that pairs $\bm{T}^{(2)}$. (b) Because the binding symmetry across multiple $k$-orders, the $k$-th order spectra in binding symmetry are bound with $(k-2)$-th order spectra. Such a binding symmetry persists to lowest-order spectra until their associated control tensors are not further contractible.   }
    \label{fig:symmetry_demo} 
\end{figure}

Before proceeding, we note the following. First, it is possible that the contracted lower-order control tensor still has a 3-streak and is further contractible. The contraction persists until a lowest-order tensor is no longer contractible.  This forms a contraction chain, including multiple control tensors of different dimensions.
Second, $c_{\vec{q};\vec{\mu};\vec{n}} =0$ occurs whenever a full cancellation under linear combination map $\lambda$ happens. This corresponds to the case where the control tensor fully vanishes, rather than the contraction discussed above. 
Since the details of each control tensor depend on $(-1)^{\bar{f}^{(k)}_{\pi}(\vec{\mu})}\in \{-1,1\}$, it is possible that $ \lambda^{(k)}[\sum_{\vec{u}} \prod_i y_{[q_j],u_i}(t_i) \Lambda_{{u_i}}]\rightarrow 0$ even though $\sum_{\vec{u}} \prod_i y_{[q_j],u_i}(t_i) \Lambda_{{u_i}}$ is non-zero.

\subsection{Symmetries in digital-CA noise spectra}
Not only are the frames tied to the control passed onto the noise correlator to form CA spectra, but so are the symmetries underlying the contraction of the digital control tensor that induce symmetries among CA spectra. We now discuss such symmetries in single-qubit and two-qubit CA digital spectra.

\paragraph{Binding symmetry}
A crucial symmetry that fundamentally influences QNS sample complexity is \textit{binding symmetry}. Binding symmetry results from the qubit control's \textit{backaction} on $E$, where the contractibility of qubit control translates into constraints on certain spectra in qubit dynamics. Referring back to the 3-streak in Proposition~\ref{prop-1}, if a high-order control tensor contracts to a low-order control tensor ($c\neq0$), then in qubit dynamics, the corresponding high-order CA spectra are effectively ``filtered'' by the same control as the corresponding low-order CA spectra. As such, these CA spectra at different orders always ``bind'' together in qubit dynamics, entering through specific combinations, irrespective of the details ($P$) of the restricted control. This control-induced binding symmetry is illustrated in Fig.~\ref{fig:symmetry_demo}(a). This symmetry prevents certain CA spectra from being distinguished from others regardless of $P$ under persistent $\mathcal{C}_{\text{W}}$.  Thus, one can only infer their ``bound'' form, rather than each spectrum individually. Importantly, this suffices to characterize the digital qubit dynamics. 

For zero-mean non-Gaussian noise truncated at $K=4$,  binding symmetry exists between certain 2nd and 4th-order CA spectra, which is summarized as: 
\begin{widetext}
    \begin{equation}
\begin{aligned}
 \bm{T}^{(2)}_{\nu}(n,m) & \times [\bar{S}^{(\nu)}({n,m})+\sum_{\mu_1,\mu_2=0}^1 c_{\mu_1,\mu_2} \bar{S}^{(\mu_1,\mu_2,\nu)}({n,n,n,m})+\sum_{\mu_3,\mu_4=0}^1 c_{\mu_3,\mu_4}\bar{S}^{(\nu,\mu_3,\mu_4)}({n,m,m,m})],\quad \forall   \nu \in\{0,1\}; \\
 \bm{T}^{(2)}_{(q,q');\nu}(n,m) &\times [\bar{S}^{(\nu)}_{(q,q')}({n,m})+ \sum_{\vec{\mu}}c_{\vec{q};\vec{\mu}} \bar{S}^{(\vec{\mu},\nu)}_{(\vec{q},q')}({n,n,n,n,n,m})+\sum_{\vec{\mu}}c_{\vec{q'};\vec{\mu}}\bar{S}^{(\nu,\vec{\mu})}_{(q,\vec{q'})}({n,m,m,m,m,m})],\quad \forall \nu\in\{0,1\},~\forall q,q'\in Q;
  \end{aligned}
\end{equation}
\end{widetext}
where the first (second) line corresponds to the single-(two-) qubit case. 
We emphasize that the above equation is specifically formulated to indicate that the spectra in the brackets are associated with the same control tensor, up to control-independent parameters, to be incorporated into the qubit dynamics (e.g., in the Dyson series). Moreover, as suggested in the above equations, each high-order spectrum independently binds the low-order spectrum. 
As we can only characterize the bound form of some indistinguishable spectra, in practice, such a binding symmetry implies that in the spectral estimation algorithm one can assume that $\bar{S}^{(\nu)}({n,m})$ is the learnable quantity.  The others are assumed as $\bar{S}^{(\mu_1,\mu_2,\nu)}({n,n,n,m}) = 0 =\bar{S}^{(\nu,\mu_3,\mu_4)}({n,m,m,m})$ persistently, $\forall \nu_j\in\{0,1\},\forall j$, although any of them could be non-vanishing. Thus, the actual value of $c$ does not matter.   This binding symmetry  on single CA spectrrum is defined as
    \begin{equation}
\begin{aligned}
\bar{S}^{(\nu)}(n,m) & \Leftarrow \bar{S}^{(\nu,\mu_3,\mu_4)}(n,m,m,m),\\
\bar{S}^{(\nu)}(n,m)  & \Leftarrow \bar{S}^{(\mu_1,\mu_2,\nu)}(n,n,n,m); \\
& 1\leq n,m \leq L, ~ \forall\mu,\nu\in\{0,1\};
\end{aligned}
\end{equation}
where spectra on the right of $\Leftarrow$ are assumed zero. A similar approach applies when using the equations for prediction. It is interesting to point out that binding symmetry needs non-Gaussian spectra to be involved, a special feature different from purely Gaussian digital CA QNS.  As shown in Fig.~\ref{fig:symmetry_demo}(b),  the binding symmetries may extend across multiple $k$ values due to the chain contraction of control tensors. 

\paragraph{Dark spectra}
Some control tensors always vanish, ensuring that their associated spectra do not contribute to qubit dynamics. We dub these spectra with zero control tensors as ``dark spectra''.  Notice that the disappearance of the control tensor depends only on $\vec{n}$ and $\vec{\mu}$,  regardless of the control parameters $P$,
\begin{equation} 
\bm{T}^{(k)}_{\vec{q},\vec{\mu}}(\vec{n}) \equiv 0,  \quad  \exists \vec{n},\vec{\mu},~~\forall P.
\end{equation}
The vanishing of this control tensor occurs due to the complete cancellation under the linear combination process ($\lambda$), which only occurs for $\vec{\mu} \neq \vec{0}$ for some $\vec{n}$. Classical spectra are never dark because of $(-1)^{\bar{f}}\equiv 1$, whereas cancellation requires some $-1$ coefficients that only present in quantum spectra. Although representing a general rule for full cancellation is inherently complex, the leading orders of dark spectra are as follows: On the Gaussian level, the dark single-qubit CA spectra are $S^{(1)}(n,n)$. At the non-Gaussian level for $k=3$, the dark spectra are 
\begin{equation}
    S^{(0,1)}(n_1,n_2=n_3),~  S^{(1,0)}(n_1,n_2=n_3),~  S^{(1,1)}(n_1=n_2,n_3);
\end{equation}
and for $k=4$, the dark spectra are 
\begin{widetext}
    \begin{equation}
\begin{aligned}
    &S^{(0,0,1)}(n_1,n_2,n_3=n_4), S^{(0,1,0)}(n_1,n_2,n_3=n_4) ,S^{(0,1,0)}(n_1,n_2=n_3,n_4),S^{(0,1,1)}(n_1,n_2=n_3,n_4), \\
&S^{(1,0,0)}(n_1,n_2=n_3,n_4), S^{(1,0,0)}(n_1=n_2,n_3,n_4), S^{(1,0,1)}(n_1,n_2=n_3,n_4) ,S^{(1,0,1)}(n_1=n_2,n_3,n_4), \\
&S^{(1,1,0)}(n_1=n_2,n_3,n_4), S^{(1,1,0)}(n_1,n_2,n_3=n_4) ,S^{(1,1,1)}(n_1=n_2,n_3,n_4).
\end{aligned}
\end{equation}
\end{widetext}
Two-qubit dark spectra can be defined in a manner similar to that of single-qubit dark spectra described above.

\paragraph{Swap Symmetry} In the two- and multi-qubit cases, an additional symmetry---swap symmetry---emerges. For classical cross-spectra containing a 2-streak in the window $n$-string, the $q$ values at the same location can take two different configurations, eg., $(q,\bar{q})$ and $(\bar{q},q)$ with $q\neq \bar{q}$, that gives two different spectra. These two spectra, transformed to each other by swapping $(q,\bar{q})$, are filtered by the same control tensor.  Such symmetry of classical noise can be expressed as:
\begin{widetext}
\begin{equation}
\begin{aligned}
\bm{T}^{(k)}_{\vec{q},\vec{0}}(\vec{n})\times [\bar{S}^{(\vec{0})}_{..., q,\bar{q},...}(..., n,n,...) + \bar{S}^{(\vec{0})}_{..., \bar{q},q,...}(..., n,n,...)]  ~~~\forall q\neq \bar{q} \in Q\quad\quad \forall P;
 \end{aligned}
\end{equation}
\end{widetext}
where the position of $(q, \bar{q})$ in the $q$-string matches the position of $(n, n)$ in the $n$-string, and details in $\cdots$ of these two spectra are correspondingly identical.  
In fact, under swap symmetry,  spectra $\bar{S}^{(\vec{0})}_{..., q,\bar{q},...}(..., n,n,...)$ and $\bar{S}^{(\vec{0})}_{..., \bar{q},q,...}(..., n,n,...)$ are also bound together. The CA QNS can only learn their bound form, similar to the standard binding symmetry.
Note that swap symmetry applies to 2-streaks and higher, specifically for same-order ($k$) spectra. In contrast, binding symmetry necessitates a minimum of a 3-streak and encompasses different orders.

Though swap symmetry is illustrated above for classical cross-spectra, it also exists in quantum cross-spectra. A more rigorous discussion of this point can be found in Appendix~\ref{app:swap_symmetry}. In the following section, we only showcase two-qubit classical CA spectral estimation; thus, the classical expression provided above will suffice for the readers.

As we have explained all essential symmetries in CA QNS, interested readers are encouraged to consult Appendix~\ref{app:incremental_QNS} for a case study of non-Gaussian CA QNS for reconstructing single-qubit classical and quantum spectra (binding symmetry and dark spectra are used therein).

To conclude this section, we emphasize that the symmetries of the CA spectra are rooted in the digital control. When this control symmetry is removed, the CA spectral symmetries collapse, unlike the intrinsic properties of noise distributions, such as stationarity or zero-mean.

\section{Fundamental digital QNS: Saturation order and QNS sample complexity }
\label{sec:Fundamental_QNS}
Our prior analysis indicates that as $k$ grows, the resource efficiency advantage of CA QNS relative to SP QNS becomes more evident. However, it is clear that the sample complexity of CA QNS, denoted by $N^{(K)}_{\mathcal{C}_{\text{W}}}$, still exhibits an exponential scaling with $K$. As non-Gaussian effects become more significant --- such as when $Q$ has ultra-strong coupling with $E$ --- the necessary truncation order $K$ tends toward infinity, making the qubit dynamics unsolvable as a typical non-perturbative problem.
In such scenarios, the number of spectra that must be learned increases enormously and $N^{(\infty)}_{\mathcal{C}_{\text{W}}}$ becomes unbounded, leading to a breakdown of the CA QNS. Fortunately, due to the control-induced symmetry upon the spectra, though the non-vanishing number of spectra still balloons, the $N^{(\infty)}_{\mathcal{C}_{\text{W}}}$ is bounded through some mechanisms, which saves us from non-perturbative complexity.

In this section, we demonstrate that the complexity of digitally-controlled qubits depends solely on the size $|Q|L$ of the digital circuit, rather than the complexity of the environment (i.e., non-Gaussianity in perturbation theory). This implies that with finite control (finite $L$), a finite $N^{(\infty)}_{\mathcal{C}_{\text{W}}}$ is sufficient to manage an indefinitely complex noise environment ($k\rightarrow\infty$). We denote this complexity \textit{ the fundamental complexity of the digital control}.
This assertion stems naturally from the universal binding symmetry discussed earlier.

\begin{figure*}[!htbp]
\centering
\includegraphics[width=0.85\textwidth]{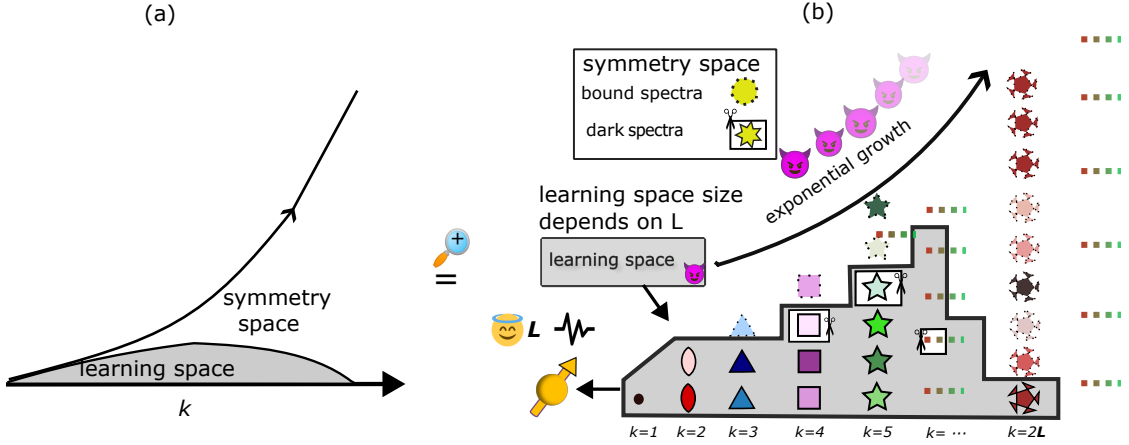}
\caption{\textbf{Fundamental digital sample complexity of digital CA QNS.}  (a) Both the total number of distinct noise spectra (curved arrow), and the symmetry space of spectra (white margin) balloon against $k$, while their subtraction ---learning space (gray) --- is bounded. 
(b) Further explanation of the left diagram. The control-induced symmetries in CA spectra and time ordering give a substantial reduction of sampling complexity: the number of bound spectra grows exponentially (as shown in the dash-edged simplex on a white background),  which saturates the total spectra at order $k=2L$. The dark spectra that do not contribute to qubit dynamics are marked. The learning complexity, which equals the number of learnable spectra, is shaded in gray. The size of the sampling space is fundamentally determined by control complexity $L$, instead of the noise.    }
 \label{fig:learning_manifold}
\end{figure*}

\subsection{Saturation order in single-qubit CA QNS and fundamental digital QNS sample complexity}

We first focus on the single-qubit case. As the non-Gaussianity increases, the cumulative number of CA spectra scales exponentially. However, these spectra are subject to the control-induced symmetries. We emphasize that the dimension/size of the symmetry space (including the dark space for spectra and the binding symmetry) also scales exponentially as $K$. The learnable space can be defined as their subtraction. In what follows, we show that at $K=2L$, the symmetry space saturates the total spectral space, beyond which higher-order spectra always lie in the symmetry space and do not need to be learned. This idea is presented in the following theorem and illustrated in Fig.~\ref{fig:learning_manifold}.

\begin{theorem}[\textbf{Single-qubit dynamics saturation order}]
The complexity of the dynamics of a noisy single-qubit under digital control is fundamentally determined by the complexity of control (pulse number $L$) rather than the complexity of the noise (which can be infinite in non-Gassianity order). Characterizing digital CA spectra (through digital CA QNS) requires only a perturbation order truncation at $K=2L$ that saturates. 
\end{theorem}
\begin{proof}
All control tensors at order $k>2L$ are contractible ($\forall \vec{n}$) as described in the Proposition~\ref{prop-1}. In other words, all $\bar{S}^{\vec{\mu}}(\vec{n}) |_{\forall \mu;\forall \vec{n}}$ at order $k>2L$ are bound to some spectra at order $k\leq 2L$. 
\end{proof}
The above theorem implies  
\begin{equation}
    N^{(\infty)}_{\mathcal{C}_{\text{W}}}(L) = N^{(K=2L)}_{\mathcal{C}_{\text{W}}}(L)
\end{equation}
for single-qubit system, and we will simply use $N^{(\infty)}_{\mathcal{C}_{\text{W}}}$ for situations where saturation is reached.  At order $k=2L$,  the learnable spectra are $\bar{S}^{\vec{\mu}} (L,L,L-1,L-1,...,1,1)|_{\forall \mu}$, which include a unique $n$-string that avoids a 3-streak. In particular, classical spectrum $\bar{S}^{(\bm{0})}(L,L,L-1,L-1,...,1,1)$ is neither contractible nor dark.

\begin{definition} [\textbf{Fundamental digital QNS}]
A single-qubit digital CA QNS protocol truncated at saturation order is defined as fundamental digital QNS. The fundamental digital QNS reconstructs all CA spectra relevant to digitally-driven qubit dynamics.
\end{definition}
We have four remarks regarding fundamental digital QNS.
\begin{itemize}
    \item \textit{Fundamental digital QNS is only usefully defined for systems under instantaneous control, wherein the frames tied to the control consist of a finite number of non-overlapping functions in time. The overlapping, non-orthogonal frames used in Ref.~\cite{Teerawat_PRXQ_Frame}, can map to $L\rightarrow\infty$ window (or Walsh) functions. In such a situation, the saturation order is infinite, making the saturation order analysis lose any advantage.}  
    \item \textit{For non-fundamental digital QNS truncated below the saturation order, though the sample complexity is smaller, it only works well in a parameter regime (i.g., coupling is not very strong) where the unincorporated CA spectra are negligible; see  Case Study-0 in Appendix.~\ref{app:incremental_QNS} or Ref.~\cite{DongWZ_APL} for more discussion. Therefore, from a fault-tolerant perspective, using fundamental digital QNS at the cost of high sample complexity is necessary in a general parameter regime. }
    \item \textit{The concept of fundamental digital QNS is not limited to the single-qubit case; it also extends to two- and multi-qubit systems. In the following, we will introduce the two-qubit saturation bound, where the fundamental digital QNS naturally arises.  }
    \item \textit{As a specific digital CA QNS, fundamental digital QNS is \textit{all you need}: CA QNS trucated beyond saturation order is both redundant and impractical because higher-order dynamics (control tensors) are all simplifiable (contractible) and higher-order spectra are bound. }
\end{itemize} 

We now calculate the sample complexity of fundamental digital QNS to see how $N^{(\infty)}_{\mathcal{C}_{\text{W}}}$ scales with $L$, where $\infty$ indicates non-perturbative features fundamental digital QNS handles. This is simply counting how many independent/learnable CA spectra exist up to the saturation order. The relevant derivation based on simple combinatorics is elaborated in Appendix~\ref{app:bound}, and we conclude the following.

\begin{conclusion}[\textbf{Single-qubit fundamental digital QNS sample complexity}]
For single-qubit under digital control, the QNS involves sufficient $Q$-dynamics truncates at $k=2L$ ($L$ is the window frame size). The total sample overhead to learn classical noise is 
\begin{equation}
\begin{aligned}
     N^{(\infty)}_{\mathcal{C}_{\text{W}}}(L)\Big|_{1,C} &= 
  \sum_{k=1}^{2L}\Bigg[  \binom{L}{k}+ \sum_{t=\ceil{k/2}}^{k-1}  \binom{t}{k-t}  \binom{L}{t} \Bigg] \\
  &\approx {\Theta}(e^{1.1 L});
\end{aligned}
    \end{equation}
    and the resources needed to learn quantum noise is 
     \begin{equation}
     \begin{aligned}
          N^{(\infty)}_{\mathcal{C}_{\text{W}}}(L)\Big|_{1,Q} &\leq
        \sum_{k=1}^{2L}\Bigg[  \binom{L}{k}+ \sum_{t=\ceil{k/2}}^{k-1}  \binom{t}{k-t}  \binom{L}{t}\Bigg] 2^{k-1} \\
        &\approx {\mathcal{O}}(e^{2 L}).
     \end{aligned}
    \end{equation} 
    \label{corollary:1}  
\end{conclusion}

\subsubsection{Case Study 1: Accurate single-qubit dynamics description using fundamental digital QNS} 
In this case study, we apply a fundamental digital CA QNS protocol to effectively describe single-qubit dynamics under non-perturbative noise.

We consider a qubit strongly coupled to a semi-symmetric random telegraph noise (RTN) fluctuator, thus $H_{QE} =  g\sigma_z \beta(t)\mathbb{I}_E$ with $\beta(0)\equiv+1$. This fluctuator, a two-level system (TLS), randomly switches between $\pm1$ at a rate $\gamma$, where the number of switches in a time interval $(t_0, t_0+t)$ follows a Poisson distribution with mean $\gamma t$~\cite{Bergli_NJP_2009,Galperin_PRL_2006, Paladino_RMP_2014,Chantasri_2022_RTN, Balandin_APL_2024}.
The qubit coherence element $|\expect{\rho_{01}(t)}|$ ($|\expect{\rho_{01}(0)}|\equiv1$) decays rapidly under free evolution, which can be solved exactly~\cite{Santos_PRA_2005}, and is plotted in Fig.~\ref{fig:steps_QNS_recon} (c). Notice that the multiple ``steps'' in the coherence decay are a direct signature of strong non-Gaussianity, with the number of steps determined by the coupling strength $g/\gamma\gg 1$ and the long time $\gamma t \gtrsim 1$. 
Such a qubit decoherence is paradigmatically non-perturbative. As we have shown in Fig.~\ref{fig:steps_QNS_recon} (a,b), including the true spectra up to $K=14$ in the Dyson series still fails to accurately describe the qubit decoherence.

With the saturation bound in place, we conduct a series of numerical fundamental digital QNS simulations (see Table~\ref{tab:case_1} in Appendix~\ref{app:case_1} for protocol specifics). We choose $L=4$ digital control, thus necessitating the truncation only up to $K=8$, where $N^{(\infty)}_{\mathcal{C}_{\text{W}}}(L=4)=80$.  Information from all higher-order CA spectra influences is comprehensively captured within the well-considered CA spectra. 
Notice that for a fixed measurement time $T$, CA spectral information can only well describe qubit dynamics at time $mT/L$, where $m\in\{1,2,..., L\}$. To reconstruct the temporarily fine-grained ``steps'', we conduct measurements over multiple values of $T$ up to a maximum time $T_M$, thus allowing predictions of qubit dynamics at more frequent intervals as desired.
The reconstructed digital spectra are then plugged into the free evolution to accurately predict the coherence decay, as expected. 
As Fig.~\ref{fig:steps_QNS_recon} (c) shows, our fundamental digital QNS successfully describes the non-perturbative behavior of qubit coherence decay.
 \begin{figure*}[!htpb]
    \centering
    \includegraphics[width =0.8 \textwidth]{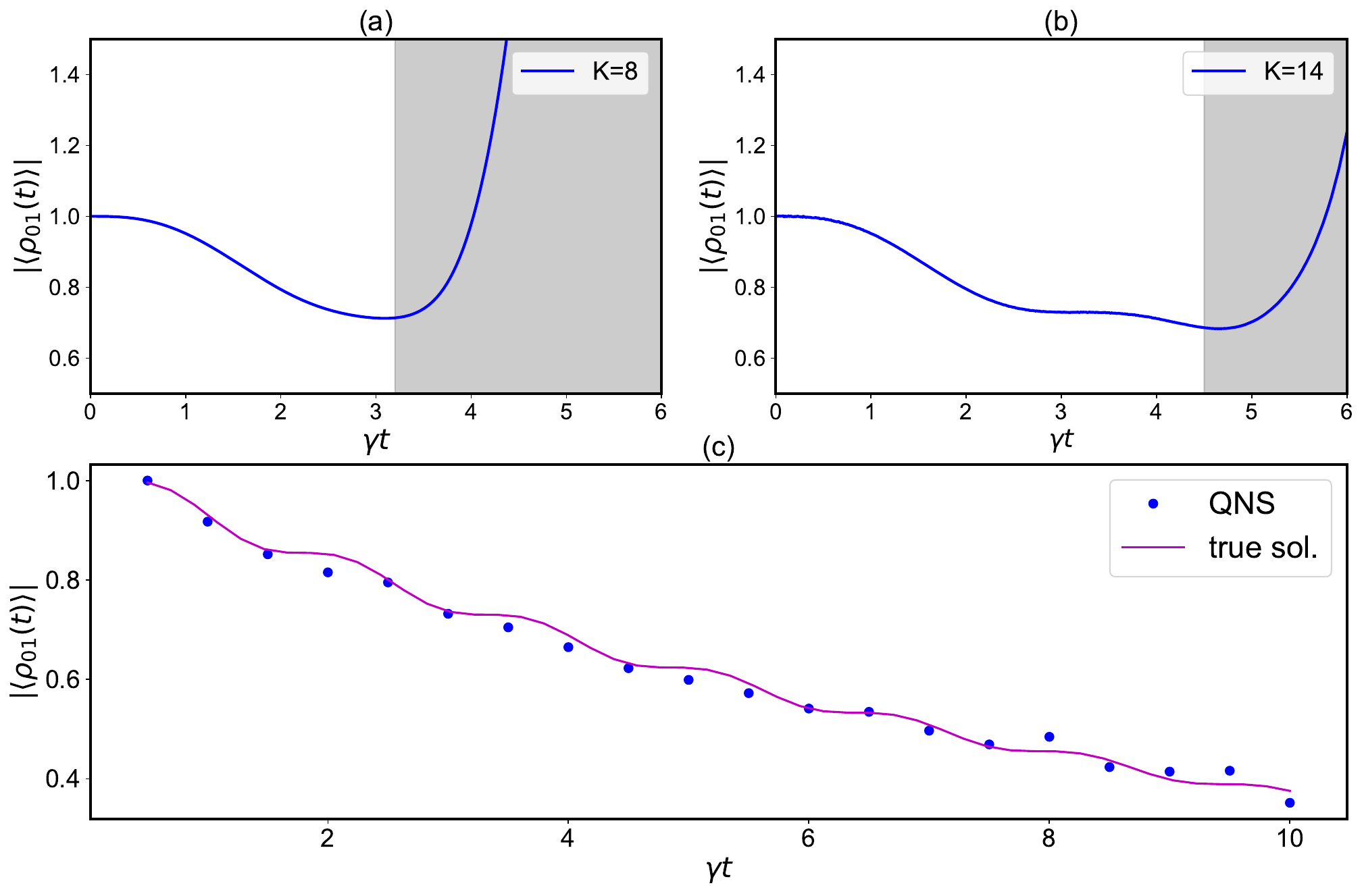}
    \caption{\textbf{Strong qubit decoherence dynamics description using fundamental digital QNS}
    (\textbf{a,b}) Qubit coherence decays over time, as demonstrated by incorporating different orders of the true spectra into the Dyson series with truncation. The noisy dynamics is inherently non-perturbative: truncation at (a) $K = 8$  is insufficient to predict even the first step. Truncation at  (b) $K = 14$  achieves the prediction of the first step.
    Prediction errors due to insufficient truncation are highlighted by the shaded areas, where coherence unphysically revives and even exceeds 1. Given the strong coupling $g\gamma \gg1$ and long evolution time $T\gamma \gtrsim1$, the principal contribution order far surpasses $K=14$; however, identifying the exact order is challenging due to the non-perturbative nature of the problem.
(\textbf{c}) Plot of the qubit coherence decay using exact full solution (without using perturbation theory) [purple line]. 
The coherence decay prediction uses the reconstructed spectra from fundamental digital QNS (a total of 80 spectra, not shown) with $L=4$ and $K=8$ [blue dots] and aligns closely with the exact solution, as justified by the saturation bound and binding symmetry analysis from the main text.
Fundamental QNS is fully performed at each of the 20 measurement times $T\in [1/20 \, T_M, 2/20 \, T_M, \dots, T_M]$, enabling dynamics predictions at times finer than $t=[1/4 \,T,1/2 \,T,3/4 \,T,T]$ for one fixed $T$. In simulation, we set  $T_M=10^{-5}$s,  $T_M\gamma=10$ and $g/\gamma=20$. 
}
    \label{fig:steps_QNS_recon}
\end{figure*}

\subsection{Saturation order in two-qubit CA QNS and fundamental digital QNS sample complexity}
 
The saturation order and fundamental QNS sample complexity in the two-qubit case are a generalization of the single-qubit result, using Proposition~\ref{prop-2}.
\begin{theorem} [\textbf{Two-qubit  saturation order}]
 The complexity of noisy two-qubit dynamics under digital control is fundamentally determined by the complexity of control (pulse number $L$) rather than the complexity of the environment. Characterizing noise only requires a truncation of the perturbation order at $K=4L$ that saturates. 
\end{theorem}
For $|Q|=2$, we have $N^{(\infty)}_{\mathcal{C}_{\text{W}}}(L) = N^{(K=4L)}_{\mathcal{C}_{\text{W}}}(L)$. Based on Proposition~\ref{prop-2}, we know that the learnable spectra at $K=4L$ are identified as $\bar{S}^{(\vec{\mu})}_{\vec{q}} (L,L,L,L,...,1,1,1,1)|_{\forall \mu}$, which is the only $n$-string configuration that avoids a 5-streak. All spectra at $k>4L$ are bound to lower-order spectra. The two-qubit fundamental digital QNS sample complexity is concluded below, whose details are in Appendix~\ref{app:bound}.

\begin{conclusion}[\textbf{Two-qubit fundamental digital QNS sample complexity}]
For two qubits under digital control, the QNS sample complexity is determined by the complexity of control rather than noise; the bound saturates at $k=4L$. The learning sample overhead to learn classical spectra is
\begin{equation}
\begin{aligned}
      N^{(\infty)}_{\mathcal{C}_{\text{W}}}(L)\Big|_{2,C} &= 
        \sum_{k=1}^{4L} \sum^k_{t=\ceil{k/4}}\binom{L}{t}  \sum_{p_4=0} ^{\floor{(k-t)/3}} \binom{t}{p_4} \sum_{p_3=0}^{\floor{k-t-3p_4}/2} \\
        &\binom{t-p_4}{p_3} 2^{p_3}  \binom{t-p_4-p_3}{k-t-2p_3-3p_4} 3^{k-t-2p_3-3p_4} \\
        &\approx \mathcal{O}(e^{5.1L})
\end{aligned}
    \end{equation}
\end{conclusion}

\subsubsection{Case Study 2: Two-qubit  noise C$\&$C using fundamental digital QNS }
To illustrate two-qubit fundamental digital QNS, we consider two qubits, $A$ and $B$, simultaneously coupled to a symmetric RTN. The RTN is modeled in the same manner as in Case Study 1. Both qubits are coupled equally to the fluctuator with a coupling strength $g$, resulting in an error Hamiltonian given by $H_{QE}(t) = g \sum_{q\in\{A,B\}}\sigma^{[q]}_z \beta(t)\mathbb{I}_E$, where $\beta(t)$ describes the RTN.  We consider $L=2$ for the digital control and $K=8$ for saturation, in which case there are a total of 80 independent (self and cross) spectra.  

In the two-qubit case, not only does each qubit undergo standard decoherence, but also, mediated by the fluctuator, these two qubits experience noisy spatial correlation, which is captured by the CA cross-spectra. This correlation can lead to non-unitary, non-local decoherence. We illustrate the spectral estimation of noise parameterized as $g/\gamma =5$ and show the $k=2,4$ order CA spectra in Fig.~\ref{fig:2qb_QNS_k2},~\ref{fig:2qb_QNS_k4} respectively. Although spectra at higher order up to $K=8$ are constructed, they are not shown here; however, they are involved in noise-tailored control design hereafter.
The control details for this two-qubit fundamental digital QNS  are given in Table~\ref{tab:2qb_fundamental_QNS} in Appendix~\ref{app:case_2}. As Figs.~\ref{fig:2qb_QNS_k2},~\ref{fig:2qb_QNS_k4} show, reconstructed spectra and spectral true values [labeled as ``raw''] do not match. This is due to tensor contraction and swap symmetry: certain spectra are bound by these symmetries, and the QNS protocol only reconstructs the bound form.  To properly illustrate the performance of spectral estimation, we apply the binding symmetry to the ``raw'' true spectral values, labeling the results as ``bound'' in the figures for reconstruction comparison.  For Gaussian spectral reconstruction in Fig.~\ref{fig:2qb_QNS_k2}, under swap symmetry, the reconstructed spectrum satisfies
$\widehat{\bar{S}}^{(0)}_{A,B}(1,1) \approx {\bar{S}}^{(0)}_{A,B}(1,1)|_{\text{bound}}  = \bar{S}^{(0)}_{A,B}(1,1)+ \bar{S}^{(0)}_{B,A}(1,1)|_{\text{raw}}$. For $k=4$ spectral reconstruction, only about a quarter of the raw $k=4$ CA spectra require explicit reconstruction, as our symmetry analysis enables model reduction.

\begin{figure*}
    \centering
    \includegraphics[width =1.0 \textwidth]{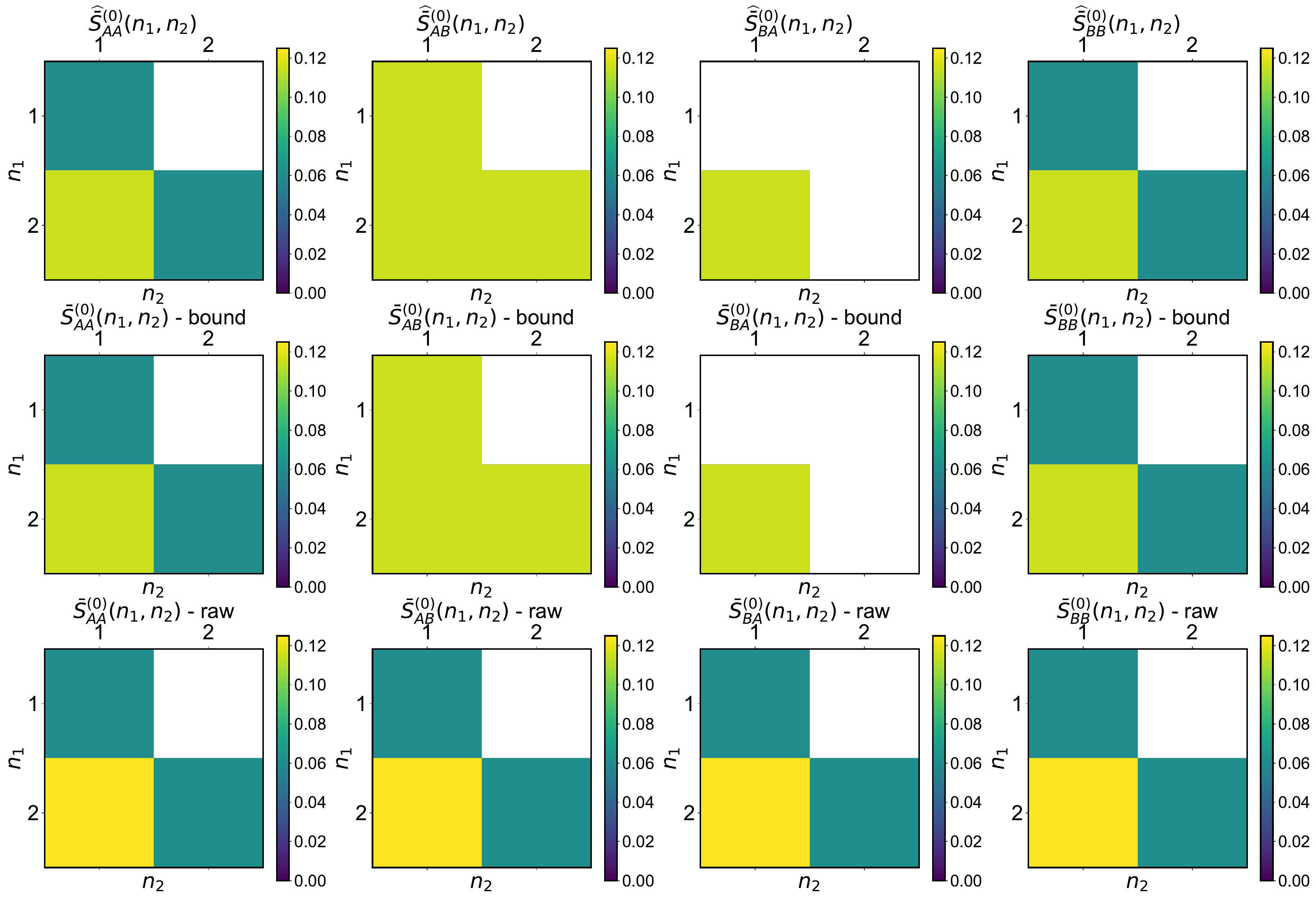}
    \caption{\textbf{Visualization of Gaussian components in two-qubit spectra from fundamental digital QNS}
    (\textbf{top}) The reconstructed two-qubit Gaussian ($k=2$) CA spectra $\{\widehat{\bar{S}}^{(0)}_{q_1,q_2}(n_1,n_2) \}$ ($q\in \{A,B\}, n\in\{1,2\}$) using  $K=8~L=2$ fundamental digital QNS, where each panel representing a fixed $(q_1,q_2)$ [labeled on-site]. In the numerical simulation, we consider a symmetric RTN with the parameters $\gamma =0.02$ MHz, and $\gamma T=0.08$. The RTN interacts with two qubits, both coupled equally with $g_A=g_B =5\gamma$. 
    (\textbf{middle}) The binding-symmetry form of true values [labeled as ``bound''] of the noise used in the simulation above. The binding symmetry incorporates swap symmetry and higher-order spectral binding symmetry. Such a bound form helps direct comparison with spectral reconstruction. 
    (\textbf{bottom}) The raw true values [labeled as ``raw''] of CA spectra without applying symmetry analysis. The middle plot is the product after applying symmetry analysis on the raw spectra.    Data plaquettes presented here but absent in the top and middle plots are subject to swap symmetry.      }
    \label{fig:2qb_QNS_k2}
\end{figure*}

\begin{figure*}
    \centering
    \includegraphics[width =1.0 \textwidth]{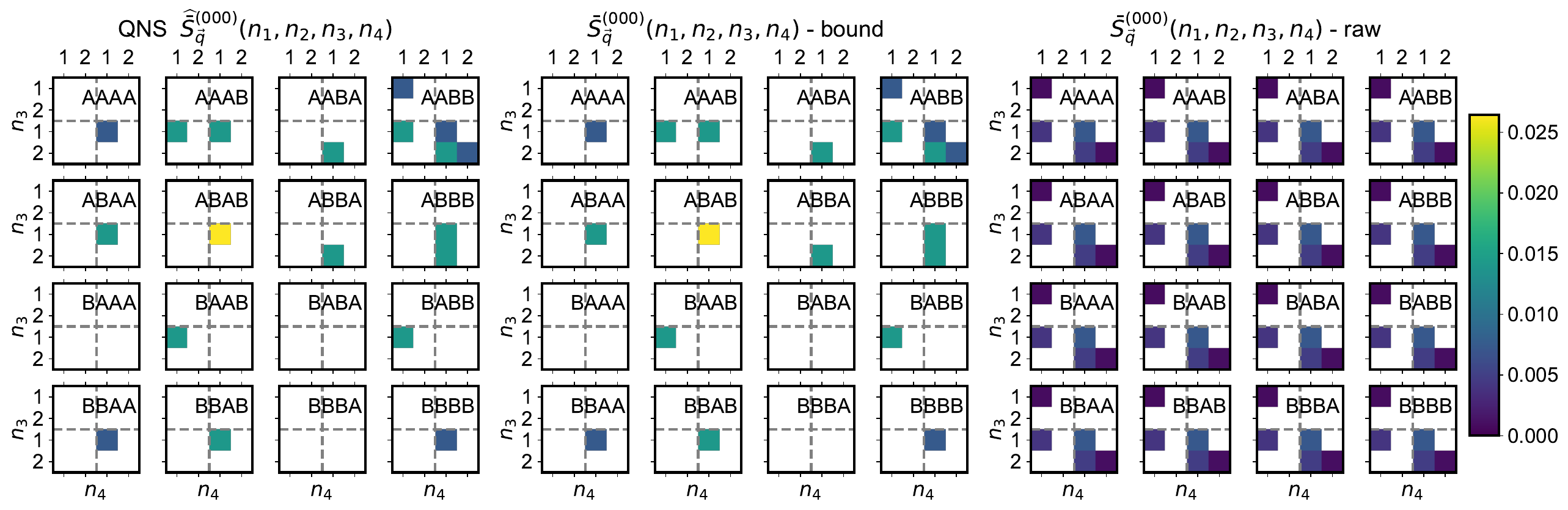}
    \caption{\textbf{Visualization of $4$-th order non-Gaussian components in two-qubit spectra from fundamental digital QNS}
    (\textbf{left}) The reconstructed two-qubit CA spectra $\{\bar{S}^{(0,0,0)}_{q_1,q_2,q_3, q_4}(n_1,n_2,n_3,n_4) \}$ ($q\in \{A,B\},~n\leq L=2$), with each panel representing a fixed $(q_1,q_2,q_3,q_4)$ [labeled on-site].  Each panel as a $(4,4)$ grid is further partitioned by the dashed lines into (2,2) blocks, representing fixing the $(n_1,n_2)$ values at $(1,1)$ [upper-left], $(1,2)$ [upper-right], $(2,1)$ [lower-left], and $(2,2)$ [lower-right].  Each panel's axes are labeled $n_3$ and $n_4$ (on-site) and are both marked $(1,2,1,2)$ due to the stride across different values of $n_1,n_2$. 
    (\textbf{middle}) The bound form of true values with binding and swap symmetry [labeled as ``bound''].
    (\textbf{right}). The raw form of true values [labeled as ``raw'']. 
    The noise parameters are the same as in Fig.~\ref{fig:2qb_QNS_k2}, as these two figures show two subsets of all spectral reconstruction with $K=8,~L=2$ fundamental digital QNS.  }
    \label{fig:2qb_QNS_k4}
\end{figure*}

We now present the results of the noise reduction. Given the two-qubit CA spectra, our control objective is to craft a noise-optimized \textit{idle circuit} (i.e., a two-qubit identity gate) over a fixed evolution period $T$~\cite{multi_time_opt_comment}. Notice that under the digital control assumption, setting control objectives to other actions (say, a CNOT gate) is trivial: one can always apply such a gate right before measurement without changing the noise propagator.
For every initial state $\rho_Q(0)$, the time-evolved state of the qubit can be expressed as $\rho_Q(T) = \sum_{u,v=0}^{15}\chi_{u,v}(T)\Lambda_u\rho_Q(0)\Lambda_v$, where $\rho_Q(T)$  departs $\rho_Q(0)$ by an amount that can be quantified by a cost function, e.g. process fidelity.~\cite{multi_time_opt_comment} 
The reconstructed spectra give us tomographically complete information to solve $\chi(T)$. To obtain the optimized idle circuit, we maximize the fidelity of the process $F(T)|^K_{\mathcal{F}} = \Tr[{\chi}(T) \tilde{\chi}(T)]^K_{\mathcal{F}}$, where ${\chi}(T)~[\tilde{\chi}(T)]$ is the actual (ideal) $\chi$ matrix, and the subscript $\mathcal{F}$ stands for the $L=2$ window frame $\mathcal{C}_\text{W}$ we are working on, and the truncation order $K$ indicates that only spectral information up to the $K$-th order is used in the numerical optimization. 
Concretely, the optimized digital control in such a frame, denoted  $P_{\text{opt}}|^{K}_{\mathcal{F}}$, is determined by: 
\begin{equation*}
        P_{\text{opt}}|^{K}_{\mathcal{F}} \equiv \argmax_{P} F \Big(T; P|_{\mathcal{F}},\widehat{\bar{S}}_{\vec{q}}^{(k\leq K)}(\vec{n}) \Big)\big|^{K}_{\mathcal{F}},
\end{equation*}
where $\mathcal{F} = \mathcal{F}|^{L=2}_{\mathcal{C}_\mathcal{W}}$ and we saturate truncation at $K=4L=8$.
The optimized control parameters are
\begin{widetext}
    \begin{equation}
\begin{aligned}
P = & \{ \theta_A(1) ,  \alpha_A(1),   \phi_A(1) ,  \theta_B(1) ,  \alpha_B(1) ,  \phi_B(1),   \Theta(1),   \rho(1) ,  \omega(1) ,  \theta^{\prime}_A(1) ,  \alpha^{\prime}_A(1) ,  \phi^{\prime}_A(1) ,  \theta^{\prime}_B(1)  , \alpha^{\prime}_B(1)   ,\phi^{\prime}_B(1) , \\
& \theta_A(2) ,  \alpha_A(2) ,  \phi_A(2)  , \theta_B(2)  , \alpha_B(2)  , \phi_B(2)  , \Theta(2)  , \rho(2)  , \omega(2)  , \theta^{\prime}_A(2)  , \alpha^{\prime}_A(2)  , \phi^{\prime}_A(2)  , \theta^{\prime}_B(2) ,  \alpha^{\prime}_B(2)  , \phi^{\prime}_B(2) \}
 \end{aligned}
 \end{equation}
\end{widetext}
which constitute the digital control propagator via $KAK$ decomposition:
\begin{widetext}
    \begin{equation}
\begin{aligned}
    U_0(n\tau)|_P  = \prod_{q\in Q} &e^{-i \theta_q(n) [\cos \phi_q(n) \sin \beta_q(n) \sigma_x^{[q]} + \cos \phi_q(n) \cos \beta_q(n) \sigma_y^{[q]} + \sin  \phi_q(n)  \sigma_z^{[q]}] } \\
    	        \times  &e^{-i \Theta(n) [\cos \varrho(n) \sin \omega(n) \sigma_x\sigma_x + \cos \varrho(n) \cos \omega(n) \sigma_y\sigma_y +    \sin \varrho(n)  \sigma_z\sigma_z ] } \times  \\ 
	        \prod_{q\in Q} &e^{-i \theta'_q(n)  [\cos \phi'_q(n) \sin \beta'_q(n) \sigma_x^{[q]} + \cos \phi'_q(n) \cos \beta'_q(n) \sigma_y^{[q]} + \sin  \phi'_q(n)  \sigma_z^{[q]}]  }, \\
\end{aligned}
\end{equation}
\end{widetext}
where $n=\{1,2\}$ and $Q=\{A,B\}$. The optimal control $P_{\text{opt}}$ is optimized by associating the above two equations. 

The maximal fidelity $F(T; P_{\text{opt}}|^{K}_{\mathcal{F}},\widehat{\bar{S}}{}^{(k\leq K)}(\vec{n}))|^{K}_{\mathcal{F}}$ obtained by the optimizer already accounts for the convolution between control and CA spectra beyond the truncation order ($\bar{S}^{(k > K)}(\vec{n})$) because the truncation is saturated; therefore, it represents the best fidelity one can obtain at infinite truncation order $ F(T; P_{\text{opt}}|^{K}_{\mathcal{F}},\widehat{\bar{S}}{}^{(k\leq K)}(\vec{n}))|^{K}_{\mathcal{F}} = F(T; P_{\text{opt}}|^{\infty}_{\mathcal{F}},\bar{S}{}^{(k < \infty)}(\vec{n}))|^{\infty}_{\mathcal{F}}.$

\begin{table*}
\centering
\resizebox{1.8\columnwidth}{!}{%
\begin{tabular}{| r|  rrrrrrrrrrrrrrr| }
\toprule
$g/\gamma$ & $\theta_A(1)$ &  $\alpha_A(1)$ &  $\phi_A(1)$ &  $\theta_B(1)$ &  $\alpha_B(1)$ &  $\phi_B(1)$ &  $\Theta(1)$ &  $\varrho(1)$ &  $\omega(1)$ &  $\theta^{\prime}_A(1)$ &  $\alpha^{\prime}_A(1)$ &  $\phi^{\prime}_A(1)$ &  $\theta^{\prime}_B(1)$ &  $\alpha^{\prime}_B(1)$ &  $\phi^{\prime}_B(1)$ \\
 &  $\theta_A(2)$ &  $\alpha_A(2)$ &  $\phi_A(2)$ &  $\theta_B(2)$ &  $\alpha_B(2)$ &  $\phi_B(2)$ &  $\Theta(2)$ &  $\varrho(2)$ &  $\omega(2)$ &  $\theta^{\prime}_A(2)$ &  $\alpha^{\prime}_A(2)$ &  $\phi^{\prime}_A(2)$ &  $\theta^{\prime}_B(2)$ &  $\alpha^{\prime}_B(2)$ &  $\phi^{\prime}_B(2)$ \\
\midrule
0.5& 0.0066 &0.6705 &-0.1233 &  -0.0377 &  -0.0083 & 0.0185 &-1.0224 &  0.0550 &-0.0519 &  -0.0167 &  -0.1566 &-0.2403 &  -0.0096 &0.2382 &-0.0556 \\
& -0.0128 & 0.2538 &-0.4564 &0.0175 &0.0659 &-0.0733 & 0.6807 &  0.0919 & 0.1550 &0.0005 &  -0.0194 &-0.1576 &  -0.0644 &0.0332 & 0.1568 \\
\hline
 1.0 &0.0101 &1.1962 &-0.4446 &  -0.0071 &  -0.0079 & 0.0207 &-0.9817 &  0.0298 & 0.1066 &  -0.0126 &  -0.1276 &-0.1780 &  -0.0187 &0.2379 &-0.0221 \\  &  -0.0143 &0.6568 &-0.4681 &0.0039 &  -0.1240 &-0.0523 & 0.6954 &  0.0963 & 0.0558 &0.0008 &  -0.0398 &-0.5486 &  -0.0295 &0.0804 & 0.2762 \\
 \hline
5&0.0133 &1.0146 &-0.1298 &  -0.0217 &  -0.0253 &-0.0412 &-0.9876 &  0.0199 & 0.0154 &  -0.0204 &  -0.0862 &-0.6216 &  -0.0061 &0.2046 &-0.0152 \\  & -0.0142 &0.8457 &-0.0754 &0.0018 &  -0.3661 &-0.1664 & 0.6770 &  0.1035 & 0.0311 &0.0009 &  -0.2459 & 2.1519 &  -0.0317 &0.0181 &-0.3409 \\
 \hline
 8& 0.0158 &1.0873 &-1.6006 &  -0.0281 &0.0412 &-0.1430 &-1.0099 &  0.0193 & 0.0319 &  -0.0216 &  -0.2141 &-0.6901 &0.0019 &0.5420 &-0.0126 \\  & 
 -0.0181 &0.7962 &-0.0479 &0.0014 &0.2880 & 0.2829 & 0.6595 &  0.1107 & 0.0709 &  -0.0039 &  -0.4244 & 0.4873 &  -0.0345 &  -0.1208 &-0.6822 \\
 \hline
10&  0.0238 &0.4380 &-2.7380 &  -0.0464 &  -0.2419 & 0.0593 &-1.1091 & -0.0094 &-0.0114 &  -0.0292 &  -0.0316 &-1.6617 &0.0017 &0.9232 & 0.1785 \\  
&-0.0199 &0.4429 & 1.5993 &0.0112 &4.1219 & 7.6854 & 0.5508 &  0.1534 &-0.1744 &  -0.0412 &0.0780 & 5.1569 &  -0.0645 &0.0077 &-0.8491 \\
  \bottomrule
\end{tabular}%
}
\caption{Optimized control $P_{\text{opt}}|^{K}_{\mathcal{F}}$ at different values of $g/\gamma$ (using the Nelder-Mead algorithm).  }
\label{tab:C_opt}
\end{table*}

In Figure~\ref{fig:2qb_optimal_control}, we show the performance of noise-optimized control (optimized) versus uncorrected (bare) control (the corresponding optimized controls are in Table.~\ref{tab:C_opt}). The significant fidelity margin of optimized control over bare control shows the usefulness of noise C$\&$C. This margin is more pronounced with the increase of coupling strength. We note that even though we know all essential spectral information relevant to such an $L=2$ digital control, optimized control is still not perfect and its performance will degrade with the increase of coupling strength (see Figure's insect ). In fact we note that there is a limit to the  ``correctability'' of digital control: precise knowledge of all necessary spectra does not mean one can design a control to fully correct the error.  Such correctability is control-limited. The most simple scenario illustrating this argument is for $L=1$ digital control. Even with precise knowledge of the CA spectra, the optimized control, in general, cannot change the noise evolution trajectory to reduce decoherence; it can only correct the unitary/coherent error.  

\begin{figure}[!htbp]
    \centering
    \includegraphics[width =0.45 \textwidth]{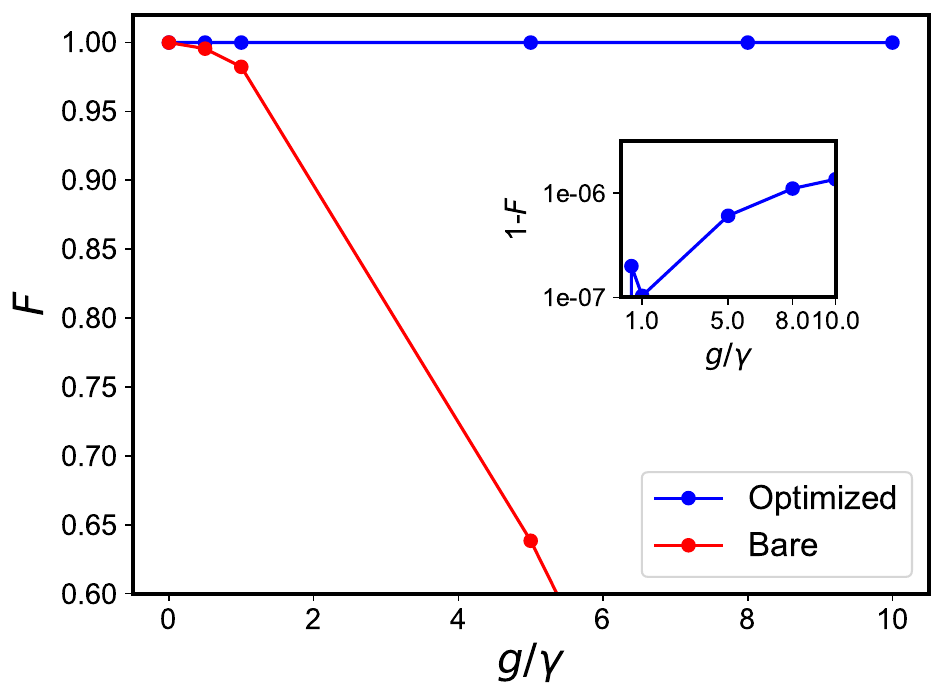}
    \caption{\textbf{Performance  of  noise-optimized  two-qubit control}.   The two-qubit process fidelity is plotted against the RTN coupling strength.  The ideal unitary is a two-qubit memory gate. Insect shows a zoomed-in view of optimized control.  }
    \label{fig:2qb_optimal_control}
\end{figure}

\subsection{Multi-qubit bound}
 
Though all physical unitaries are implemented by single and two-qubit gates, a study of multi-qubit bound is still useful for two reasons. First, error propagation within the quantum circuit leads to spatially non-local errors requiring characterization. Second, at the logical level, it is usually several physical qubits entangled with several ancillae that form the correction code to produce one clean logical qubit. In such a process, studying multi-qubit C\&C is instrumental for integrating a better logical qubit. 

For spatial noise beyond the two-qubit case, we identify that for spatially fully-correlated $|Q|$-qubit system, the corresponding  \textit{ |Q|-qubit saturation bound is} $K=2|Q|L$. 

\section{Discussion}
\label{sec:Discussion}
In this work, we have developed a comprehensive framework for characterizing and controlling digitally-driven quantum systems subject to general non-Gaussian dephasing noise. Supported by two pillars---control adaption~\cite{Teerawat_PRXQ_Frame} and digital control symmetry analysis---we have demonstrated how open quantum dynamics can be effectively managed using a resource-efficient approach that aligns with constrained controllability. Our findings reveal that even in the presence of highly complex, non-perturbative noise, the effective complexity of qubit dynamics is governed not by the noise but by the control. This insight underscores the power of control-driven approaches for mitigating noise and optimizing open quantum dynamics.

A central outcome of this study is the identification of saturation orders in the perturbative expansion for single- and two-qubit systems (and beyond) under digital control. By demonstrating that noise characterization saturates at specific truncation levels, our framework offers a bounded and tractable method for addressing non-perturbative noise. This result has significant implications for QNS, particularly in reducing the sample (computational and experimental)  overhead associated with learning high-order, non-Gaussian noise spectra. Through numerical simulations, we have illustrated the practical utility of our framework. Moreover, the framework facilitates the design of noise-adapted circuits, offering a pathway to improving the performance of near-term quantum devices without excessive resource demands.

Despite these advances, several challenges persist. The scalability of the framework is a notable limitation, as the learning complexity grows exponentially with the number of qubits and the size of the digital control. This growth is fundamentally tied to the non-Clifford nature of universal quantum circuits~\cite{Gottesman_Knill} and the intrinsic properties of non-Gaussian noise, highlighting the necessity of further innovation in control design and noise mitigation strategies. Moreover, the framework is currently restricted to digital controls, where noise C\&C can be done alternatively by non-Markovian process tomography. Our present results do not address the complexities introduced by non-instantaneous controls, which are increasingly relevant in fast-clock quantum circuits. Extending the principles of our formalism to these scenarios, along with incorporating relaxation (in this case our central results still hold) and other types of errors~\cite{Qasim_SPAM_paper}, represents an important avenue for future research~\cite{forthcoming_comment}.


In conclusion, this study demonstrated that aligning noise characterization with control capabilities offers a powerful strategy for addressing the complexities of open quantum systems. Our framework may open new avenues for resource-efficient quantum error mitigation and optimization. While challenges remain, the insights gained here contribute to advancing the control and characterization of quantum devices in the presence of high-complexity noise.

\section*{Acknowledgments}
W.D. conducted this work independently with self-support. Y.W. was supported by the Innovation Program for Quantum Science and Technology (2023ZD0301400) and the National Natural Science Foundation of China (grant 12288201). The authors would like to express their gratitude to Gerardo Paz-Silva and Lorenza Viola for their insightful discussions on frames and non-Gaussian quantum noise spectroscopy.  W.D. thanks Ren-Bao Liu, Gerardo Adesso, and Madalin Guta for discussions. 
W.D. is also grateful to Xiu-Hao Deng and his Quantum Dynamics \& Control Group for providing a collaborative environment where open-minded intellectuals can engage in productive exchanges. Additionally, W.D. thanks Arian Vezvaee for his valuable comments on the manuscript, and Sophia Economou for her unwavering guidance over the years, teaching not only how to be an excellent researcher, but also, and more importantly, how to be a better human being.

\onecolumngrid
\appendix
\section{Dyson series of qubit dynamics }
\label{app:dyson}
In this Appendix, we give the algebraic details to derive Eq.~\ref{eq:dyson} in the main text. 
\subsection{Derivation of Dyson series}
Recall that 
\begin{equation}
\begin{aligned}
  \mathbb{E}[O(T)]_{\rho_Q\otimes \rho_E}=\Tr_Q\Big [\expect{\mathcal{T}_{+} e^{-i \int^T_{-T}{H}_O(t)dt }} \rho_Q \tilde{O}(T)\Big]= \Tr[\sum_{k=0} ^{\infty}   \mathcal{D}^{(k)}_O(T)/k!  \rho_Q  \tilde{O}(T) ] .
  \end{aligned}
\end{equation}
The Dyson term is defined as
\begin{equation}
\begin{aligned}
\mathcal{D}^{(k)}_O(T)/k! &= (-i)^k \int^T_{-T} d_> \vec{t}_{[k]} \expect{H_O(t_1) ... H_O(t_k) } \\
&= (-i)^k \sum_{\vec{q}}\sum\limits^{k}_{l=0}\sum\limits_{\pi\in \Pi_{l;k}} \int^T_{0}d_>\vec{t}_{[k]} \Big\langle \prod\limits^l_{j=1}\bar{H}_{q_j} (t_{\pi(j)}) \prod\limits^k_{j'=l+1} \tilde{H}_{q_{j'}} (t_{\pi(j')})\Big\rangle \\
& = (-i)^k \sum_{\vec{q}}\sum\limits^{k}_{l=0}\sum\limits_{\pi\in \Pi_{l;k}} \int^T_{0}d_>\vec{t}_{[k]}  \prod\limits^l_{j=1}\bar{h}_{q_j} (t_{\pi(j)}) \prod\limits^k_{j'=l+1} \tilde{h}_{q_{j'}} (t_{\pi(j')}) \Big\langle B_{q_1}(t_{\pi(1)})\cdots B_{q_k}(t_{\pi(k)})\Big\rangle,  \\
\end{aligned}
\end{equation}
where we have used the binary expansion of  $H_O(t)$ supported on $[-T,0]$ and $[0,T]$. The last equality is simply the separation of system and environment. The $\pi \in \Pi_{l;k}$ therein organize all admissible permutations, and $ \text{Card}( \{\Pi_{l;k}\}_{l=0}^k) = 2^k$ equals the number of  Hamiltonian string combinations.  The definition of Dyson series and the extension of the time integral domain lower bound to $-T$ limit that the times are always ordered in the control string $\prod_j \tilde{h}_q(t_{\pi(j)})$ and reversely-ordered in $\prod_j \bar{h}_q(t_{\pi(j)})$ respectively. For example, for $k=2$, one has 
\begin{equation}\label{eq:A3}
\mathcal{D}^{(2)}_O(T)/2! = - \sum_{\vec{q}} \int^T_{0} d_> \vec{t}_{[2]} \Big[  \expect{\tilde{H}_{q_1}(t_1)\tilde{H}_{q_2}(t_2)} +  \expect{\bar{H}_{q_1}(t_1)\tilde{H}_{q_2}(t_2)} + \expect{\bar{H}_{q_1}(t_2)\tilde{H}_{q_2}(t_1)} +  \expect{\bar{H}_{q_1}(t_2)\bar{H}_{q_2}(t_1)}   	 \Big].
\end{equation} 
To proceed, notice that we can write $\bar{H}(t)$ as
\begin{equation}
\begin{aligned}
    \bar{H}(t) =\sum_{q\in Q} \bar{H}_q(t) &=  \sum_{q \in Q} \sum_{u=1}^{4^{|Q|}-1}y_{[q],u}(t) \big[\sum^{4^{|Q|}-1}_{c=1} f^{u}_c \Lambda_c \big]\otimes B_q(t)\\
    &=\sum_{q \in Q} \sum_{u;c}^{4^{|Q|}-1}y_{[q],c}(t) f^{c}_u \Lambda_u \otimes B_q(t)
\end{aligned}
\end{equation} 
where the conjugation factor $f^{c}_u \equiv -\frac{1}{2^{|Q|}} \Tr[\tilde{O}^{-1}(T) \Lambda_{c}\tilde{O}(T) \Lambda_{u}] $ is symmetric $f^u_c = f^c_u$ for  unitary and Hermitian $O(T)$.  Though there are $2^k$ distinct Hamiltonian strings, there are only $2^{k-1}$ noise correlators, that are represented as $\expect{B_{q_1}(t_{\pi(1)}) ... B_{q_k}(t_{\pi(k)})}$. For example, $\mathcal{D}^{(2)}_{O}(T)$ only admits $\expect{B_{q_1}(t_1)B_{q_2}(t_2)}$ and $\expect{B_{q_1}(t_2)B_{q_2}(t_1)}$ from the four Hamiltonian strings in Eq.~\ref{eq:A3}.  This is because the conjugation only applies to the system $Q$ rather than the environment $E$; more details are explained in Appendix~\ref{app:dyson} 2. Applying the above decomposition, the Dyson term can be further expressed as
\begin{equation}
\begin{aligned}
      \mathcal{D}^{(k)}_O(T)/k! 
      =(-i)^k&  \sum_{\substack{\vec{q}; ~l; \\ \pi; ~\vec{u}; ~\vec{c}}}   \int^T_0d_>\vec{t}_{[k]} 
      \Big[  \prod_{j=1}^l\prod_{j'=l+1}^k f^{c_j}_{u_{j}}y_{[q_j],c_j}(t_{\pi(j)})y_{[q_{j'}],u_{j'}}(t_{\pi(j')})\Lambda_{u_j}\Lambda_{u_{j'}}  \Big] \times  \Big\langle B_{q_1}(t_{\pi(1)})\cdots B_{q_k}(t_{\pi(k)})\Big\rangle.
\end{aligned}
\end{equation}
Notice that the noise correlators are generally out-of-time-ordered due to the action by $\pi$, and we will switch to a time-ordered nested bracket representation through some linear transformation. Let us define 
\begin{equation}
\mathcal{B}^{\vec{\mu}_{[k-1]}}_{ \vec{q}_{[k]}} (\vec{t})\equiv     \frac{1}{2^{k-1}} \widehat{P}_{\vec{q}}\big(\llbracket ...\llbracket B_{q_1}(t_1), B_{q_2}(t_2) \rrbracket_{\mu_1},  B_{q_3}(t_3) \rrbracket_{\mu_2}, ..,  B_{q_k}(t_k) \rrbracket_{\mu_{k-1}} \big),
\end{equation}
where  $ \llbracket X,Y \rrbracket_{\mu} = XY+(-1)^{\mu} YX$, and the operator $\widehat{P}_{\vec{q}}$ is a $q$-ordering operator such that the expanded $k$-length bath operator string is always indexed by $\{q_1,...,q_k\}$ from left to right.  As  in the  out-of-time-ordered noise correlator, the $q$-string is always ordered. Indeed, $\widehat{P}_{\vec{q}}$ helps all $q$  avoid forming faulty permutation by bracket. 

Now we claim that the above nested bracket representations are complete: any admissible noise correlator can be written as their combination as  
\begin{equation}
     B_{q_1}(t_{\pi(1)}) \cdots  B_{q_k}(t_{\pi(k)}) = \sum_{\vec{\mu}_{[k-1]}\in \{0,1\}^{\otimes k-1}}  (-1)^{\bar{f}^{(k)}_{\pi}(\vec{\mu})}  \mathcal{B}^{\vec{\mu}_{[k-1]}}_{ \vec{q}_{[k]}} (t_1, ..., t_k),
\end{equation}
that also trivially validates
\begin{equation}
    \Big\langle B_{q_1}(t_{\pi(1)}) \cdots B_{q_k}(t_{\pi(k)})\Big\rangle = \sum_{\vec{\mu}_{[k-1]}\in \{0,1\}^{\otimes k-1}}  (-1)^{\bar{f}^{(k)}_{\pi}(\vec{\mu})}  \Big\langle\mathcal{B}^{\vec{\mu}_{[k-1]}}_{ \vec{q}_{[k]}} (t_1, ..., t_k)\Big\rangle,
\end{equation}
where $\bar{f}^{(k)}_{\pi}(\vec{\mu}) \in \{0,1\}$.  To clarify the punchline, we wish the readers to refer to the next subsection for its justification.   As for now,  we elaborate on $k=2$ as a pedagogic example,  in which case the conventional noise correlator includes $\expect{B_{q_1}(t_1)B_{q_2}(t_2)}$ and $\expect{B_{q_1}(t_2)B_{q_2}(t_1)}$ for $q_1,q_2 \in Q$. While in the bracketor representation, one has  
\begin{equation}
\begin{aligned}
        \expect{\mathcal{B}^{\mu}_{ \vec{q}_{[2]}} (t_1,t_2)} = &\frac{1}{2}  \Big[\widehat{P}_{\vec{q}}\expect{B_{q_1}(t_1)B_{q_2}(t_2)} +(-1)^{\mu} \widehat{P}_{\vec{q}}\expect{ B_{q_2}(t_2)B_{q_1}(t_1)} \Big] \\
        =& \frac{1}{2} \Big[\expect{B_{q_1}(t_1)B_{q_2}(t_2)} +(-1)^{\mu} \expect{B_{q_1}(t_2)B_{q_2}(t_1)} \Big],
\end{aligned}
\end{equation}
where $\mu\in \{0,1\}$ and in the second equality, the $q$-ordering is regularized by $\widehat{P}$. It is ready to check the relation between these two representations: for $k=2$, under triviality $\pi = e$ and two-points permutation $\pi =(1,2)$ gives
\begin{equation}
\begin{aligned}
\expect{B_{q_1}(t_1)B_{q_2}(t_2)}= \expect{\mathcal{B}_{q_1,q_2}^{(0)}} +\expect{ \mathcal{B}_{q_1,q_2}^{(1)}}, \quad \expect{B_{q_1}(t_2)B_{q_2}(t_1)}= \expect{\mathcal{B}_{q_1,q_2}^{(0)}} - \expect{\mathcal{B}_{q_1,q_2}^{(1)}}.
\end{aligned}
\end{equation}
from which one can identify the  sign function $\bar{f}^{(2)}_{\pi =e}(\mu=0) = \bar{f}^{(2)}_{\pi =e}(\mu=1) =0 $ and  $\bar{f}^{(2)}_{\pi =(1,2)}(\mu=0) =0, \bar{f}^{(2)}_{\pi =(1,2)}(\mu=1) =1 $. 

To conclude, using the representations introduced above, we finally obtain the final form in Eq.~\ref{eq:dyson}.  

\subsection{Completeness of nested bracket representation}
\label{app:nest_comm_complete}

This section is fully dedicated to the details of nested bracket representation introduced above. 

To begin with, recall that $ \text{Card}(\{\Pi_{l;k}\}_{l=0}^k) = 2^k$ on the $Q$ side, and we explain why there are only $2^{k-1}$ distinct noise correlators in $E$. We already know that in any $\tilde{H}$ ($\bar{H}$) string the times are ordered (reversely-ordered), and therefore the only effect of $\pi$ is just to allocate each $t_{\pi(i)}$ in $\tilde{H}$ or $\bar{H}$. Once this allocation is completed, the two strings are fully determined. Crucially, $t_1$ is always the largest time point and whether it is allocated to $\tilde{H}$ or $\bar{H}$ does not change its position in the whole $k$-length string, which thus results in the extra factor of 2.


To proceed, let us introduce a binary representation as follows:
\begin{equation}
    \bm{\pi}_Q := \text{binary form }(\pi),
\end{equation}
where $\bm{\pi}_Q$ is essentially a $k$-length binary string with each element in $\{0,1\}$. Specifically, the $i$-th element of $\bm{\pi}_Q $  is 1 if $t_{i}|_{i=\pi(j)}$ is in $\bar{H}$ string, and 0 if in $\tilde{H}$.  For any $k$, all $\bm{\pi}_Q$ form a group of order $2^k$ under bitwise  XOR operation. The group $\{\bm{\pi}_Q\}$ includes complete information of permutations on $Q$.
The $2^{k-1}$ distinct noise correlators can be identified using the following quotient  group
\begin{equation}
    \{\bm{\pi}_E \}:= \{\bm{\pi}_Q \}/\mathbb{Z}_2,
\end{equation}
where $\bm{\pi}_E$ is defined as permutations on $E$. Interestingly, $\bm{\pi}_E$ is a $(k-1)$-length binary string after deleting the first element of $\bm{\pi}_Q$. This is due to the two-fold degeneracy of $t_1$'s location.  We remind the readers that in $k$-th order Dyson series, $\bm{\pi}_E$ and $\vec{\mu}_{[k-1]}$ are both binary vectors with length $k-1$. For example, at $k=2$, $\{\pi_Q\} =\{(00),~(10),~(01),~(11)\}$ (corresponding to the four terms in Eq.~\ref{eq:A3}) and $\{\pi_E\} = \{(0),~(1)\}$. 

Next we demonstrate that the nested bracket representation is complete: each correlator in the restricted permutation subset is a linear combination of them as we claimed.  Notice that due to the fact that  $q$'s are always ordered in both representations, they will be dropped for convenience.
\begin{theorem}
Any $k$-point bath correlator $B(t_{\pi(1)})\cdots B(t_{\pi(k)})$ is a linear combination of $(k-1)$-level nested bracket set of $\mathcal{B}^{(\mu_1,\cdots,\mu_{k-1})}(t_1,\cdots,t_k)$, $\forall \mu_j\in\{0,1\}$, as
\begin{equation}
    B(t_{\pi(1)})\cdots B(t_{\pi(k)})|_{\pi\sim \bm{\pi}_E} = \sum_{\bm{\mu}\in \{0,1\}^{\otimes k-1}} (-1)^{\bm{\mu}\cdot\bm{\pi}_E} \mathcal{B}^{\bm{\mu}}(\vec{t}),
\end{equation}
where $\bm{\mu}=\vec{\mu}_{[k-1]} = (\mu_1,\cdots,\mu_{k-1})$ and  $\bm{\pi}_E$ is the $(k-1)$-length binary vector aforementioned. 
\end{theorem}
\begin{proof} 
We prove it by induction. The basic case is $k=2$. It is easy to check that 
\begin{equation}
\begin{aligned}
    {B(t_1)B(t_2)}|_{\bm{\pi}_E=(0)}=  \mathcal{B}^{{0}}({t})+ \mathcal{B}^{{1}}({t}) = \sum_{\mu=0,1}(-1)^{\mu\cdot 0} \mathcal{B}^{{\mu}}({t}), \\
     {B(t_2)B(t_1)}|_{\bm{\pi}_E=(1)}  =  \mathcal{B}^{{0}}({t})- \mathcal{B}^{{1}}({t})= \sum_{\mu=0,1}(-1)^{\mu\cdot 1} \mathcal{B}^{{\mu}}({t}).\\    
\end{aligned}
\end{equation}
We assume our inductive hypothesis, that is for $k$, one has 
\begin{equation}
    {B(t_{\pi(1)})\cdots B(t_{\pi(k)})}|_{\pi\sim \bm{\pi}_E} = \sum_{\bm{\mu}} (-1)^{\bm{\mu}\cdot\bm{\pi}_E} \mathcal{B}^{\bm{\mu}}(\vec{t}),
\label{eq:induction_B}
\end{equation}
as true. 
For $k+1$, as $t_{k+1}$ is the smallest time point, there are only two types of viable configurations, namely $(B(t_{\pi(1)})\cdots  B(t_{\pi(k)}))B(t_{k+1})$ and $B(t_{k+1})(B(t_{\pi(1)})\cdots  B(t_{\pi(k)}))$, which correspond to $\bm{\pi}_E\rightarrow\bm{\pi}_E\oplus0$ and $\bm{\pi}_E\rightarrow\bm{\pi}_E\oplus1$ in the extended structure, respectively, where $\pi$ and $\bm{\pi}_E$ are consistent with Eq.~\ref{eq:induction_B}. We focus on the first case, and we have
\begin{equation}
    \begin{aligned}
     &(B(t_{\pi(1)})\cdots  B(t_{\pi(k)}))B(t_{k+1})|_{\pi\sim \bm{\pi}_E}  =  \sum_{\bm{\mu}} (-1)^{\bm{\mu}\cdot\bm{\pi}_E} \mathcal{B}^{\bm{\mu}}(\vec{t})B(t_{k+1})  \\
     &= \sum_{\bm{\mu}} (-1)^{\bm{\mu}\cdot\bm{\pi}_E} \big[ \mathcal{B}^{\bm{\mu}}(\vec{t}) B(t_{k+1}) + B(t_{k+1}) \mathcal{B}^{\bm{\mu}}(\vec{t}) +  (-1)^0[\mathcal{B}^{\bm{\mu}}(\vec{t}) B(t_{k+1}) - B(t_{k+1})\mathcal{B}^{\bm{\mu}}(\vec{t})] \big]/2 \\
     &= \sum_{\bm{\mu}} (-1)^{\bm{\mu}\cdot\bm{\pi}_E} [\mathcal{B}^{\bm{\mu}\oplus0}(\vec{t}) + (-1)^0 \mathcal{B}^{\bm{\mu}\oplus1}(\vec{t})] \\
     & = \sum_{\bm{\mu}} (-1)^{\bm{\mu}\cdot\bm{\pi}_E + 0\cdot0} \mathcal{B}^{\bm{\mu}\oplus0}(\vec{t}) +  \sum_{\bm{\mu}} (-1)^{\bm{\mu}\cdot\bm{\pi}_E + 1\cdot0}\mathcal{B}^{\bm{\mu}\oplus1}(\vec{t}) \\
     &= \sum_{\bm{\mu}\oplus0} (-1)^{(\bm{\mu}\oplus0)\cdot(\bm{\pi}_E \oplus0)} \mathcal{B}^{\bm{\mu}\oplus0}(\vec{t}) +  \sum_{\bm{\mu}\oplus1} (-1)^{(\bm{\mu}\oplus1)\cdot (\bm{\pi}_E \oplus0)}\mathcal{B}^{\bm{\mu}\oplus1}(\vec{t}) \\
     & = \sum_{\bm{\mu}\oplus\mu}(-1)^{(\bm{\mu}\oplus\mu)\cdot(\bm{\pi}_E\oplus0)} \mathcal{B}^{\bm{\mu}\oplus\mu}(\vec{t}),
    \end{aligned}
\end{equation}
where $\bm{\mu}\oplus\mu =(\mu_1, \cdots,\mu_k)$,
which validates our hypothesis. The second case can be resolved similarly. Thus the statement holds for all $k\geq2$. 
\end{proof}

The sign function $\bar{f}$ in the main text and in this section is detailed as $\bar{f}^{(k)}_{\pi}(\bm{\mu}) = \bm{\mu}\cdot \bm{\pi}_E$. We thus finished all the technical details of  Eq.~\ref{eq:dyson} in the main text.

\section{Control-adapted frame representation of Dyson series}
\label{app:ugly_D1_D4}
In this Appendix,  we expand and obtain the full form of single-qubit Dyson series in the CA picture for single-qubit. We notice that the frame $\mathcal{C}$ used hereafter in this single section is not limited to digital $\mathcal{C}_{\text{W}}$, but also allows non-overlapping yet unequal-distance or overlapping functions as frames. 

The first two orders' expressions match the results in ~\cite{Teerawat_PRXQ_Frame}:
\begin{equation}
    \begin{aligned}
      \mathcal{D}^{(1)}_O(T) &= (-i)\Big(\int^T_0ds\langle\bar{H}(s)\rangle+\int^T_0ds\langle\tilde{H}(s)\rangle \Big)\\
     &= (-i) \sum_u\sum_n \sigma_u F^{(1,+)}_u(n)\bar{{S}}(n)
    \end{aligned}
\end{equation}

\begin{equation}
\begin{aligned}
\Tr[\mathcal{D}^{(1)}_O(T)\sigma_0\sigma_{\gamma}]&=0  \\
     \Tr[\mathcal{D}^{(1)}_O(T)\sigma_{\gamma}\sigma_{\gamma}]&=0 \\
     \Tr[\mathcal{D}^{(1)}_O(T)\sigma_r\sigma_{\gamma}]&= 4\sum_{\beta}\sum_n \epsilon(\beta,r,\gamma)F^{(1)}_{\beta}(n)\bar{{S}}(n)
     \end{aligned}
\end{equation}

\begin{equation}
    \begin{aligned}
     \frac{ \mathcal{D}^{(2)}_O(T)}{2!} &= (-i)^2 \int^T_0 d_>\vec{t}_{[2]} \Big[ \langle \tilde{H}(t_1)\tilde{H}(t_2) \rangle+ \langle\bar{H}(t_1)\tilde{H}(t_2) \rangle+\langle\bar{H}(t_2)\tilde{H}(t_1) \rangle+\langle \bar{H}(t_2)\bar{H}(t_1)\rangle  \Big] \\
      &=(-i)^2 \Big[\frac{\sigma_0}{2} \sum_{\mu}\sum_{\vec{n}}\sum_u\bar{{S}}^{(\mu)}(n_1,n_2)F^{(1,0)}_u(n_1)F^{(1,\mu)}_u(n_2) \\ 
&\quad\quad\quad\quad+\sum_{\mu}\sum_{\vec{n}}\sum_{u\neq v} \frac{\sigma_u\sigma_v}{2}\bar{{S}}^{(\mu)}(n_1,n_2) F^{(1,0)}_u(n_1)F^{(1,\bar{\mu})}_v(n_2) \Big]
    \end{aligned}
\end{equation}

\begin{equation}
    \begin{aligned}
      \Tr[\frac{ \mathcal{D}^{(2)}_O(T)}{2!}\sigma_0\sigma_{\gamma}] &= -4i \sum_{\vec{n}} \sum_{u<v\neq \gamma} \epsilon(u,v,\gamma) F^{(1)}_u(n_1)F^{(1)}_v(n_2) [\bar{S}^{(1)}(n_1,n_2)-\bar{{S}}^{(1)}(n_2,n_1)], \\
      \Tr[\frac{ \mathcal{D}^{(2)}_O(T)}{2!}\sigma_{\gamma}\sigma_{\gamma}] &= -4 \sum_{\vec{n}} \sum_{u\neq\gamma}F^{(1)}_u(n_1) F^{(1)}_u(n_2) \bar{S}^{(0)}(n_1,n_2), \\
      \Tr[\frac{ \mathcal{D}^{(2)}_O(T)}{2!}\sigma_r\sigma_{\gamma}] &= 4 \sum_{\vec{n}} F^{(1)}_r(n_1) F^{(1)}_{\gamma}(n_2) \bar{{S}}^{(0)}(n_1,n_2).
    \end{aligned}
\end{equation}

Starting  from the third order, expressions are complicated, and it is helpful to introduce a shorthand notation for frame-based ``CA filter strings'' as
$$ \bm{F}^{(1;\vec{p})}_{\vec{u}} \equiv \prod_{j=1}^k F^{(1,p_j)}_{u_j}(n_j),$$
$$ F^{(1,\pm)}_{u}(n) =  F^{(1)}_{u}(n)\pm \sum_c f^c_u F^{(1)}_{c}(n),$$
\begin{equation}
    \begin{aligned}
      \textcolor{black}{\frac{\mathcal{D}^{(3)}_O(T)}{3!}}&= (-i)^3 \int^T_0 d_>\vec{t}_{[3]} \Big[  \langle\tilde{H}(t_1)\tilde{H}(t_2)\tilde{H}(t_3) \rangle+ \langle \bar{H}(t_2)\bar{H}(t_1)\tilde{H}(t_3)\rangle+\langle \bar{H}(t_3)\bar{H}(t_1)\tilde{H}(t_2)\rangle+\langle \bar{H}(t_3)\bar{H}(t_2)\tilde{H}(t_1) \rangle\\
        &\quad \quad\quad \quad\quad \quad\quad \quad+ \langle\bar{H}(t_2)\tilde{H}(t_1)\tilde{H}(t_3) \rangle+\langle\bar{H}(t_3)\tilde{H}(t_1)\tilde{H}(t_2) \rangle+\langle \bar{H}(t_1)\tilde{H}(t_2)\tilde{H}(t_3)\rangle+\langle \bar{H}(t_3)\bar{H}(t_2)\bar{H}(t_1)\rangle \Big] \\
        & = \frac{(-i)^3}{4}\Big\{ \sum_{\vec{n}}\sum_{\mu,\nu}\sum_{u\neq v\neq w} \textcolor{black}{i \epsilon_{u,v,w}}\sigma_0 F^{(1,0)}_u(n_1)F^{(1,\bar{\nu})}_v(n_2)F^{(1,\mu)}_w(n_3) \bar{{S}}^{(\nu,\mu)}(n_1,n_2,n_3) \\
        &\quad\quad\quad\quad +\sum_{\vec{n}}\sum_{\mu,\nu} \sum_{u\neq v} \sigma_v [ F^{(1,0)}_u(n_1)F^{(1,\nu)}_u(n_2)F^{(1,\mu)}_v(n_3) -F^{(1,0)}_u(n_1)F^{(1,\bar{\nu})}_v(n_2)F^{(1,\bar{\mu})}_u(n_3) \\    &\quad\quad\quad\quad\quad\quad\quad\quad\quad\quad   + F^{(1,0)}_v(n_1)F^{(1,\bar{\nu})}_u(n_2)F^{(1,\bar{\mu})}_u(n_3)+F^{(1,0)}_v(n_1)F^{(1,{\nu})}_v(n_2)F^{(1,{\mu})}_v(n_3)]\bar{{S}}^{(\nu,\mu)}(n_1,n_2,n_3)\Big\} \\
        &= \frac{(-i)^3}{4}\Big\{\sum_{\vec{n}}\sum_{\mu,\nu}\sum_{u\neq v\neq w} \textcolor{black}{i \epsilon_{u,v,w}}\sigma_0 \bm{F}^{(1;0,\bar{\nu},\mu)}_{u,v,w} \bar{{S}}^{(\nu,\mu)}(n_1,n_2,n_3)+ \\
        &\quad \quad \quad \quad \sum_{\vec{n}}\sum_{\mu,\nu} \sum_{u\neq v} \sigma_v [ \bm{F}^{(1;0,\nu,\mu)}_{u,u,v}-\bm{F}^{(1;0,\bar{\nu},\bar{\mu})}_{u,v,u}+ \bm{F}^{(1;0,\bar{\nu},\bar{\mu})}_{v,u,u}+\bm{F}^{(1;0,{\nu},{\mu})}_{v,v,v}    
        ]\bar{{S}}^{(\nu,\mu)}(n_1,n_2,n_3) \Big\}
    \end{aligned}
\end{equation}

\begin{equation}
    \begin{aligned}
      \Tr&[\textcolor{black}{\frac{\mathcal{D}^{(3)}_O(T)}{3!}}\sigma_0\sigma_{\gamma}] \\
       &=+4i\sum\limits_{\vec{n}}\sum\limits_{u\neq \gamma}  F^{(1)}_u(n_1)F^{(1)}_u(n_2)F^{(1)}_{\gamma}(n_3) \Big[\bar{{S}}^{\textcolor{black}{(0,1)}}(n_1,n_2,n_3)-\bar{S}^{\textcolor{black}{(0,1)}}(n_1,n_3,n_2)\Big], \\
    \Tr&[\textcolor{black}{\frac{\mathcal{D}^{(3)}_O(T)}{3!}}\sigma_{\gamma}\sigma_{\gamma}] \\
      &= \textcolor{black}{-4 \sum\limits_{\vec{n}}\sum\limits_{u\neq v\neq\gamma}\epsilon_{u\gamma v}F^{(1)}_{u}(n_1)F^{(1)}_{\gamma}(n_2)F^{(1)}_{v}(n_3)\Big[\bar{{S}}^{(0,0)}(n_1,n_2,n_3)-\bar{{S}}^{(1,1)}(n_1,n_2,n_3)\Big];} \\
            &\Tr[\textcolor{black}{\frac{\mathcal{D}^{(3)}_O(T)}{3!}}\sigma_{r}\sigma_{\gamma}] =\\
      &-4 \epsilon (\beta,r,\gamma ) \sum\limits_{\vec{n}} \Big[
      F^{(1)}_r(n_1)F^{(1)}_r(n_2)F^{(1)}_{\beta}(n_3)\bar{{S}}^{(0,0)}(n_1,n_2,n_3) \\
      &\quad \quad\quad \quad\quad \quad   + F^{(1)}_{\beta}(n_1)F^{(1)}_{\beta}(n_2)F^{(1)}_{\beta}(n_3) \bar{{S}}^{(0,0)}(n_1,n_2,n_3)  +F^{(1)}_{\beta}(n_1)F^{(1)}_{\gamma}(n_2)F^{(1)}_{\gamma}(n_3)\bar{{S}}^{(0,0)}(n_1,n_2,n_3) \\
      &\quad \quad \quad \quad\quad \quad         +F^{(1)}_{\beta}(n_1)F^{(1)}_{r}(n_2)F^{(1)}_{r}(n_3)\bar{{S}}^{(1,1)}(n_1,n_2,n_3)
       - F^{(1)}_{r}(n_1)F^{(1)}_{\beta}(n_2)F^{(1)}_{r}(n_3)\bar{{S}}^{(1,1)}(n_1,n_2,n_3)\Big]
 \end{aligned}
\end{equation}

\begin{equation}
\begin{aligned}
    \mathcal{D}_O^{(4)}(t)/4! &=\frac{1}{8} \sum\limits_{\vec{n}} \sum\limits_{\mu,\nu,\xi} \bar{{S}}^{(\xi,\nu,\mu)}(n_1,n_2,n_3,n_4)\Big\{ \Big[\sigma_0\sum\limits_{v \neq w}\big( \bm{F}^{(1;0,\xi,\nu,\mu)}_{w,w,v,v}+ \bm{F}^{(1;0,\xi,\nu,\mu)}_{w,w,w,w} + \bm{F}^{(1;0,\bar{\xi},\bar{\nu},\mu)}_{w,v,v,w} - \bm{F}^{(1;0,\bar{\xi},\bar{\nu},\mu)}_{w,v,w,v})\Big] \\
   + &\sum\limits_{u\neq v \neq w } \sigma_v\sigma_w[-\bm{F}^{(1;0,\bar{\xi},\bar{\nu},\bar{\mu})}_{w,u,u,v} +\bm{F}^{(1;0,\bar{\xi},\nu,\mu)}_{w,u,v,u}
 -\bm{F}^{(1;0,\bar{\xi},\nu,\mu)}_{w,v,u,u}(n_1)+\bm{F}^{(1;0,\bar{\xi},\bar{\nu},\bar{\mu})}_{u,w,u,v}  -\bm{F}^{(1;0,\bar{\xi},\nu,\mu)}_{u,w,v,u}\\
 &\quad\quad\quad\quad\quad -\bm{F}^{(1;0,\xi,\nu,\bar{\mu})}_{u,u,w,v} - \bm{F}^{(;,0,\bar{\xi},\bar{\nu},\bar{\mu})}_{w,v,v,v} +\bm{F}^{(1;0,\bar{\xi},\bar{\nu},\bar{\mu})}_{v,w,v,v}+\bm{F}^{(1;0,\xi,\nu,\bar{\mu})}_{v,v,v,w} -\bm{F}^{(1;0,\xi,\nu,\bar{\mu})}_{v,v,w,v}]
    \Big\}  \\
\end{aligned}
\label{eq:dyson4}
\end{equation}

\begin{equation}
 \begin{aligned}
   \Tr&[\mathcal{D}^{(4)}_O(T)\sigma_{0}\sigma_{\gamma}] \\
   &= \textcolor{black}{4}i \sum\limits_{\vec{n}}\sum\limits_{v\neq w}\epsilon(v,w,\gamma) \Big( F^{(1)}_v(n_1)F^{(1)}_{\gamma}(n_2)[F^{(1)}_{\gamma}(n_3)F^{(1)}_{w}(n_4)- F^{(1)}_{\gamma}(n_4)F^{(1)}_{w}(n_3)] \bar{S}^{(1,0,0)}(n_1,n_2,n_3,n_4)  \\
   &\quad\quad + F^{(1)}_v(n_1)F^{(1)}_w(n_2) F^{(1)}_{\gamma}(n_3) F^{(1)}_{\gamma}(n_4) \bar{S}^{(1,1,1)}(n_1,n_2,n_3,n_4) \\
   &\quad \quad + F^{(1)}_{v}(n_1) F^{(1)}_{v}(n_2)[F^{(1)}_{v}(n_3)F^{(1)}_{w}(n_4)-F^{(1)}_{w}(n_4)F^{(1)}_{v}(n_3)]\bar{S}^{(1,0,0)}(n_1,n_2,n_3,n_4)\\
  & \quad\quad + [F^{(1)}_{v}(n_1)F^{(1)}_{w}(n_2)-F^{(1)}_{w}(n_1)F^{(1)}_{v}(n_2)]F^{(1)}_{v}(n_3)F^{(1)}_{v}(n_4)]\bar{S}^{(1,1,1)}(n_1,n_2,n_3,n_4)
   \Big)
 \end{aligned}    
\end{equation}

\begin{equation}
 \begin{aligned}
   \Tr&[\mathcal{D}^{(4)}_O(T)\sigma_{\gamma}\sigma_{\gamma}] \\
   =& \textcolor{black}{4}\sum\limits_{\vec{n}}\sum\limits_{\{v, w\}\neq \gamma} F^{(1)}_{v}(n_1)F^{(1)}_{v}(n_2)F^{(1)}_{w}(n_3) F^{(1)}_{w}(n_4) \bar{S}^{(0,0,0)}(n_1,n_2,n_3,n_4)  \\
   & +\textcolor{black}{4}\sum\limits_{\vec{n}}\sum\limits_{w\neq\gamma} F^{(1)}_{w}(n_1)F^{(1)}_{\gamma}(n_2)F^{(1)}_{\gamma}(n_3)F^{(1)}_{w}(n_4) \bar{S}^{(0,0,0)}(n_1,n_2,n_3,n_4) \\
   &+\textcolor{black}{4}\sum\limits_{\vec{n}}\sum\limits_{v\neq w\neq \gamma} F^{(1)}_{w}(n_1)F^{(1)}_{v}(n_2) [F^{(1)}_{v}(n_3)F^{(1)}_{w}(n_4)-F^{(1)}_{w}(n_3)F^{(1)}_{v}(n_4)]\bar{S}^{(0,1,1)}(n_1,n_2,n_3,n_4) \\
   & + \textcolor{black}{4}\sum\limits_{\vec{n}}\sum\limits_{w\neq\gamma} F^{(1)}_{w}(n_1)[F^{(1)}_{w}(n_2)F^{(1)}_{\gamma}(n_3)-F^{(1)}_{\gamma}(n_2)F^{(1)}_{w}(n_3)]F^{(1)}_{\gamma}(n_4) \bar{S}^{(1,1,0)}(n_1,n_2,n_3,n_4)
 \end{aligned}    
\end{equation}

\begin{equation}
 \begin{aligned}
   &\Tr[\mathcal{D}^{(4)}_O(T)\sigma_{r}\sigma_{\gamma}] \\
   =& \textcolor{black}{4}\sum\limits_{\vec{n}}
   \Big\{\bar{S}^{(1,1,0)}_{\vec{n}}\resizebox{0.8\linewidth}{!}{$  F_r^{(1)}({n_4}) \left[F_r^{(1)}({n_1}) \left(F_r^{(1)}({n_2}) F_{\gamma }^{(1)}({n_3})-F_{\gamma }^{(1)}({n_2}) F_r^{(1)}({n_3})\right)+F_{\beta }^{(1)}({n_1}) \left(F_{\beta }^{(1)}({n_2}) F_{\gamma }^{(1)}({n_3})-F_{\gamma }^{(1)}({n_2}) F_{\beta }^{(1)}({n_3})\right)\right]$}\\
   -&\bar{S}^{(0,1,1)}_{\vec{n}} \resizebox{0.6\linewidth}{!}{$\left[F_{\beta }^{(1)}({n_1}) F_r^{(1)}({n_2})-F_r^{(1)}({n_1}) F_{\beta }^{(1)}({n_2})\right] \left(F_{\gamma }^{(1)}({n_3}) F_{\beta }^{(1)}({n_4})-F_{\beta }^{(1)}({n_3}) F_{\gamma }^{(1)}({n_4})\right)$}\\
   +& \bar{S}^{(0,0,0)}_{\vec{n}}\resizebox{0.9\linewidth}{!}{ $ \left[F_r^{(1)}({n_3}) \left(F_{\beta }^{(1)}({n_1}) \left(F_{\gamma }^{(1)}({n_2}) F_{\beta }^{(1)}({n_4})-F_{\beta }^{(1)}({n_2}) F_{\gamma }^{(1)}({n_4})\right)-F_r^{(1)}({n_1}) F_r^{(1)}({n_2}) F_{\gamma }^{(1)}({n_4})\right)-F_r^{(1)}({n_1}) F_{\gamma }^{(1)}({n_2}) \left(F_{\beta }^{(1)}({n_3}) F_{\beta }^{(1)}({n_4})+F_{\gamma }^{(1)}({n_3}) F_{\gamma }^{(1)}({n_4})\right)\right]$}
   \Big\},
 \end{aligned}    
\end{equation}
where set $\{\gamma,r,\beta\}=\{x,y,z\}$.

\section{Principles of  QNS protocols design}
\label{app:qns_ctrl_design}
We delineate the principles underlying the design of the CA QNS protocols employed in our study and discuss potential challenges when scaling to larger systems.

Despite the ``deconvolution'' between control and noise in the CA framework, crafting an appropriate control $F^{(1)} $ to establish a solvable linear system of CA spectra is far from trivial. At non-Gaussian orders, $F^{(1)} $ (it represents $F^{(1)}_{u}(n)~\forall u,\forall n\leq L$, for simplicity) forms a higher-order, non-linear algebraic structure in $Q$ dynamics, rendering the simple analytical designs, such as those used in Gaussian cases \cite{Teerawat_PRXQ_Frame}, ineffective. It becomes imperative to conduct a numerical search for different  $F^{(1)} $ configurations to adequately shape the linear system of the CA spectra. To this end, we prepare $m$ sets of  $F^{(1)} $ sequences (each sequence has length $L$), where $m\gg N[\mathcal{C}_{\text{W}}]$, which equals the number of independent/learnable CA spectra. We then numerically search the first $N_{\mathcal{C}_{\text{W}}}$ independent subsets out of $m$ to form the $N_{\mathcal{C}_{\text{W}}}$-by-$N_{\mathcal{C}_{\text{W}}}$ matrix that defines our CA QNS protocol.

It is noteworthy that for higher orders (\( k>4 \)), the resulting matrices from these linear systems tend to be poorly conditioned, a problem that exacerbates with increasing $k$. This poses practical challenges, making the QNS protocol prone to input errors (e.g., measurement errors). To address this, we implement a ``random shuffle'' strategy aimed at optimizing the condition number. In this strategy, the originally prepared $m$ sets are replicated $M\gg1$ times. For each replica, the order of different sequences within the $ m$ sets is randomly shuffled. Thus, the first $N_{\mathcal{C}_{\text{W}}}$ independent sets of these $M$ copies obtained by the search solver differ from each other. We select the independent subset with the optimal condition number as the protocol of CA QNS.

While this method is effective in our specific cases, it reveals significant limitations as $L$, $Q$, and $K$ increase, necessitating a much larger linear system to learn more spectra. The corresponding $M$ and $m$ should be larger for sufficient shuffle.  Despite the availability of other condition number optimization techniques, a critical caveat in the application of QNS—and indeed, all non-perturbative QNS protocols—emerges:  As the number of CA spectra $N_{\mathcal{C}_{\text{W}}}$ grows exponentially with the complexity of control, no optimization can circumvent the fact that the  QNS control searched is essential a random control. The qubit random control ultimately hampers the effectiveness of QNS due to the concentration of expectation values, because random unitary control forms unitary 1-design Haar distribution~\cite{McClean_BP_NatCom_2018}, likely representing a fundamental bottleneck for most, if not all, non-Gaussian QNS protocols.

\section{Proof of swap symmetry in cross-spectra}
\label{app:swap_symmetry}
Here we provide some algebraic details of the swap symmetry of classical (and quantum) digital cross-spectra, which is discussed in Sec.~\ref{sec:Symmetry_analysis}. We formulate the claim of classical cross-spectra  in a Proposition as follows, and the proof also covers the quantum case.

\begin{prop}
For a classical digital spectrum containing a 2-streak in  the $k$-th order $n$-string (i.e, $\exists~n_{1}\geq n_i=n_{i+1} \geq n_{k}, ~1\leq i\leq k-1 $),  swapping the index-corresponding $q$ values ($q_i\leftrightarrow q_{i+1}$) gives another admissible spectrum for $q_i\neq q_{i+1}$ and $q\in Q$. These two spectra are filtered by equivalent control tensors, thus allowing us to define a \textit{swap symmetry} ($q_i\leftrightarrow q_{i+1}$) of spectra. It is impossible, and there is no need, to differentiate these two spectra.
\end{prop}
\begin{proof}
 Recall the CA Dyson term
\begin{equation}
\begin{aligned}
 \mathcal{D}^{(k)}_O(T)/k! =(-i)^k \sum_{\vec{n},\vec{q}}\Bigg\{ \sum_{\vec{\mu}_{[k-1]}}  \sum_{l,\pi} (-1)^{\bar{f}^{(k)}_{\pi}(\vec{\mu})} \Big[  \prod_{j=1}^l \sum\limits_{u_j} F^{(1)}_{q_j,u_j} (n_{\pi(j)})\bar{\Lambda}_{u_j} \Big] \Big[    \prod_{j'=l+1}^{k} \sum\limits_{u_j'} F^{(1)}_{q_{j'},u_{j'}}(n_{\pi(j')}) \Lambda_{u_{j'}} \Big] \Bigg\}     \bar{S}^{\vec{\mu}_{[k-1]}}_{ \vec{q}_{[k]}} (\vec{n}); 
       \end{aligned}
\end{equation}
which includes two brackets $[\cdots]$ in the control tensor. For a 2-streak in the $n$-string, $n_i=n_{i+1} =\tilde{{n}}$, we know that both $q_i=q_{\circ}\neq q_{i+1}= q_{\bullet}$ and $q_i=q_{\bullet}\neq q_{i+1}= q_{\circ}$ are  admissible situations in Dyson series, where $q_{\circ}, q_{\bullet} \in Q$. These two situations are dual and are transformed to each other under swap $q_{\circ}, q_{\bullet}$.
We focus on situation where $n_i=n_{i+1}= \tilde{n}$ and $q_i=q_{\circ}\neq q_{i+1}= q_{\bullet}$. They appear in the Dyson series as (1) they both reside in the left bracket, (2) they are in two brackets separately, and (3) they both reside in the right bracket.  Notice that $k>1$ Dyson term always includes these three cases under all possible permutations $\{\pi\}$, and the validation of swap symmetry thrives on satisfying in all cases. 
In case (1), the above Dyson series includes $\tilde{n},q_{\circ,\bullet}$ as
\begin{equation}
\begin{aligned}
 \mathcal{D}^{(k)}_O(T)/k!  \ni (-i)^k \overline{\sum_{\vec{n},\vec{q}}}\Bigg\{&\sum_{\vec{\mu}_{[k-1]}}  \sum_{l} \overline{\sum_{\pi}} (-1)^{\bar{f}^{(k)}_{\pi}(\vec{\mu})} \Big[ (....) (\sum_u F^{(1)}_{q_{\bullet},u} (\tilde{n})\bar{\Lambda}_{u} )(\sum_v F^{(1)}_{q_{\circ},v} (\tilde{n})\bar{\Lambda}_{v} )  (....)\Big] \\
 \times& \Big[  \prod_{j'=l+1}^{k} \sum\limits_{u_j'} F^{(1)}_{q_{j'},u_{j'}}(n_{\pi(j)}) \Lambda^{[q_{j'}]}_{u_{j'}} \Big] \Bigg\}     \bar{S}^{\vec{\mu}_{[k-1]}}_{ q_1 ..q_{\bullet}q_{\circ}..q_k} (n_1,..,\tilde{n},\tilde{n},..,n_k); 
       \end{aligned}
\end{equation}
where the $\bar{\sum}$ represents reduced summation after fixing 2-streak $\tilde{n}$, $q_i =q_{\circ},~q_{i+1}=q_{\bullet}$, and $(....)$ indicates the remainders  that adapt such a fixing. We recall that  $\sum_l\bar{\sum}_{\pi}$ correspond to $\sum_{\pi\sim \bm{\pi}_Q}$, where $\bm{\pi}_Q|_{(i,i+1)}=(1,1)$. Evidently, this is because the specific allocations of $n_i,n_{i+1}$ fix  the 2-digit snippet of $\bm{\pi}_Q$ at corresponding locations. 
Recall that the toggled system Hamiltonian in the CA picture is $\bar{h}_{q}(n)=\sum_u F^{(1)}_{q,u}(n) \bar{\Lambda}_{u} =-\tilde{O}^{-1}(T) U^{\dagger}_0(n\tau) \sigma_z^{[q]} U_0(n\tau) \tilde{O}(T)$, where $U_0(n\tau)$ represents control propagator at $t=n\tau$.  It is obvious that $[\sigma^{[q_{\circ}]}_z, \sigma^{[q_{\bullet}]}_z]\equiv0$ leads to 
$$[\sum_u F^{(1)}_{q_{\bullet},u} (\tilde{n})\bar{\Lambda}_{u} ,\sum_v F^{(1)}_{q_{\circ},v} (\tilde{n})\bar{\Lambda}_{v} ]=0.$$
Therefore if we swap  $q_{\circ}$ and $q_{\bullet}$ by letting $q_{i+1}=q_{\circ}$ and $q_i=q_{\bullet}$, we will have another spectrum $\bar{S}^{\vec{\mu}_{[k-1]}}_{ q_1 ..q_{\circ}q_{\bullet}..q_k} (n_1,..,\tilde{n},\tilde{n},..,n_k)$ in dual situation that is filtered by the same control tensor. 
This analysis applies similarly in case (3) as well. Case (2) is disparate as 
\begin{equation}
\begin{aligned}
 \mathcal{D}^{(k)}_O(T)/k! \ni (-i)^k \overline{\sum_{\vec{n},\vec{q}}}\Bigg\{&\sum_{\vec{\mu}_{[k-1]}}  \sum_{l} \overline{\sum_{\pi}} (-1)^{\bar{f}^{(k)}_{\pi}(\vec{\mu})} \Big[ (....) (\sum_u F^{(1)}_{q_i=q_{\circ},u} (\tilde{n})\bar{\Lambda}_{u} )   (....)\Big] \Big[  (....) (\sum_v F^{(1)}_{q_{i+1}=q_{\bullet},v} (\tilde{n}){\Lambda}_{v} )  (....) \Big] \Bigg\}\\
 &     \bar{S}^{\vec{\mu}_{[k-1]}}_{ q_1 ..q_{\circ}q_{\bullet}..q_k} (n_1,..,\tilde{n},\tilde{n},..,n_k) \\
 + (-i)^k \overline{\sum_{\vec{n},\vec{q}}}\Bigg\{&\sum_{\vec{\mu}_{[k-1]}}  \sum_{l} \overline{\sum_{\pi}} (-1)^{\bar{f}^{(k)}_{\pi}(\vec{\mu})} \Big[ (....) (\sum_u F^{(1)}_{q_{i+1}=q_{\bullet},u} (\tilde{n})\bar{\Lambda}_{u} )   (....)\Big] \Big[  (....) (\sum_v F^{(1)}_{q_{i}=q_{\circ},v} (\tilde{n}){\Lambda}_{v} )  (....) \Big] \Bigg\}\\
 &     \bar{S}^{\vec{\mu}_{[k-1]}}_{ q_1 ..q_{\circ}q_{\bullet}..q_k} (n_1,..,\tilde{n},\tilde{n},..,n_k)
       \end{aligned}
\end{equation}
where the first two lines and the second two lines represent $\bm{\pi}_Q|_{(i,i+1)}=(1,0)$ and $\bm{\pi_Q}|_{(i,i+1)}=(0,1)$, respectively, The two remaining $(k-2)$-digit snippets of $\bm{\pi}_Q$ {are the same because all $(\cdots)$  in these two lines are correspondingly identical}. Our hope is that the entire control tensor is invariant under $q$-swap (by letting $q_i=q_{\bullet},~q_{i+1}=q_{\circ}$), while the different residences of our pairs disable the commutation trick used in other cases.  Imagine we make such a swap, 
then the new first line is generally different from the old third line because of $\bar{f}^{(k)}_{\pi}(\vec{\mu})  \equiv \bm{\mu}\cdot \bm{\pi}_E$, which have different $\bm{\pi}_Q$ snippet at $(i,i+1)$. 
To eliminate the discrepancy, we need 
\begin{equation}
    \begin{aligned}
    \bm{\mu}\cdot \bm{\pi}_E|_{\bm{\pi}_Q|_{(i,i+1)}=(1,0)}\mod 2~\equiv~  \bm{\mu}\cdot \bm{\pi}_E|_{\bm{\pi}_Q|_{(i,i+1)}=(0,1)} \mod 2
    \end{aligned}
\end{equation}
to be true,  which amounts to 
\begin{equation}
    \begin{aligned}
        &\bm{\mu}|_{(i-1)} =\bm{\mu}|_{(i)},\quad k-1\geq i>1, \\
        &\bm{\mu}|_{(i)} =(0),\quad\quad  i=1. \\
    \end{aligned}
\end{equation}
The above are the necessary and sufficient conditions of $(k-1)$-length $\bm{\mu}$ such that control tensor is invariant. The expression only suggests  ``local'' constraints, while the remaining signs on the total $\bm{\mu}$ are left arbitrary. Thus  ``global'' $\bm{\mu} \equiv \bm{0}$ is a sufficient condition that leads to swap symmetry in classical cross-spectra. 
\end{proof}

\begin{figure*}
\includegraphics[width=1\textwidth]{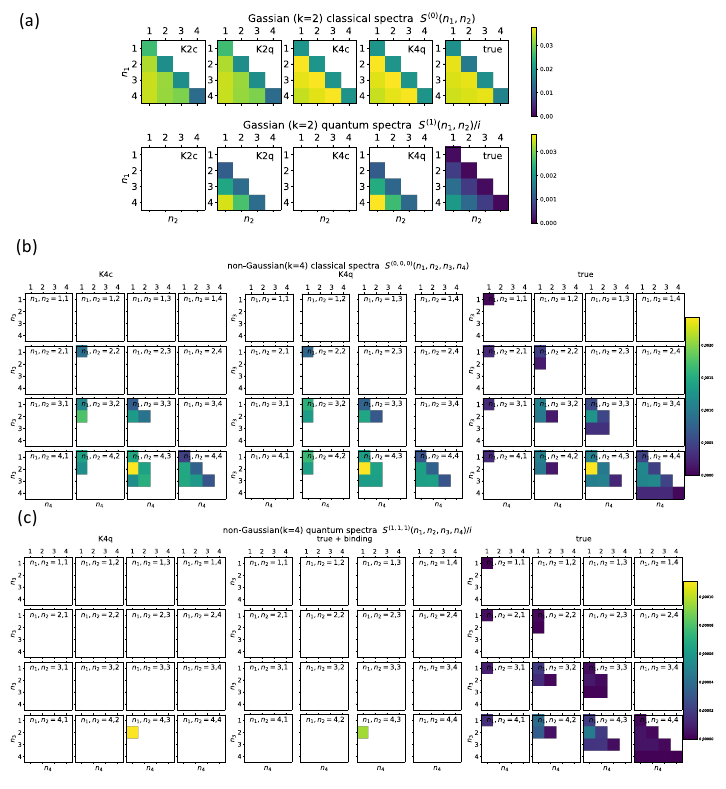}
\caption{\textbf{Reconstruction of the $L=4$ CA spectra  with different QNS protocols } 
    (a)  Gaussian $k=2$ CA classical and quantum spectral reconstruction by four protocols. Diagonal quantum spectra are dark and are non-learnable.
    (b)  Non-Gaussian $k=4$ CA classical $\bm{\mu}=(0,0,0)$ spectral reconstruction by $K4c$ and $K4q$ and true solution.
    (c)  Non-Gaussian ($k=4$) quantum spectra ($\mu=(1,1,1)$) reconstructed by  $K4q$ (left) and  true solution (right). Only one spectrum, $\bar{S}^{(1,1,1)}(4,3,2,1)$, is reconstructed, while others are either dark or bound.  The significant mismatch is mainly due to the binding symmetry presented $\bar{S}^{(1,1,1)}(4,3,2,1)$ with higher order spectra out of truncation order [$\bar{S}^{(0,0,1,1,1)}(4,4,4,3,2,1)$, $\bar{S}^{(0,0,1,1,1)}(4,3,3,3,2,1)$, $\bar{S}^{(1,0,0,1,1)}(4,3,2,2,1)$ and $\bar{S}^{(1,1,0,0,1)}(4,3,2,1,1,1)$ etc.]. We account for these spectra and show the bound form in the middle (true+binding), thus the reconstruction error is reduced. 
    In the numerical simulation of the environment, we use $\gamma =0.02$ MHz, $g=0.1$ MHz,  $T=4~\mu$s and  $\tilde{T}=5~\mu$s. }
    \label{fig:fig_case0_spectra_ALL}
\end{figure*} 

\section{Case Study 0: Comparing Gaussian and non-Gaussian, classical and quantum CA QNS protocols}

\label{app:incremental_QNS}
In this case study, we compare four CA QNS protocols designed to reconstruct single-qubit Gaussian and non-Gaussian, classical and quantum digital-CA spectra. This study extends beyond the Gaussian and non-Gaussian CA QNS protocols denoted by $K2c$ and $K4c$, which reconstruct purely classical noise for digital control up to four windows ($L=4$)~\cite{DongWZ_APL}. The extension includes two new protocols, $K2q$ and $K4q$, which also accommodate quantum noise spectral estimation. Specifically, $K2q$ is tailored to reconstruct Gaussian classical and quantum spectra, while $K4q$ extends to  classical and quantum spectra up to non-Gaussian order $k=4$. Notice that none of these four protocols are fundamental QNS that is employed in two case studies in the main text.

Inspired by~\cite{Qasim_SPAM_paper}, we study the quantum noise in a toy model using an auxiliary qubit as the environment, which couples to our system qubit as 
\begin{equation} 
H_{QE}(t) = g \sigma_z \otimes [ \beta_x(t)\tau_x  + \beta_y(t) \tau_y ], 
\end{equation}
where $g$ represents the coupling strength, $\beta_{x/y}(t)$ denote classical stochastic processes in the defined directions, and $\tau$ is the Pauli operator for the environment qubit. The bath operator $B(t) = \beta_x(t)\tau_x  + \beta_y(t) \tau_y $ is constructed by relating these two stochastic functions  to a main stochastic function: 
\begin{equation} 
\begin{aligned} 
\beta_x(t) &= \beta(t), \\
\beta_y(t) &= \beta(t+\tilde{T}), 
\end{aligned} 
\end{equation} 
where $\tilde{T}> 0$ is chosen such that $[B(t),B(t')]\neq 0$ for $t\neq t'$, and  $t \in [0, \infty)$.

We choose  $\beta(t)$ as a standard random telegraph noise (RTN) for concreteness.  The RTN process is a bi-stable signal that switches randomly between two levels, $\pm1$~\cite{Klyatskin2005,Galperin2004}. The signal at any time $t$ is described by 
\begin{equation*}
    \beta(t) =  \beta(0)(-1)^{n(0,t)},
\end{equation*}
where $\beta(t)=\pm1$ depends on the number of switches $n(0,t)$ within the interval. For a symmetric RTN, the switching rate $\gamma$ is the same for transitions $-1\rightarrow 1$ and $1\rightarrow -1$, with $n(0,t)$ following a Poisson distribution:
\begin{equation*}
    P_{n(0,t)=n} = \frac{(\gamma t)^n}{n!} e^{-\gamma t}.
\end{equation*}
For an initial state $\beta(0)=\pm1$ with equal probability, the RTN is zero-mean ($\expect{\beta(t)}_{\text{C}} \equiv 0$), and the second-order moment simplifies to
\begin{equation*}
    \expect{\beta(t_1)\beta(t_2)}_{\text{C}} = e^{-2\gamma (t_1-t_2)},
\end{equation*}
where $t_1 \geq t_2$.  Higher-moment functions for $t_1 \geq t_2 \geq \ldots \geq t_n$ obey a recursive relation, 
leading to the $k$-th order moment for zero-mean RTN as:
\begin{equation*}
    \Big\langle{\prod^k_{j=1}\beta(t_j)}\Big\rangle_{\text{C}} = e^{2\gamma \sum^{k}_{j=1} (-1)^j t_j},
\end{equation*}
for even $k$; while odd-order moments strictly vanish.

The above modeling of the environment qubit produces a non-vanishing, non-Gaussian, non-classical noise. 
By letting the initial bath qubit state be $\rho_B(0)=(\sigma_0+\sigma_z)/2$, we obtain the classical and quantum noise correlators as
\begin{equation}
    \begin{aligned}
        \expect{\mathcal{B}^{(0)}(t_1,t_2)} = \expect{\{B(t_1),B(t_2)\}}/2  &=[\expect{\beta(t_1)\beta(t_2)}_{\text{C}}+\expect{\beta(t_1+\tilde{T})\beta(t_2+\tilde{T})}_{\text{C}}]\\
        & = 2\expect{\beta(t_1-t_2)\beta(0)}_{\text{C}},\\        
        \expect{\mathcal{B}^{(1)}(t_1,t_2)} =\expect{[B(t_1),B(t_2)]}/2  &=i [\expect{\beta(|t_1-t_2-\tilde{T}|)\beta(0)}_{\text{C}}-\expect{\beta(|t_1+\tilde{T}-t_2|)\beta(0)}_{\text{C}}].
    \end{aligned}
\end{equation}
On the non-Gaussian level, we can get, for example,
\begin{equation}
    \begin{aligned}
       \expect{\mathcal{B}^{(0,0,0)}(t_1,t_2,t_3,t_4)} = \expect{\{\{\{B(t_1),B(t_2)\},B(t_3)\},B(t_4)\}}/8  &= [\expect{\beta_x(t_1)\beta_x(t_2)\beta_x(t_3)\beta_x(t_4)+\beta_y(t_1)\beta_y(t_2)\beta_x(t_3)\beta_x(t_4), \\
        &\quad+\beta_x(t_1)\beta_x(t_2)\beta_y(t_3)\beta_y(t_4)+\beta_y(t_1)\beta_x(t_2)\beta_x(t_3)\beta_y(t_4)}_{\text{C}}] \\
     \expect{ \mathcal{B}^{(1,1,1)}(t_1,t_2,t_3,t_4) }= \expect{[[[ B(t_1),B(t_2)],B(t_3)],B(t_4)]}/8 &= -i[ 
\expect{-\beta_y(t_1)\beta_x(t_2)\beta_x(t_3)\beta_x(t_4)+\beta_x(t_1)\beta_y(t_2)\beta_x(t_3)\beta_x(t_4) \\
&\quad -\beta_y(t_1)\beta_x(t_2)\beta_y(t_3)\beta_y(t_4)+\beta_x(t_1)\beta_y(t_2)\beta_y(t_3)\beta_y(t_4)}_{\text{C}}],\\
    \end{aligned}
\end{equation}
and it is easy to check that the nested commutator does not vanish for $\tilde{T}\neq0$. 

The above-mentioned non-Gaussian,  quantum noise toy model has a Lorentzian profile classical spectral distribution with a peak at zero frequency. The peak frequency can be modulated by incorporating a modulation parameter $\Omega$ within $\beta(t)$, as described in \cite{DongWZ_APL} in more detail. By incorporating this modulation parameter, the peak is shifted to $\omega = \Omega$ without affecting other aspects of the spectral profile.  

We implement four CA QNS protocols ($K2c$, $K2q$, $K4c$, and $K4q$), with the protocol details provided in Table~\ref{tab:QNS_ctrl_other3} and \ref{tab:QNS_ctrl_k4q}, to reconstruct the spectra of noise parameterized by $\gamma = 0.02$ MHz, $g = 0.1$ MHz, $T = 4~\mu$s,  $\tilde{T} = 5~\mu$s,  and $\Omega/\gamma=0$. The spectral reconstruction results are shown in Fig.~\ref{fig:fig_case0_spectra_ALL}. Since the quantum spectra $\bar{S}^{(1)}(\vec{n})$ and $\bar{S}^{(1,1,1)}(\vec{n})$ are purely imaginary, we only present their imaginary components.
In Fig.\ref{fig:fig_case0_spectra_ALL}(a), although all protocols successfully reconstruct the Gaussian classical spectra, only $K2q$ and $K4q$ can reconstruct the Gaussian quantum spectra. The dark quantum spectra (diagonal plaquettes) are irrelevant to $Q$ dynamics and are inherently non-learnable. In Fig.\ref{fig:fig_case0_spectra_ALL}(b), both $K4c$ and $K4q$ reconstruct the non-Gaussian classical spectra equally well. However, some classical spectra cannot be reconstructed due to their binding symmetry.
Figure~\ref{fig:fig_case0_spectra_ALL}(c) illustrates how $K4q$ reconstructs the non-Gaussian quantum spectrum $\bar{S}^{(1,1,1)}(4,3,2,1)$. Crucially, all other $\vec{n}$ values correspond to dark spectra, highlighting the advantage of model reduction. The spectral reconstruction exhibits significant deviations from the true values, primarily because the binding symmetry with higher-order ($k=6$) spectra plays a substantial role for this type of spectra. To verify this, we account for the bound form and present the binding fixed true values in the middle diagram, where the deviations of estimation are significantly reduced. 
It is important to note that all protocols are subject to reconstruction errors due to their insufficient truncation.

We further analyze the performance of noise-tailored, optimized gate designs using the reconstructed spectra obtained from the four protocols and plug-and-play CPMG.  These four noise-tailored designs follow the procedure in Cast Study 2 in Sec.\ref{sec:Fundamental_QNS}, and in Ref.~\cite{DongWZ_APL}. We tested the optimized control under combinations of $\Omega/ \gamma$, $\tilde{T}$, and $g/\gamma$, which modulate the ``color'', quantumness, and non-Gaussianity of the environment, respectively. The two plots in  Fig.~\ref{fig:fid_1qb_CC} show the performance of four protocols alongside  CPMG. The interplay among protocols proved to be complex and highly dependent on parameter regimes.
Our findings can be summarized as follows:
(i) The truncation order significantly affects the performance of optimized control.  As $g$ increases,  higher-order spectra enter the qubit dynamics unaccounted for by these protocols. These spectra are projected into the reconstructed spectra, leading to large deviations.  The optimized control, based on imprecise spectra, diverges unpredictably. For example, in Fig.~\ref{fig:fid_1qb_CC}(a), even though $K4q$ is expected to outperform the other three learn-based protocols because of its more precise spectral estimation,  we argue that, especially in the strong coupling regime, spectral deviation in $K4q$ is more adversarial to designing a robust control than in other protocols. However, $K4q$ has the best performance in Fig.~\ref{fig:fid_1qb_CC}(b). This indicates that it is not always higher-order truncation that is better; it is, rather, dependent on the noise parameters. (ii) Fig.~\ref{fig:fid_1qb_CC}(b) shows that CPMG loses its effectiveness at non-zero noise frequencies ($\Omega/\gamma>0$), where learning noise becomes crucial. This result supports the observation that CPMG  filters well low-frequency noise in Ref.\cite{DongWZ_APL}. (iii) We set $\tilde{T} = 5~\mu$s to produce non-vanishing quantum spectra, whose magnitude is approximately one order smaller than classical spectra. At this value, combined with the ``color'' parameter $\Omega/\gamma=80$, CPMG performance declines, while $K4q$ outperforms all other protocols. Adjusting $\tilde{T}$ could amplify the relative quantum spectrum magnitude, potentially enhancing $K4q$'s advantage over $K4c$. However, this adjustment may also improve CPMG performance to surpass $K4q$ since quantum spectra remain dark under CPMG control (not shown).
In summary, for general non-Gaussian, non-classical noise C\&C, QNS-based optimized control with insufficient truncation exhibits a complex interplay with plug-and-play CPMG performance. The best strategy for robust control design is strongly dependent on the noise parameter regime.

\begin{table*}[!htbp]
\centering
\resizebox{0.7\textwidth}{!}{\begin{tabular}{|cccc | c|| cccc|c| }
\hline$n=1$ & $n=2$ & $n=3$ & $n=4$ & $\widetilde{\rho}_Q$ & $n=1$ & $n=2$ & $n=3$ & $n=4$ & $\widetilde{\rho}_Q$ \\
\hline$\{1,0,0\}$ & $\{1,0,0\}$ & $\{1,0,0\}$ & $\{1,0,0\}$ & $\sigma_z$ &
$\{1,0,0\}$ & $\{1,0,0\}$ & $\{1,0,0\}$ & $\{1,0,0\}$ & $\sigma_y$ \\
$\{1,0,0\}$ & $\{1,0,0\}$ & $\{1,0,0\}$ & $\{0,1,0\}$ & $\sigma_z$ &
$\{1,0,0\}$ & $\{1,0,0\}$ & $\{1,0,0\}$ & $\{0,1,0\}$ & $\sigma_y$ \\
$\{1,0,0\}$ & $\{1,0,0\}$ & $\{1,0,0\}$ & $\{0,0,1\}$ & $\sigma_z$ &
$\{1,0,0\}$ & $\{1,0,0\}$ & $\{0,1,0\}$ & $\{1,0,0\}$ & $\sigma_z$ \\
$\{1,0,0\}$ & $\{1,0,0\}$ & $\{0,1,0\}$ & $\{1,0,0\}$ & $\sigma_y$ &
$\{1,0,0\}$ & $\{1,0,0\}$ & $\{0,1,0\}$ & $\{0,1,0\}$ & $\sigma_z$ \\
$\{1,0,0\}$ & $\{1,0,0\}$ & $\{0,0,1\}$ & $\{1,0,0\}$ & $\sigma_z$ &
$\{1,0,0\}$ & $\{0,1,0\}$ & $\{1,0,0\}$ & $\{1,0,0\}$ & $\sigma_z$ \\
$\{1,0,0\}$ & $\{0,1,0\}$ & $\{1,0,0\}$ & $\{1,0,0\}$ & $\sigma_y$ &
$\{1,0,0\}$ & $\{0,1,0\}$ & $\{1,0,0\}$ & $\{0,1,0\}$ & $\sigma_z$ \\
$\{1,0,0\}$ & $\{0,1,0\}$ & $\{0,1,0\}$ & $\{1,0,0\}$ & $\sigma_z$ &
$\{1,0,0\}$ & $\{0,0,1\}$ & $\{1,0,0\}$ & $\{1,0,0\}$ & $\sigma_z$ \\
\hline 
\end{tabular}}
\vspace*{4mm}
\resizebox{0.7\textwidth}{!}{
\begin{tabular}{|cccc | c|| cccc|c| }
\hline$n=1$ & $n=2$ & $n=3$ & $n=4$ & $\widetilde{\rho}_Q$ & $n=1$ & $n=2$ & $n=3$ & $n=4$ & $\widetilde{\rho}_Q$ \\
\hline
$\{1, 0, 0\}$ & $\{1, 0, 0\}$ & $\{1, 0, 0\}$ & $\{1, 0, 0\}$ & $\sigma_{z}$ & $\{1, 0, 0\}$ & $\{1, 0, 0\}$ & $\{1, 0, 0\}$ & $\{1, 0, 0\}$ & $\sigma_y$ \\
$\{1, 0, 0\}$ & $\{1, 0, 0\}$ & $\{1, 0, 0\}$ & $\{0, 1, 0\}$ &$\sigma_{z}$ & $\{1, 0, 0\}$ & $\{1, 0, 0\}$ & $\{1, 0, 0\}$ & $\{0, 1, 0\}$ & $\sigma_{y}$ \\
$\{1, 0, 0\}$ & $\{1, 0, 0\}$ & $\{1, 0, 0\}$ & $\{0, 0, 1\}$&$\sigma_{z}$ & $\{1, 0, 0\}$ & $\{1, 0, 0\}$ & $\{1, 0, 0\}$ & $\{0, 0, 1\}$ & $\sigma_{y}$ \\
$\{1, 0, 0\}$ & $\{1, 0, 0\}$ & $\{0, 1, 0\}$ & $\{1, 0, 0\}$&$\sigma_{z}$ & $\{1, 0, 0\}$ & $\{1, 0, 0\}$ & $\{0, 1, 0\}$ & $\{1, 0, 0\}$ & $\sigma_{y}$ \\
$\{1, 0, 0\}$ & $\{1, 0, 0\}$ & $\{0, 1, 0\}$ & $\{0, 1, 0\}$&$\sigma_{z}$ & $\{1, 0, 0\}$ & $\{1, 0, 0\}$ & $\{0, 1, 0\}$ & $\{0, 1, 0\}$ & $\sigma_{y}$ \\
$\{1, 0, 0\}$ & $\{1, 0, 0\}$ & $\{0, 1, 0\}$ & $\{0, 0, 1\}$&$\sigma_{z}$ & $\{1, 0, 0\}$ & $\{1, 0, 0\}$ & $\{0, 1, 0\}$ & $\{0, 0, 1\}$ & $\sigma_{y}$ \\
$\{1, 0, 0\}$ & $\{1, 0, 0\}$ & $\{0, 1, 0\}$ & $\{\frac{1}{\sqrt{2}}, \frac{1}{\sqrt{2}}, 0\}$&$\sigma_{z}$ & $\{1, 0, 0\}$ & $\{1, 0, 0\}$ & $\{0, 1, 0\}$ & $\{\frac{1}{\sqrt{2}}, \frac{1}{\sqrt{2}}, 0\}$ & $\sigma_{y}$ \\
$\{1, 0, 0\}$ & $\{1, 0, 0\}$ & $\{0, 1, 0\}$ & $\{ 0, \frac{1}{\sqrt{2}}, \frac{1}{\sqrt{2}}\}$&$\sigma_{y}$ & $\{1, 0, 0\}$ & $\{1, 0, 0\}$ & $\{0, 0, 1\}$ & $\{1, 0, 0\}$ & $\sigma_{z}$ \\
$\{1, 0, 0\}$ & $\{1, 0, 0\}$ & $\{0, 0, 1\}$ & $\{1, 0, 0\}$&$\sigma_{y}$ & $\{1, 0, 0\}$ & $\{0, 1, 0\}$ & $\{1, 0, 0\}$ & $\{1, 0, 0\}$ & $\sigma_{z}$ \\
$\{1, 0, 0\}$ & $\{0, 1, 0\}$ & $\{1, 0, 0\}$ & $\{1, 0, 0\}$&$\sigma_{y}$ & $\{1, 0, 0\}$ & $\{0, 1, 0\}$ & $\{1, 0, 0\}$ & $\{0, 1, 0\}$ & $\sigma_{z}$ \\
$\{1, 0, 0\}$ & $\{0, 1, 0\}$ & $\{1, 0, 0\}$ & $\{0, 1, 0\}$&$\sigma_{y}$ & $\{1, 0, 0\}$ & $\{0, 1, 0\}$ & $\{1, 0, 0\}$ & $\{0, 0, 1\}$ & $\sigma_{z}$ \\
$\{1, 0, 0\}$ & $\{0, 1, 0\}$ & $\{1, 0, 0\}$ & $\{0, 0, 1\}$ &$\sigma_{y}$ & $\{1, 0, 0\}$ & $\{0, 1, 0\}$ & $\{1, 0, 0\}$ & $\{\frac{1}{\sqrt{2}}, \frac{1}{\sqrt{2}}, 0\}$ & $\sigma_{z}$ \\
$\{1, 0, 0\}$ & $\{0, 1, 0\}$ & $\{1, 0, 0\}$ & $\{\frac{1}{\sqrt{2}}, \frac{1}{\sqrt{2}}, 0\}$ &$\sigma_{y}$ & $\{1, 0, 0\}$ & $\{0, 1, 0\}$ & $\{1, 0, 0\}$ & $\{\frac{1}{\sqrt{2}}, 0, \frac{1}{\sqrt{2}}\}$ & $\sigma_{y}$ \\
$\{1, 0, 0\}$ & $\{0, 1, 0\}$ & $\{1, 0, 0\}$ & $\{0, \frac{1}{\sqrt{2}}, \frac{1}{\sqrt{2}}\}$ &$\sigma_{y}$ & $\{1, 0, 0\}$ & $\{0, 1, 0\}$ & $\{0, 1, 0\}$ & $\{1, 0, 0\}$ & $\sigma_{z}$ \\
$\{1, 0, 0\}$ & $\{0, 1, 0\}$ & $\{0, 1, 0\}$ & $\{1, 0, 0\}$ &$\sigma_{y}$ & $\{1, 0, 0\}$ & $\{0, 1, 0\}$ & $\{0, 1, 0\}$ & $\{0, 1, 0\}$ & $\sigma_{z}$ \\
$\{1, 0, 0\}$ & $\{0, 1, 0\}$ & $\{0, 1, 0\}$ & $\{0, 1, 0\}$&$\sigma_{y}$ & $\{1, 0, 0\}$ & $\{0, 1, 0\}$ & $\{0, 1, 0\}$ & $\{0, 0, 1\}$ & $\sigma_{z}$ \\
$\{1, 0, 0\}$ & $\{0, 1, 0\}$ & $\{0, 0, 1\}$ & $\{1, 0, 0\}$&$\sigma_{z}$ & $\{1, 0, 0\}$ & $\{0, 1, 0\}$ & $\{0, 0, 1\}$ & $\{1, 0, 0\}$ & $\sigma_{y}$\\
$\{1, 0, 0\}$ & $\{0, 1, 0\}$ & $\{0, 0, 1\}$ & $\{0, 1, 0\}$&$\sigma_{y}$ & $\{1, 0, 0\}$ & $\{0, 1, 0\}$ & $\{\frac{1}{\sqrt{2}}, \frac{1}{\sqrt{2}}, 0\}$ & $\{1, 0, 0\}$ & $\sigma_{z}$\\
$\{1, 0, 0\}$ & $\{0, 1, 0\}$ & $\{\frac{1}{\sqrt{2}}, \frac{1}{\sqrt{2}}, 0\}$ & $\{1, 0, 0\}$&$\sigma_{y}$ & $\{1, 0, 0\}$ & $\{0, 1, 0\}$ & $\{0, \frac{1}{\sqrt{2}}, \frac{1}{\sqrt{2}}\}$ & $\{1, 0, 0\}$  &$\sigma_{y}$\\
$\{1, 0, 0\}$ & $\{0, 0, 1\}$ & $\{1, 0, 0\}$ & $\{1, 0, 0\}$ & $\sigma_{z}$ &  $\{1, 0, 0\}$ & $\{0, 0, 1\}$ & $\{1, 0, 0\}$ & $\{1, 0, 0\}$ &$\sigma_{y}$ \\
 $\{0, 1, 0\}$ & $\{1, 0, 0\}$ & $\{1, 0, 0\}$ & $\{1, 0, 0\}$ & $\sigma_{y}$ & $\{0, 1, 0\}$ & $\{1, 0, 0\}$ & $\{1, 0, 0\}$ & $\{0, 1, 0\}$ & $\sigma_{y}$ \\
 $\{0, 1, 0\}$ & $\{1, 0, 0\}$ & $\{1, 0, 0\}$ & $\{0, \frac{1}{\sqrt{2}}, \frac{1}{\sqrt{2}}\}$ &$\sigma_{y}$ &  $\{0, 1, 0\}$ & $\{1, 0, 0\}$ & $\{0, 1, 0\}$ & $\{1, 0, 0\}$  & $\sigma_{y}$\\
 $\{0, 1, 0\}$ & $\{1, 0, 0\}$ & $\{0, 0, 1\}$ & $\{0, 1, 0\}$ &$\sigma_{y}$ & $\{0, 1, 0\}$ & $\{1, 0, 0\}$ & $\{0, \frac{1}{\sqrt{2}}, \frac{1}{\sqrt{2}}\}$ & $\{1, 0, 0\}$ &$\sigma_{y}$ \\
 $\{0, 1, 0\}$ & $\{0, 0, 1\}$ & $\{1, 0, 0\}$ & $\{0, 1, 0\}$ &$\sigma_{y}$ & $\{0, 1, 0\}$ & $\{0, 0, 1\}$ & $\{0, 1, 0\}$ & $\{1, 0, 0\}$ & $\sigma_{y}$ \\
 $\{0, 1,  0\}$ & $\{0, \frac{1}{\sqrt{2}}, \frac{1}{\sqrt{2}}\}$ & $\{1, 0, 0\}$  & $\{1, 0, 0\}$  &$\sigma_{y}$ &  &   &   &  &  \\
\hline
\end{tabular}}
\vspace*{4mm}
\resizebox{0.7\textwidth}{!}{
\begin{tabular}{|cccc | c|| cccc|c| }
\hline$n=1$ & $n=2$ & $n=3$ & $n=4$ & $\widetilde{\rho}_Q$ & $n=1$ & $n=2$ & $n=3$ & $n=4$ & $\widetilde{\rho}_Q$ \\
\hline 
  $\{1, 0, 0\}$ & $\{1, 0, 0\}$ & $\{1, 0, 0\}$ & $\{1, 0, 0\}$  & $\sigma_z$ &  $\{1, 0, 0\}$ & $\{1, 0, 0\}$ & $\{1, 0, 0\}$ & $\{1, 0, 0\}$ & $\sigma_y$ \\
  $\{1, 0, 0\}$ & $\{1, 0, 0\}$ & $\{1, 0, 0\}$ & $\{0, 1, 0\}$  & $\sigma_0$ &  $\{1, 0, 0\}$ & $\{1, 0, 0\}$ & $\{1, 0, 0\}$ & $\{0, 1, 0\}$ & $\sigma_z$\\
  $\{1, 0, 0\}$ & $\{1, 0, 0\}$ & $\{1, 0, 0\}$ & $\{0, 1, 0\}$  & $\sigma_y$ &  $\{1, 0, 0\}$ & $\{1, 0, 0\}$ & $\{1, 0,  0\}$ & $\{0, 0, 1\}$  & $\sigma_z$\\
  $\{1, 0, 0\}$ & $\{1, 0, 0\}$ & $\{0, 1, 0\}$ & $\{1, 0, 0\}$  & $\sigma_0$ &  $\{1, 0, 0\}$ & $\{1, 0, 0\}$ & $\{0, 1, 0\}$ & $\{1, 0, 0\}$  & $\sigma_z$ \\
  $\{1, 0, 0\}$ & $\{1, 0, 0\}$ & $\{0, 1, 0\}$ & $\{1, 0, 0\}$  & $\sigma_y$ &  $\{1, 0, 0\}$ & $\{1, 0, 0\}$ & $\{0, 1, 0\}$ & $\{0, 1, 0\}$  & $\sigma_z$\\
  $\{1, 0, 0\}$ & $\{1, 0, 0\}$ & $\{0, 1, 0\}$ & $\{0, 0, 1\}$  & $\sigma_0$ &  $\{1, 0, 0\}$ & $\{1, 0, 0\}$ & $\{0, 0, 1\}$ & $\{1, 0, 0\}$ & $\sigma_z$\\
  $\{1, 0, 0\}$ & $\{0, 1, 0\}$ & $\{1, 0, 0\}$ & $\{1, 0, 0\}$  & $\sigma_0$ &  $\{1, 0, 0\}$ & $\{0, 1, 0\}$ & $\{1, 0, 0\}$ & $\{1, 0, 0\}$ & $\sigma_z$  \\
  $\{1, 0, 0\}$ & $\{0, 1, 0\}$ & $\{1, 0, 0\}$ & $\{1, 0, 0\}$  & $\sigma_y$ &  $\{1, 0, 0\}$ & $\{0, 1, 0\}$ & $\{1, 0, 0\}$ & $\{0, 1, 0\}$  & $\sigma_z$\\
  $\{1, 0, 0\}$ & $\{0, 1, 0\}$ & $\{1, 0, 0\}$ & $\{0, 0, 1\}$  & $\sigma_0$ &  $\{1, 0, 0\}$ & $\{0, 1, 0\}$ & $\{0, 1, 0\}$ & $\{1, 0, 0\}$  & $\sigma_z$\\
  $\{1, 0, 0\}$ & $\{0, 1, 0\}$ & $\{0, 0, 1\}$ & $\{1, 0, 0\}$  & $\sigma_0$ &  $\{1, 0, 0\}$ & $\{0, 0, 1\}$ & $\{1, 0, 0\}$ & $\{1, 0, 0\}$ & $\sigma_z$\\
  \hline
\end{tabular}}
\caption{\textbf{Regular CA QNS protocols: $K2c$, $K2q$, $K4c$}  (\textbf{top}) $K2c$ CA QNS protocol. The switching function $\{y_x(n),y_y(n),y_z(n)\}$, $n=1,2,3,4$, and initial states $\widetilde{\rho}_Q$ required for Gaussian (${K}=2$) CA QNS with $L=4$. We have  $N_{\mathcal{C}_{\text{W}}}|^{{K}=2,L=4}=14$. In all experiments, ${O}\equiv\sigma_z$.
(\textbf{middle}) $K4c$ CA QNS. $\mathcal{C}|^{{K}=4,L=4}_{\text{QNS}}$ consists of 49 control vectors $\{y_x(n),y_y(n),y_z(n)\}$ in different windows, $n=1,2,3,4$, and corresponding initial states $\widetilde{\rho}_Q$. (\textbf{bottom}) $K2q$ CA QNS. $N_{\mathcal{C}_{\text{W}}}|^{{K}=4,L=4}_{\text{QNS}}=20$.
In all experiments, ${O}\equiv\sigma_z$.  
All protocols do not {\em a priori} assume stationarity and zero mean.   }
    \label{tab:QNS_ctrl_other3}
\end{table*}

\begin{table*}[!htbp]
    \centering
    \begin{tabular}{|cccc | c|| cccc|c| }
 \hline$n=1$ & $n=2$ & $n=3$ & $n=4$ & $\widetilde{\rho}_Q$ & $n=1$ & $n=2$ & $n=3$ & $n=4$ & $\widetilde{\rho}_Q$ \\ \hline
 $\{1, 0, 0\}$ &  $\{1, 0, 0\}$ &  $\{1, 0, 0\}$ &  $\{1, 0, 0\}$ &  $\sigma_z$  &$\{1, 0,  0\}$ &  $\{1, 0, 0\}$ &  $\{1, 0, 0\}$ &  $\{1, 0, 0\}$ & $\sigma_y$ \\ 
$\{1, 0, 0\}$ &  $\{1, 0, 0\}$ &  $\{1, 0, 0\}$ &  $\{0, 1, 0\}$ &  $\sigma_0$  &$\{1, 0, 0\}$ &  $\{1, 0,  0\}$ &  $\{1, 0, 0\}$ &  $\{0, 1, 0\}$ & $\sigma_z$ \\ 
$\{1, 0, 0\}$ &  $\{1, 0, 0\}$ &  $\{1,  0, 0\}$ &  $\{0, 1, 0\}$ &  $\sigma_y$  &$\{1, 0, 0\}$ &  $\{1, 0, 0\}$ &  $\{1, 0,  0\}$ &  $\{0, 0, 1\}$ & $\sigma_z$ \\ 
$\{1, 0, 0\}$ &  $\{1, 0, 0\}$ &  $\{1, 0, 0\}$ &  $\{0, 0, 1\}$ &  $\sigma_y$  &$\{1, 0, 0\}$ &  $\{1, 0, 0\}$ &  $\{1, 0, 0\}$ &  $\{0.707, 0, 0.707\}$ & $\sigma_0$ \\ 
$\{1, 0, 0\}$ &  $\{1, 0, 0\}$ &  $\{0, 1, 0\}$ &  $\{1, 0, 0\}$ &  $\sigma_0$  &$\{1, 0, 0\}$ &  $\{1, 0, 0\}$ &  $\{0, 1, 0\}$ &  $\{1, 0,  0\}$ & $\sigma_z$ \\ 
$\{1, 0, 0\}$ &  $\{1, 0, 0\}$ &  $\{0, 1, 0\}$ &  $\{1, 0, 0\}$ &  $\sigma_y$  &$\{1, 0, 0\}$ &  $\{1, 0, 0\}$ &  $\{0, 1, 0\}$ &  $\{0, 1,  0\}$ & $\sigma_0$ \\ 
$\{1, 0, 0\}$ &  $\{1, 0, 0\}$ &  $\{0, 1, 0\}$ &  $\{0, 1, 0\}$ &  $\sigma_z$  &$\{1, 0, 0\}$ &  $\{1, 0, 0\}$ &  $\{0, 1, 0\}$ &  $\{0, 1, 0\}$ & $\sigma_y$ \\ 
$\{1, 0, 0\}$ &  $\{1, 0, 0\}$ &  $\{0, 1, 0\}$ &  $\{0, 0, 1\}$ &  $\sigma_0$  &$\{1, 0, 0\}$ &  $\{1, 0, 0\}$ &  $\{0, 1, 0\}$ &  $\{0, 0, 1\}$ & $\sigma_z$ \\ 
$\{1, 0, 0\}$ &  $\{1, 0, 0\}$ &  $\{0, 1, 0\}$ &  $\{0, 0, 1\}$ &  $\sigma_y$  &$\{1, 0, 0\}$ &  $\{1, 0, 0\}$ &  $\{0, 1, 0\}$ &  $\{0.707, 0.707, 0\}$ & $\sigma_0$ \\ 
$\{1, 0, 0\}$ &  $\{1, 0, 0\}$ &  $\{0, 1, 0\}$ &  $\{0.707, 0.707, 0\}$ &  $\sigma_z$  &$\{1, 0, 0\}$ &  $\{1, 0, 0\}$ &  $\{0, 1, 0\}$ &  $\{0.707, 0.707, 0\}$ & $\sigma_y$ \\ 
$\{1, 0, 0\}$ &  $\{1, 0, 0\}$ &  $\{0, 1, 0\}$ &  $\{0.707, 0, 0.707\}$ &  $\sigma_0$  &$\{1, 0, 0\}$ &  $\{1, 0, 0\}$ &  $\{0, 1, 0\}$ &  $\{0, 0.707, 0.707\}$ & $\sigma_y$ \\ 
$\{1, 0, 0\}$ &  $\{1, 0, 0\}$ &  $\{0, 0, 1\}$ &  $\{1, 0, 0\}$ &  $\sigma_0$  &$\{1, 0, 0\}$ &  $\{1, 0, 0\}$ &  $\{0, 0, 1\}$ &  $\{1, 0, 0\}$ & $\sigma_z$ \\ 
$\{1, 0, 0\}$ &  $\{1, 0, 0\}$ &  $\{0, 0, 1\}$ &  $\{1, 0, 0\}$ &  $\sigma_y$  &$\{1, 0, 0\}$ &  $\{1, 0,  0\}$ &  $\{0.707, 0, 0.707\}$ &  $\{1, 0, 0\}$ & $\sigma_0$ \\ 
$\{1, 0, 0\}$ &  $\{0, 1, 0\}$ &  $\{1, 0, 0\}$ &  $\{1, 0, 0\}$ &  $\sigma_0$  &$\{1, 0, 0\}$ &  $\{0, 1,  0\}$ &  $\{1, 0, 0\}$ &  $\{1, 0, 0\}$ & $\sigma_z$ \\ 
$\{1, 0, 0\}$ &  $\{0, 1, 0\}$ &  $\{1, 0, 0\}$ &  $\{1, 0, 0\}$ &  $\sigma_y$  &$\{1, 0, 0\}$ &  $\{0, 1, 0\}$ &  $\{1, 0,0\}$ &  $\{0, 1, 0\}$ & $\sigma_0$ \\ 
$\{1, 0, 0\}$ &  $\{0, 1, 0\}$ &  $\{1, 0, 0\}$ &  $\{0,1, 0\}$ &  $\sigma_z$  &$\{1, 0, 0\}$ &  $\{0, 1, 0\}$ &  $\{1, 0, 0\}$ &  $\{0, 1, 0\}$ & $\sigma_y$ \\ 
$\{1, 0, 0\}$ &  $\{0, 1, 0\}$ &  $\{1, 0, 0\}$ &  $\{0, 0,  1\}$ &  $\sigma_0$  &$\{1, 0, 0\}$ &  $\{0, 1, 0\}$ &  $\{1, 0, 0\}$ &  $\{0, 0, 1\}$ & $\sigma_z$ \\ 
$\{1, 0, 0\}$ &  $\{0, 1, 0\}$ &  $\{1, 0, 0\}$ &  $\{0, 0, 1\}$ &  $\sigma_y$  &$\{1, 0, 0\}$ &  $\{0, 1, 0\}$ &  $\{1, 0, 0\}$ &  $\{0.707, 0.707, 0\}$ & $\sigma_0$ \\ 
$\{1, 0, 0\}$ &  $\{0, 1, 0\}$ &  $\{1, 0,  0\}$ &  $\{0.707, 0.707, 0\}$ &  $\sigma_z$  &$\{1, 0, 0\}$ &  $\{0, 1, 0\}$ &  $\{1, 0, 0\}$ &  $\{0.707, 0.707, 0\}$ & $\sigma_y$ \\ 
$\{1, 0, 0\}$ &  $\{0, 1, 0\}$ &  $\{1, 0, 0\}$ &  $\{0.707, 0, 0.707\}$ &  $\sigma_0$  &$\{1, 0, 0\}$ &  $\{0, 1, 0\}$ &  $\{1, 0, 0\}$ &  $\{0.707, 0, 0.707\}$ & $\sigma_y$ \\ 
$\{1, 0, 0\}$ &  $\{0, 1, 0\}$ &  $\{1, 0, 0\}$ &  $\{0, 0.707, 0.707\}$ &  $\sigma_y$  &$\{1, 0, 0\}$ &  $\{0, 1, 0\}$ &  $\{0, 1, 0\}$ &  $\{1, 0, 0\}$ & $\sigma_0$ \\ 
$\{1, 0, 0\}$ &  $\{0, 1, 0\}$ &  $\{0, 1, 0\}$ &  $\{1, 0, 0\}$ &  $\sigma_z$  &$\{1, 0, 0\}$ &  $\{0, 1, 0\}$ &  $\{0, 1, 0\}$ &  $\{1, 0, 0\}$ & $\sigma_y$ \\ 
$\{1, 0, 0\}$ &  $\{0, 1, 0\}$ &  $\{0, 1, 0\}$ &  $\{0, 1, 0\}$ &  $\sigma_0$  &$\{1, 0, 0\}$ &  $\{0, 1, 0\}$ &  $\{0,1, 0\}$ &  $\{0, 1, 0\}$ & $\sigma_z$ \\ 
$\{1, 0, 0\}$ &  $\{0, 1, 0\}$ &  $\{0, 1, 0\}$ &  $\{0, 1, 0\}$ &  $\sigma_y$  &$\{1, 0, 0\}$ &  $\{0, 1, 0\}$ &  $\{0, 1, 0\}$ &  $\{0,0, 1\}$ & $\sigma_0$ \\ 
$\{1, 0, 0\}$ &  $\{0, 1, 0\}$ &  $\{0, 1, 0\}$ &  $\{0, 0, 1\}$ &  $\sigma_z$  &$\{1, 0, 0\}$ &  $\{0, 1, 0\}$ &  $\{0, 1, 0\}$ &  $\{0, 0, 1\}$ & $\sigma_y$ \\ 
$\{1, 0, 0\}$ &  $\{0, 1, 0\}$ &  $\{0, 1, 0\}$ &  $\{0.707, 0.707, 0\}$ &  $\sigma_z$  &$\{1, 0, 0\}$ &  $\{0, 1, 0\}$ &  $\{0, 0, 1\}$ &  $\{1, 0, 0\}$ & $\sigma_0$ \\ 
$\{1, 0, 0\}$ &  $\{0, 1, 0\}$ &  $\{0, 0, 1\}$ &  $\{1, 0, 0\}$ &  $\sigma_z$  &$\{1, 0, 0\}$ &  $\{0, 1, 0\}$ &  $\{0, 0, 1\}$ &  $\{1, 0, 0\}$ & $\sigma_y$ \\ 
$\{1, 0, 0\}$ &  $\{0, 1, 0\}$ &  $\{0, 0, 1\}$ &  $\{0, 1, 0\}$ &  $\sigma_0$  &$\{1, 0, 0\}$ &  $\{0, 1, 0\}$ &  $\{0, 0, 1\}$ &  $\{0, 1, 0\}$ & $\sigma_y$ \\ 
$\{1, 0, 0\}$ &  $\{0, 1, 0\}$ &  $\{0, 0, 1\}$ &  $\{0, 0, 1\}$ &  $\sigma_0$  &$\{1, 0, 0\}$ &  $\{0, 1, 0\}$ &  $\{0, 0, 1\}$ &  $\{0.707, 0.707, 0\}$ & $\sigma_0$ \\ 
$\{1, 0, 0\}$ &  $\{0, 1, 0\}$ &  $\{0.707, 0.707, 0\}$ &  $\{1, 0, 0\}$ &  $\sigma_0$  &$\{1, 0, 0\}$ &  $\{0, 1, 0\}$ &  $\{0.707,  0.707, 0\}$ &  $\{1, 0, 0\}$ & $\sigma_z$ \\ 
$\{1, 0, 0\}$ &  $\{0, 1, 0\}$ &  $\{0.707, 0.707, 0\}$ &  $\{1, 0, 0\}$ &  $\sigma_y$  &$\{1, 0, 0\}$ &  $\{0, 1, 0\}$ &  $\{0.707, 0.707, 0\}$ &  $\{0, 1, 0\}$ & $\sigma_z$ \\ 
$\{1, 0, 0\}$ &  $\{0, 1, 0\}$ &  $\{0.707, 0, 0.707\}$ &  $\{1, 0, 0\}$ &  $\sigma_0$  &$\{1, 0, 0\}$ &  $\{0, 1, 0\}$ &  $\{0, 0.707, 0.707\}$ &  $\{1, 0, 0\}$ & $\sigma_y$ \\ 
$\{1, 0, 0\}$ &  $\{0, 0, 1\}$ &  $\{1, 0, 0\}$ &  $\{1, 0, 0\}$ &  $\sigma_0$  &$\{1, 0, 0\}$ &  $\{0, 0, 1\}$ &  $\{1, 0, 0\}$ &  $\{1, 0, 0\}$ & $\sigma_z$ \\ 
$\{1, 0, 0\}$ &  $\{0, 0, 1\}$ &  $\{1, 0, 0\}$ &  $\{1, 0, 0\}$ &  $\sigma_y$  &$\{1, 0, 0\}$ &  $\{0, 0, 1\}$ &  $\{1, 0, 0\}$ &  $\{0, 1, 0\}$ & $\sigma_0$ \\ 
$\{1, 0, 0\}$ &  $\{0.707, 0, 0.707\}$ &  $\{1, 0, 0\}$ &  $\{1, 0, 0\}$ &  $\sigma_0$  &$\{0, 1, 0\}$ &  $\{1, 0, 0\}$ &  $\{1, 0, 0\}$ &  $\{1, 0, 0\}$ & $\sigma_y$ \\ 
$\{0, 1, 0\}$ &  $\{1, 0, 0\}$ &  $\{1, 0, 0\}$ &  $\{0, 1, 0\}$ &  $\sigma_y$  &$\{0, 1, 0\}$ &  $\{1, 0, 0\}$ &  $\{1, 0, 0\}$ &  $\{0, 0, 1\}$ & $\sigma_y$ \\ 
$\{0, 1, 0\}$ &  $\{1, 0, 0\}$ &  $\{1, 0, 0\}$ &  $\{0.707, 0, 0.707\}$ &  $\sigma_y$  &$\{0, 1, 0\}$ &  $\{1, 0, 0\}$ &  $\{1, 0, 0\}$ &  $\{0, 0.707, 0.707\}$ & $\sigma_y$ \\ 
$\{0, 1, 0\}$ &  $\{1, 0, 0\}$ &  $\{0, 1, 0\}$ &  $\{1, 0, 0\}$ &  $\sigma_y$  &$\{0, 1, 0\}$ &  $\{1, 0, 0\}$ &  $\{0, 1, 0\}$ &  $\{0,0, 1\}$ & $\sigma_y$ \\ 
$\{0, 1, 0\}$ &  $\{1, 0, 0\}$ &  $\{0, 1, 0\}$ &  $\{0.707, 0, 0.707\}$ &  $\sigma_y$  &$\{0, 1, 0\}$ &  $\{1, 0, 0\}$ &  $\{0, 0, 1\}$ &  $\{1, 0,0\}$ & $\sigma_y$ \\ 
$\{0, 1, 0\}$ &  $\{1, 0, 0\}$ &  $\{0, 0, 1\}$ &  $\{0, 1,  0\}$ &  $\sigma_y$  &$\{0, 1, 0\}$ &  $\{1, 0, 0\}$ &  $\{0, 0, 1\}$ &  $\{0, 0, 1\}$ & $\sigma_y$ \\ 
$\{0, 1, 0\}$ &  $\{1, 0, 0\}$ &  $\{0.707, 0, 0.707\}$ &  $\{1, 0, 0\}$ &  $\sigma_y$  &$\{0, 1, 0\}$ &  $\{1, 0, 0\}$ &  $\{0, 0.707, 0.707\}$ &  $\{1, 0, 0\}$ & $\sigma_y$ \\ 
$\{0, 1, 0\}$ &  $\{0, 0, 1\}$ &  $\{1, 0, 0\}$ &  $\{0, 1, 0\}$ &  $\sigma_y$  &$\{0, 1, 0\}$ &  $\{0, 0, 1\}$ &  $\{0, 1, 0\}$ &  $\{1, 0, 0\}$ & $\sigma_y$ \\ 
$\{0, 1, 0\}$ &  $\{0, 0.707, 0.707\}$ &  $\{1, 0, 0\}$ &  $\{1, 0, 0\}$ &  $\sigma_y$  &  &  &  &  &   \\ \hline
    \end{tabular}
    \caption{\textbf{Regular CA QNS protocol of $K4q$} for the same parameters as in Table~\ref{tab:QNS_ctrl_other3} }
    \label{tab:QNS_ctrl_k4q}
\end{table*}

\begin{figure}[!htbp]
  \centering
  \includegraphics[width =0.5 \textwidth]{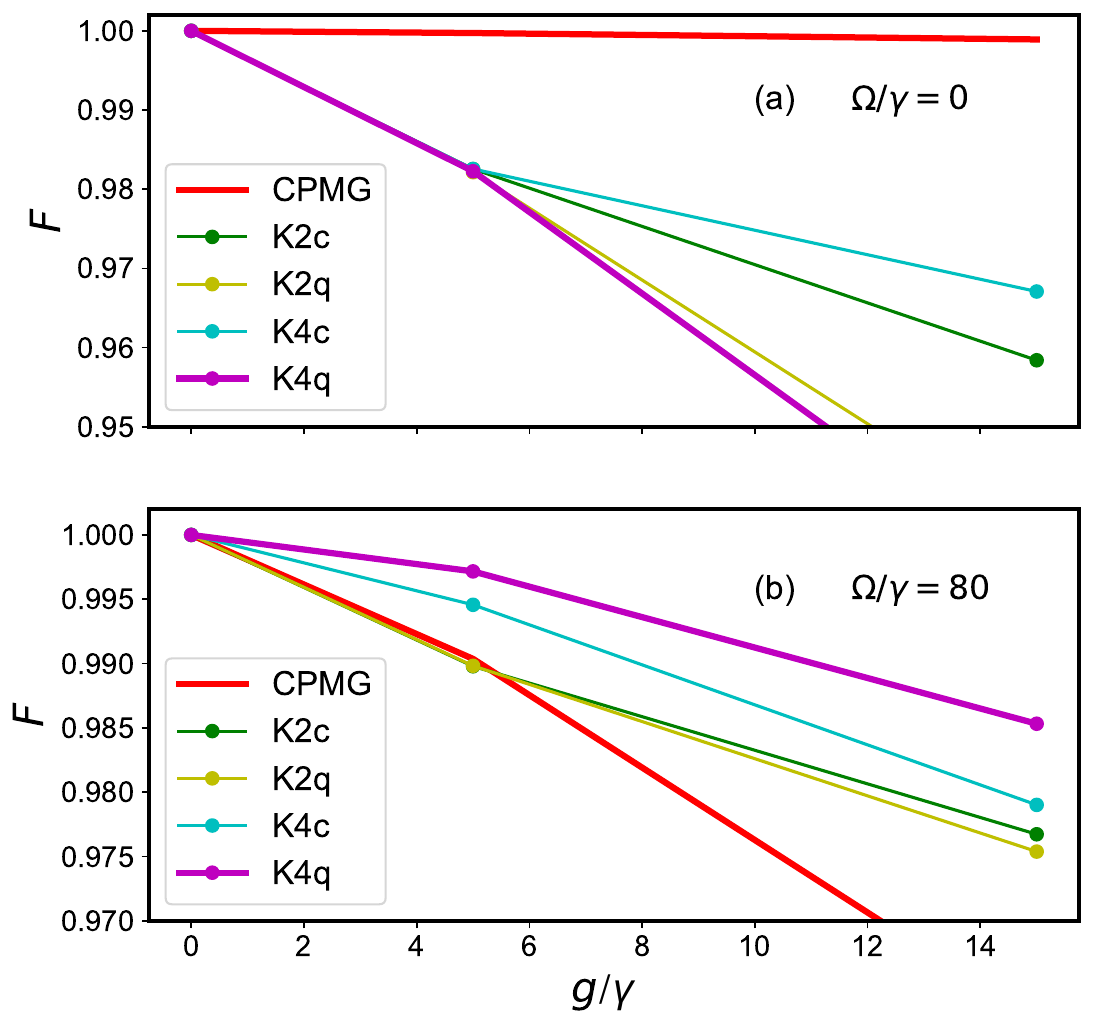}
\caption{\textbf{Noise-optimized control performance of different protocols}.  (a) The spectral peak of the noise profile at $\Omega/\gamma=0$.  (b) The spectral peak of the noise profile at $\Omega/\gamma=80$. The other noise parameters are $\gamma = 0.02$ MHz, $g = 0.1$ MHz, $T = 4~\mu$s, and $\tilde{T} = 5~\mu$s.}
\label{fig:fid_1qb_CC}
\end{figure}

\section{Derivation of QNS sample complexity bounds}
\label{app:bound}
We gave some details about deriving learning saturation bounds in the main text. Here we first briefly sketch how we derive the classical equation in Conclusion~\ref{corollary:1}. We only need to consider spectra that have no 3- (or more) streaks. 

In $k$-th order expansion, each window-frame $k$-string (which means an $n$-string with length $k$) of a learnable CA spectrum  (i) can be streak-free or (ii) includes only  2-streak.  For case (i), the combinations are given by $ \binom{L}{k}$.  For case (ii), we find the combinations are easier to handle by reducing each $k$-string to a $t$-string, keeping just one digit from the 2-streak exhaustively, such that the $t$-string is streak-free. Note that such a 2-streak contraction is only for auxiliary for analysis, rather than a legitimate contraction in control tensors.     
For each unique $t$-string, where $\ceil{k/2} \leq t <k$, it can map to multiple different $k$-string. For example, three different $k=4$-strings $(1,1,2,3)$, $(1,2,2,3)$, and $(1,2,3,3)$ map to a $t=3$-string $(1,2,3)$.  To obtain all possible $k$-strings in case (ii), we first focus on a $t$-string  and then expand backward to restore $k$-strings.  
The number of admissible $t$-string in case (ii) is $ \binom{L}{t}$, and for each unique $t$-string, the number of different $k$-string it can expand or map to is $ \binom{t}{k-t}$. Therefore, at $k$-th order, the number of learnable CA spectra corresponds to 
\begin{equation}
 \binom{L}{k} + \sum_{t = \ceil{k/2}}^{k-1}  \binom{t}{k-t}  \binom{L}{t}.
\end{equation}
The counting of quantum noise is based on the above analysis, while a factor of $2^{k-1}$ is introduced to account for all signs in nested brackets. We note that due to the complicated combinations, we are not able to account for the symmetry of dark spectra. Thus, its scaling is not as tight as in the classical case. 
The bounds $\Theta$ and $\mathcal{O}$ are not determined analytically but are instead obtained by fitting the exponent for large $L$.

The analysis of the two-qubit cases also focuses on the $n$-string first. In the $k$-th order expansion, each $k$-length $n$-string has a streak length no greater than 4, since 5-streak forms control tensor contraction. Suppose  a $n$-string has $p_j$ $j$-streak, where $j=2,3,4$. For aiding purposes, we fully contract all ($\geq2$) streaks in $n$-string to form a streak-free $t$-length $n$-string. It is evident that $k-t= p_2+2p_3+3p_4$ is a constraint. Given a $t$-length string, it can also map to multiple $k$-length strings. To that end, we sequentially select all possible 4-streaks, 3-streaks,  and then 2-streaks sequentially, with their combinations corresponding to $\binom{t}{p_4}$, $\binom{t-p_4}{p_3}$ and $\binom{t-p_4-p_3}{k-t-2p_3-3p_4}$, respectively. 

We also need to consider combinations of $q$-string in conjunction  with the above $n$-string analysis. For the 4-streak in $n$-string, we only need to include $q$-string snippet $(A,A,B,B)$ at the corresponding locations. This is because any other snippet at this location produces a dependent spectrum due to swap symmetry or binding symmetry. For example, $(A,A,A,B)$  produces a  $3$-streak in $q$-string and in $n$-string. Thus this segment behaves like a $3$-streak in a single-qubit dynamics, in which case control tensor contracts and binding symmetry arise.   For the 3-streak in $n$-string, we only need to include the $q$-string snippet $(A,A,B)$ and $(A,B,B)$ for a similar reason. For the 2-streak in $n$-string,  we need to include $(A,A)$, $(A,B)$ and $(B,B)$ [$(B,A)$ has swap symmetry with $(A,B)$]. In conclusion, at $k$-th order, all possible independent combinations are $$ \sum^k_{t=\ceil{k/4}}\binom{L}{t}  \sum_{p_4=0} ^{\floor{(k-t)/3}} \binom{t}{p_4} \sum_{p_3=0}^{\floor{k-t-3p_4}/2}   \binom{t-p_4}{p_3} 2^{p_3}  \binom{t-p_4-p_3}{k-t-2p_3-3p_4} 3^{k-t-2p_3-3p_4}.$$

\section{Details of Case Study 1: Single-qubit fundamental digital QNS protocol}
\label{app:case_1}
In Table~\ref{tab:case_1} we include the details of  $L=4$ fundamental digital QNS protocol that is truncated at saturation order $K=8$. There are  80  independent learnable CA spectra to reconstruct, which are used to predict the qubit behavior as shown in Fig.~\ref{fig:steps_QNS_recon}.

\begin{table*}  [!htbp]
    \centering
    \resizebox{0.7\textwidth}{!}{
\begin{tabular}{|cccc | c|| cccc|c| }
\hline$n=1$ & $n=2$ & $n=3$ & $n=4$ & $\widetilde{\rho}_Q$ & $n=1$ & $n=2$ & $n=3$ & $n=4$ & $\widetilde{\rho}_Q$ \\
\hline
   $\{1, 0, 0\}$ & $\{1, 0, 0\}$ & $\{1, 0, 0\}$ & $\{1, 0, 0\}$     &       $\sigma_z$ &    
   $\{1, 0, 0\}$ & $\{1, 0,  0\}$ & $\{1, 0, 0\}$ & $\{1, 0, 0\}$ &        $\sigma_y$  \\  
   $\{1, 0, 0\}$ & $\{1, 0, 0\}$ & $\{1, 0, 0\}$ & $\{0, 1, 0\}$ &        $\sigma_z$ &    
   $\{1, 0, 0\}$ & $\{1, 0, 0\}$ & $\{1, 0, 0\}$ & $\{0, 1, 0\}$ &        $\sigma_y$  \\   
   $\{1, 0, 0\}$ & $\{1, 0, 0\}$ & $\{1, 0, 0\}$ & $\{0, 0, 1\}$ &        $\sigma_z$ & 
   $\{1, 0, 0\}$ & $\{1, 0, 0\}$ & $\{1, 0,  0\}$ & $\{0, 0, 1\}$ &        $\sigma_y$ \\   
   $\{1, 0, 0\}$ & $\{1, 0, 0\}$ & $\{0, 1, 0\}$ & $\{1, 0, 0\}$ &        $\sigma_z$ &    
   $\{1,  0, 0\}$ & $\{1, 0, 0\}$ & $\{0, 1, 0\}$ & $\{1, 0, 0\}$ &       $\sigma_y$  \\
   $\{1, 0, 0\}$ & $\{1, 0, 0\}$ & $\{0,  1, 0\}$ & $\{0, 1, 0\}$ &       $\sigma_z$ &    
   $\{1, 0, 0\}$ & $\{1, 0, 0\}$ & $\{0, 1, 0\}$ & $\{0, 1,  0\}$ &        $\sigma_y$ \\
   $\{1, 0, 0\}$ & $\{1, 0, 0\}$ & $\{0, 1, 0\}$ & $\{0, 0, 1\}$ &        $\sigma_z$ &    
   $\{1, 0, 0\}$ & $\{1, 0, 0\}$ & $\{0, 1, 0\}$ & $\{0, 0, 1\}$ &        $\sigma_y$ \\
   $\{1, 0, 0\}$ & $\{1, 0, 0\}$ & $\{0, 1, 0\}$ & $\{1/\sqrt{2}, 1/\sqrt{2}, 0\}$ &   $\sigma_z$ &     
   $\{1, 0, 0\}$ & $\{1, 0, 0\}$ & $\{0, 1, 0\}$ & $\{1/\sqrt{2}, 1/\sqrt{2}, 0\}$ &      $\sigma_y$ \\ 
   $\{1, 0, 0\}$ & $\{1, 0, 0\}$ & $\{0, 1, 0\}$ & $\{0, 1/\sqrt{2}, 1/\sqrt{2}\}$ &        $\sigma_z$ &
   $\{1, 0, 0\}$ & $\{1, 0, 0\}$ & $\{0, 1, 0\}$ & $\{0, 1/\sqrt{2}, 1/\sqrt{2}\}$ &      $\sigma_y$   \\  
   $\{1, 0, 0\}$ & $\{1, 0, 0\}$ & $\{0, 0, 1\}$ & $\{1, 0, 0\}$ &        $\sigma_z$ & 
   $\{1, 0, 0\}$ & $\{1, 0, 0\}$ & $\{0, 0, 1\}$ & $\{1, 0, 0\}$ &        $\sigma_y$ \\  
   $\{1, 0, 0\}$ & $\{0, 1, 0\}$ & $\{1, 0, 0\}$ & $\{1, 0, 0\}$ &        $\sigma_z$ & 
   $\{1, 0, 0\}$ & $\{0, 1, 0\}$ & $\{1, 0, 0\}$ & $\{1, 0, 0\}$ &        $\sigma_y$ \\
   $\{1, 0, 0\}$ & $\{0, 1, 0\}$ & $\{1, 0, 0\}$ & $\{0, 1, 0\}$ &        $\sigma_z$ &    
   $\{1, 0, 0\}$ & $\{0, 1, 0\}$ & $\{1, 0, 0\}$ & $\{0, 1, 0\}$ &        $\sigma_y$ \\
   $\{1, 0, 0\}$ & $\{0, 1, 0\}$ & $\{1, 0, 0\}$ & $\{0, 0, 1\}$ &        $\sigma_z$ &  
   $\{1, 0, 0\}$ & $\{0, 1, 0\}$ & $\{1, 0, 0\}$ & $\{0, 0, 1\}$ &        $\sigma_y$  \\
   $\{1, 0, 0\}$ & $\{0, 1, 0\}$ & $\{1, 0, 0\}$ & $\{1/\sqrt{2}, 1/\sqrt{2}, 0\}$ &     $\sigma_z$ &       
   $\{1, 0, 0\}$ & $\{0, 1, 0\}$ & $\{1, 0, 0\}$ & $\{1/\sqrt{2}, 1/\sqrt{2}, 0\}$ &        $\sigma_y$ \\    
   $\{1, 0, 0\}$ & $\{0, 1, 0\}$ & $\{1, 0, 0\}$ & $\{1/\sqrt{2}, 0, 1/\sqrt{2}\}$ &       $\sigma_z$ &     
   $\{1, 0, 0\}$ & $\{0, 1, 0\}$ & $\{1, 0, 0\}$ & $\{1/\sqrt{2}, 0, 1/\sqrt{2}\}$ &       $\sigma_y$ \\   
   $\{1, 0, 0\}$ & $\{0, 1, 0\}$ & $\{1, 0, 0\}$ & $\{0, 1/\sqrt{2}, 1/\sqrt{2}\}$ &        $\sigma_z$ & 
   $\{1, 0, 0\}$ & $\{0, 1, 0\}$ & $\{1, 0, 0\}$ & $\{0, 1/\sqrt{2}, 1/\sqrt{2}\}$ &        $\sigma_y$ \\ 
   $\{1, 0, 0\}$ & $\{0, 1, 0\}$ & $\{0, 1, 0\}$ & $\{1, 0, 0\}$ &        $\sigma_z$ &    
   $\{1, 0, 0\}$ & $\{0, 1, 0\}$ & $\{0, 1, 0\}$ & $\{1, 0,  0\}$ &        $\sigma_y$ \\    
   $\{1, 0, 0\}$ & $\{0, 1, 0\}$ & $\{0, 1, 0\}$ & $\{0, 1, 0\}$ &        $\sigma_z$ & 
   $\{1, 0, 0\}$ & $\{0, 1, 0\}$ & $\{0, 1, 0\}$ & $\{0, 1, 0\}$ &        $\sigma_y$  \\
   $\{1, 0, 0\}$ & $\{0, 1, 0\}$ & $\{0, 1, 0\}$ & $\{0, 0, 1\}$ &        $\sigma_z$ &    
   $\{1, 0, 0\}$ & $\{0, 1, 0\}$ & $\{0, 1, 0\}$ & $\{0, 0, 1\}$ &        $\sigma_y$  \\
   $\{1, 0, 0\}$ & $\{0, 1, 0\}$ & $\{0, 1, 0\}$ & $\{1/\sqrt{2}, 1/\sqrt{2}, 0\}$ &      $\sigma_z$ &    
   $\{1, 0, 0\}$ & $\{0, 1, 0\}$ & $\{0, 1, 0\}$ & $\{1/\sqrt{2}, 1/\sqrt{2}, 0\}$ &      $\sigma_y$ \\      
   $\{1, 0, 0\}$ & $\{0, 1, 0\}$ & $\{0, 1, 0\}$ & $\{0, 1/\sqrt{2}, 1/\sqrt{2}\}$ &      $\sigma_z$ &  
   $\{1, 0, 0\}$ & $\{0, 1, 0\}$ & $\{0, 1, 0\}$ & $\{0, 1/\sqrt{2}, 1/\sqrt{2}\}$ &      $\sigma_y$ \\  
   $\{1, 0, 0\}$ & $\{0, 1, 0\}$ & $\{0, 0, 1\}$ & $\{1, 0, 0\}$ &      $\sigma_z$  &
   $\{1, 0, 0\}$ & $\{0, 1, 0\}$ & $\{0, 0, 1\}$ & $\{1, 0, 0\}$ &      $\sigma_y$ \\   
   $\{1, 0, 0\}$ & $\{0, 1, 0\}$ & $\{0, 0, 1\}$ & $\{0, 1, 0\}$ &      $\sigma_z$ &   
   $\{1, 0, 0\}$ & $\{0, 1, 0\}$ & $\{0, 0, 1\}$ & $\{0, 1, 0\}$ &      $\sigma_y$ \\   
   $\{1, 0, 0\}$ & $\{0, 1, 0\}$ & $\{0, 0, 1\}$ & $\{0, 0, 1\}$ &      $\sigma_z$ &   
   $\{1, 0, 0\}$ & $\{0, 1, 0\}$ & $\{0, 0, 1\}$ & $\{0, 0,  1\}$ &     $\sigma_y$ \\   
   $\{1, 0, 0\}$ & $\{0, 1, 0\}$ & $\{0, 0, 1\}$ & $\{1/\sqrt{2}, 0, 1/\sqrt{2}\}$ &  $\sigma_z$ &       
   $\{1, 0, 0\}$ & $\{0, 1, 0\}$ & $\{0, 0, 1\}$ & $\{1/\sqrt{2}, 0, 1/\sqrt{2}\}$ &     $\sigma_y$ \\   
   $\{1, 0, 0\}$ & $\{0, 1, 0\}$ & $\{1/\sqrt{2}, 1/\sqrt{2}, 0\}$ & $\{1, 0,0\}$ &      $\sigma_z$ &  
   $\{1, 0, 0\}$ & $\{0, 1, 0\}$ & $\{1/\sqrt{2}, 1/\sqrt{2}, 0\}$ & $\{1, 0,0\}$ &      $\sigma_y$ \\  
   $\{1, 0, 0\}$ & $\{0, 1, 0\}$ & $\{0, 1/\sqrt{2}, 1/\sqrt{2}\}$ & $\{1, 0,0\}$ &      $\sigma_z$ &    
   $\{1, 0, 0\}$ & $\{0, 1, 0\}$ & $\{0, 1/\sqrt{2}, 1/\sqrt{2}\}$ & $\{1, 0,0\}$ &       $\sigma_y$ \\   
   $\{1, 0, 0\}$ & $\{0, 0, 1\}$ & $\{1, 0, 0\}$ & $\{1, 0, 0\}$ &        $\sigma_z$ & 
   $\{1, 0, 0\}$ & $\{0, 0, 1\}$ & $\{1, 0, 0\}$ & $\{1, 0, 0\}$ &        $\sigma_y$ \\  
   $\{0, 1, 0\}$ & $\{1, 0, 0\}$ & $\{1, 0, 0\}$ & $\{1, 0, 0\}$ &        $\sigma_y$ & 
   $\{0, 1, 0\}$ & $\{1, 0, 0\}$ & $\{1, 0, 0\}$ & $\{0, 1, 0\}$ &        $\sigma_y$  \\
   $\{0, 1, 0\}$ & $\{1, 0, 0\}$ & $\{1, 0, 0\}$ & $\{0, 0, 1\}$ &        $\sigma_y$  &
   $\{0, 1, 0\}$ & $\{1, 0, 0\}$ & $\{1, 0,  0\}$ & $\{0, 1/\sqrt{2}, 1/\sqrt{2}\}$ &   $\sigma_y$ \\       
   $\{0, 1, 0\}$ & $\{1, 0, 0\}$ & $\{0, 1, 0\}$ & $\{1, 0, 0\}$ &      $\sigma_y$ &   
   $\{0, 1, 0\}$ & $\{1, 0, 0\}$ & $\{0, 1, 0\}$ & $\{0, 1, 0\}$ &        $\sigma_y$ \\ 
   $\{0, 1, 0\}$ & $\{1, 0, 0\}$ & $\{0, 1, 0\}$ & $\{0, 0, 1\}$ &        $\sigma_y$ &   
   $\{0, 1, 0\}$ & $\{1, 0, 0\}$ & $\{0, 1, 0\}$ & $\{1/\sqrt{2}, 1/\sqrt{2}, 0\}$ &      $\sigma_y$ \\     
   $\{0, 1, 0\}$ & $\{1, 0, 0\}$ & $\{0, 1, 0\}$ & $\{1/\sqrt{2}, 0, 1/\sqrt{2}\}$ &     $\sigma_y$ &      
   $\{0, 1, 0\}$ & $\{1, 0, 0\}$ & $\{0, 1,  0\}$ & $\{0, 1/\sqrt{2}, 1/\sqrt{2}\}$ &        $\sigma_y$ \\   
   $\{0, 1, 0\}$ & $\{1, 0, 0\}$ & $\{0, 0,  1\}$ & $\{1, 0, 0\}$ &        $\sigma_y$ &   
   $\{0, 1, 0\}$ & $\{1, 0, 0\}$ & $\{0, 0, 1\}$ & $\{0, 1, 0\}$ &        $\sigma_y$ \\   
   $\{0, 1, 0\}$ & $\{1, 0, 0\}$ & $\{0, 0, 1\}$ & $\{0, 0, 1\}$ &        $\sigma_y$ &   
   $\{0, 1, 0\}$ & $\{1, 0, 0\}$ & $\{1/\sqrt{2}, 0, 1/\sqrt{2}\}$ & $\{0, 1, 0\}$ &      $\sigma_y$ \\     
   $\{0, 1, 0\}$ & $\{1, 0, 0\}$ & $\{0, 1/\sqrt{2}, 1/\sqrt{2}\}$ & $\{1, 0, 0\}$ &        $\sigma_y$ &   
   $\{0, 1, 0\}$ & $\{0, 1, 0\}$ & $\{1, 0, 0\}$ & $\{1, 0, 0\}$ &        $\sigma_y$ \\   
   $\{0, 1, 0\}$ & $\{0, 1, 0\}$ & $\{1, 0, 0\}$ & $\{0, 1, 0\}$ &        $\sigma_y$ &   
   $\{0, 1, 0\}$ & $\{0, 1, 0\}$ & $\{1, 0, 0\}$ & $\{0, 0, 1\}$ &        $\sigma_y$ \\   
   $\{0, 1, 0\}$ & $\{0, 1, 0\}$ & $\{1, 0, 0\}$ & $\{0, 1/\sqrt{2}, 1/\sqrt{2}\}$ &      $\sigma_y$ &     
   $\{0, 1, 0\}$ & $\{0, 1, 0\}$ & $\{0, 1, 0\}$ & $\{1, 0, 0\}$ &        $\sigma_y$ \\   
   $\{0, 1, 0\}$ & $\{0, 1, 0\}$ & $\{0, 1, 0\}$ & $\{0, 1/\sqrt{2}, 1/\sqrt{2}\}$ &      $\sigma_y$ &     
   $\{0, 1, 0\}$ & $\{0, 1, 0\}$ & $\{0, 0, 1\}$ & $\{0, 1, 0\}$  &        $\sigma_y$ \\   
   $\{0, 1, 0\}$ & $\{0, 1, 0\}$ & $\{0, 1/\sqrt{2}, 1/\sqrt{2}\}$ & $\{1, 0, 0\}$ &       $\sigma_y$ &    
   $\{0, 1, 0\}$ & $\{0, 0, 1\}$ & $\{1, 0, 0\}$ & $\{0, 1, 0\}$ &        $\sigma_y$ \\   
   $\{0, 1, 0\}$ & $\{0, 0, 1\}$ & $\{0, 1, 0\}$ & $\{1, 0, 0\}$&        $\sigma_y$ &   
   $\{0, 1, 0\}$ & $\{0, 1/\sqrt{2}, 1/\sqrt{2}\}$ & $\{1, 0, 0\}$ & $\{1, 0, 0\}$&       $\sigma_y$ \\    
\hline
\end{tabular}
}
\caption{\textbf{Single-qubit fundamental digital QNS protocol ($L=4$)}  The $L=4$ protocol consists of 80 control vectors $\{y_x(n),y_y(n),y_z(n)\}$ in different windows, $n=1,2,3,4$, and corresponding initial states $\widetilde{\rho}_Q$. The initial states here can be inferred from combinations of trace-1 physical initial states; e.g, $\widetilde{\rho}_Q=\sigma_z=1/2(\mathbb{I}_2+\sigma_z) - 1/2(\mathbb{I}_2-\sigma_z).$  In all experiments, ${O}\equiv\sigma_z$.  The protocol does not {\em a priori} assume stationarity and zero mean.    }
    \label{tab:case_1}
\end{table*}

\section{Details of Case Study 2: Two-qubit fundamental digital QNS protocol}
\label{app:case_2}

Here we provide the details of two-qubit fundamental digital QNS protocol with $L=2$ and $K=8$. Due to the complexity of two-qubit switching functions, we present the Hamiltonian-level control parameters obtained through $KAK$ decomposition (see main text) rather than  $y_{[q],u}(t)$. The control parameters of the two-qubit digital QNS protocol are shown in  Table.~\ref{tab:2qb_fundamental_QNS}.

\begin{table}[!htbp]
\centering
\resizebox{1.0\columnwidth}{!}{%
\begin{tabular}{|rrrrrrrrrrrrrrr|r||rrrrrrrrrrrrrrr|r| }
\toprule
 $\theta_A(1)$ &  $\alpha_A(1)$ &  $\phi_A(1)$ &  $\theta_B(1)$ &  $\alpha_B(1)$ &  $\phi_B(1)$ &  $\Theta(1)$ &  $\rho(1)$ &  $\omega(1)$ &  $\theta^{\prime}_A(1)$ &  $\alpha^{\prime}_A(1)$ &  $\phi^{\prime}_A(1)$ &  $\theta^{\prime}_B(1)$ &  $\alpha^{\prime}_B(1)$ &  $\phi^{\prime}_B(1)$ & $\widetilde{\rho}_Q$ &
 $\theta_A(1)$ &  $\alpha_A(1)$ &  $\phi_A(1)$ &  $\theta_B(1)$ &  $\alpha_B(1)$ &  $\phi_B(1)$ &  $\Theta(1)$ &  $\rho(1)$ &  $\omega(1)$ &  $\theta^{\prime}_A(1)$ &  $\alpha^{\prime}_A(1)$ &  $\phi^{\prime}_A(1)$ &  $\theta^{\prime}_B(1)$ &  $\alpha^{\prime}_B(1)$ &  $\phi^{\prime}_B(1)$ & $\widetilde{\rho}_Q$\\
 $\theta_A(2)$ &  $\alpha_A(2)$ &  $\phi_A(2)$ &  $\theta_B(2)$ &  $\alpha_B(2)$ &  $\phi_B(2)$ &  $\Theta(2)$ &  $\rho(2)$ &  $\omega(2)$ &  $\theta^{\prime}_A(2)$ &  $\alpha^{\prime}_A(2)$ &  $\phi^{\prime}_A(2)$ &  $\theta^{\prime}_B(2)$ &  $\alpha^{\prime}_B(2)$ &  $\phi^{\prime}_B(2)$ &  
 & $\theta_A(2)$ &  $\alpha_A(2)$ &  $\phi_A(2)$ &  $\theta_B(2)$ &  $\alpha_B(2)$ &  $\phi_B(2)$ &  $\Theta(2)$ &  $\rho(2)$ &  $\omega(2)$ &  $\theta^{\prime}_A(2)$ &  $\alpha^{\prime}_A(2)$ &  $\phi^{\prime}_A(2)$ &  $\theta^{\prime}_B(2)$ &  $\alpha^{\prime}_B(2)$ &  $\phi^{\prime}_B(2)$ & \\
\midrule
0.342& -0.102& 0.934& 0.0& 0.0& 0.0& 0.0& 0.0& 0.0& 0.0& 0.0& 0.0& 0.0& 0.0& 0.0& $\Lambda_{7}$ & 0.414& 0.22& 0.455& -0.15& 0.014& -0.263& 0.114& 0.049& -0.395& 0.387& 0.172& -0.216& -0.098& 0.269& -0.041 & $\Lambda_{2}$  \\
0.988& 0.022& 0.15& 0.0& 0.0& 0.0& 0.0& 0.0& 0.0& 0.0& 0.0& 0.0& 0.0& 0.0& 0.0  &   & -0.003& -0.123& -0.323& 0.255& -0.32& 0.091& -0.032& -0.083& -0.205& -0.03& 0.013& -0.087& -0.741& 0.295& -0.106  &  \\ \hline 
0.06& 0.474& -0.399& 0.0& 0.196& -0.492& -0.556& 0.0& -0.09& -0.013& -0.029& 0.118& 0.0& 0.0& 0.0& $\Lambda_{9}$ &  -0.008& 0.562& 0.827& 0.0& 0.0& 0.0& 0.0& 0.0& 0.0& 0.0& 0.0& 0.0& 0.0& 0.0& 0.0 & $\Lambda_{11}$  \\
0.347& -0.065& 0.159& -0.315& -0.06& -0.301& 0.007& -0.433& 0.141& 0.249& -0.206& 0.174& 0.242& 0.073& -0.5  &    &  0.0& 0.0& 1.0& 0.0& 0.0& 0.0& 0.0& 0.0& 0.0& 0.0& 0.0& 0.0& 0.0& 0.0& 0.0  &   \\ \hline 
0.71& -0.379& 0.103& 0.0& 0.0& 0.15& 0.554& 0.0& 0.016& 0.0& -0.11& -0.023& 0.0& 0.0& 0.0& $\Lambda_{2}$  &  -0.721& -0.026& -0.693& 0.0& 0.0& 0.0& 0.0& 0.0& 0.0& 0.0& 0.0& 0.0& 0.0& 0.0& 0.0 & $\Lambda_{7}$  \\
0.0& 0.0& 1.0& 0.0& 0.0& 0.0& 0.0& 0.0& 0.0& 0.0& 0.0& 0.0& 0.0& 0.0& 0.0  &    &  0.218& -0.19& -0.189& 0.0& -0.615& -0.638& -0.066& 0.0& 0.039& -0.047& 0.091& -0.282& 0.0& 0.0& 0.0  &   \\ \hline
0.0& 0.208& 0.295& -0.079& 0.007& -0.284& 0.201& -0.026& -0.032& 0.154& -0.109& -0.027& 0.01& 0.686& -0.484& $\Lambda_{10}$  &   0.558& 0.744& 0.369& 0.0& 0.0& 0.0& 0.0& 0.0& 0.0& 0.0& 0.0& 0.0& 0.0& 0.0& 0.0 & $\Lambda_{15}$ \\
-0.186& 0.965& 0.184& 0.0& 0.0& 0.0& 0.0& 0.0& 0.0& 0.0& 0.0& 0.0& 0.0& 0.0& 0.0  &   &  -0.532& 0.182& 0.827& 0.0& 0.0& 0.0& 0.0& 0.0& 0.0& 0.0& 0.0& 0.0& 0.0& 0.0& 0.0  &   \\ \hline
0.0& 0.0& 1.0& 0.0& 0.0& 0.0& 0.0& 0.0& 0.0& 0.0& 0.0& 0.0& 0.0& 0.0& 0.0& $\Lambda_{2}$  &   0.839& 0.516& 0.171& 0.0& 0.0& 0.0& 0.0& 0.0& 0.0& 0.0& 0.0& 0.0& 0.0& 0.0& 0.0 & $\Lambda_{10}$  \\
-0.35& 0.292& 0.323& 0.049& -0.009& -0.063& 0.047& 0.105& -0.108& 0.028& -0.142& -0.022& -0.544& -0.008& -0.582  & &  0.0& 0.588& -0.809& 0.0& 0.0& 0.0& 0.0& 0.0& 0.0& 0.0& 0.0& 0.0& 0.0& 0.0& 0.0  &   \\ \hline
0.204& -0.701& 0.683& 0.0& 0.0& 0.0& 0.0& 0.0& 0.0& 0.0& 0.0& 0.0& 0.0& 0.0& 0.0& $\Lambda_{14}$  &  0.203& -0.311& -0.928& 0.0& 0.0& 0.0& 0.0& 0.0& 0.0& 0.0& 0.0& 0.0& 0.0& 0.0& 0.0 & $\Lambda_{11}$ \\
-0.226& 0.514& 0.827& 0.0& 0.0& 0.0& 0.0& 0.0& 0.0& 0.0& 0.0& 0.0& 0.0& 0.0& 0.0  &  & 0.067& -0.535& 0.842& 0.0& 0.0& 0.0& 0.0& 0.0& 0.0& 0.0& 0.0& 0.0& 0.0& 0.0& 0.0  &   \\ \hline 
0.559& 0.532& -0.636& 0.0& 0.0& 0.0& 0.0& 0.0& 0.0& 0.0& 0.0& 0.0& 0.0& 0.0& 0.0& $\Lambda_{10}$  & -0.176& -0.107& -0.417& 0.0& 0.133& 0.792& -0.26& 0.0& -0.123& -0.019& 0.057& -0.228& 0.0& 0.0& 0.0 & $\Lambda_{5}$ \\
0.905& 0.294& -0.309& 0.0& 0.0& 0.0& 0.0& 0.0& 0.0& 0.0& 0.0& 0.0& 0.0& 0.0& 0.0  &   &  -0.86& -0.476& -0.184& 0.0& 0.0& 0.0& 0.0& 0.0& 0.0& 0.0& 0.0& 0.0& 0.0& 0.0& 0.0  &   \\ \hline
-0.302& -0.115& 0.815& -0.033& -0.348& -0.071& -0.139& -0.045& -0.075& 0.163& -0.004& -0.076& 0.193& -0.066& 0.062& $\Lambda_{15}$  &  0.559& -0.182& -0.809& 0.0& 0.0& 0.0& 0.0& 0.0& 0.0& 0.0& 0.0& 0.0& 0.0& 0.0& 0.0 & $\Lambda_{11}$ \\
 0.0& 0.0& 1.0& 0.0& 0.0& 0.0& 0.0& 0.0& 0.0& 0.0& 0.0& 0.0& 0.0& 0.0& 0.0  &   &  -0.866& 0.283& 0.413& 0.0& 0.0& 0.0& 0.0& 0.0& 0.0& 0.0& 0.0& 0.0& 0.0& 0.0& 0.0  &   \\ \hline
 0.213& -0.718& -0.128& -0.023& 0.305& 0.161& -0.397& -0.053& -0.15& -0.095& 0.285& 0.038& -0.028& -0.037& 0.162& $\Lambda_{14}$ &  0.476& 0.631& 0.391& -0.017& -0.324& 0.238& 0.011& 0.025& 0.111& 0.013& -0.156& -0.078& 0.105& -0.048& -0.051 & $\Lambda_{12}$ \\ 
 -0.905& -0.294& 0.309& 0.0& 0.0& 0.0& 0.0& 0.0& 0.0& 0.0& 0.0& 0.0& 0.0& 0.0& 0.0  &   &  -0.497& 0.494& 0.713& 0.0& 0.0& 0.0& 0.0& 0.0& 0.0& 0.0& 0.0& 0.0& 0.0& 0.0& 0.0  &   \\ \hline
 -0.765& 0.339& 0.548& 0.0& 0.0& 0.0& 0.0& 0.0& 0.0& 0.0& 0.0& 0.0& 0.0& 0.0& 0.0& $\Lambda_{10}$ & 0.354& 0.413& 0.839& 0.0& 0.0& 0.0& 0.0& 0.0& 0.0& 0.0& 0.0& 0.0& 0.0& 0.0& 0.0 & $\Lambda_{9}$ \\
 -0.033& 0.002& -0.999& 0.0& 0.0& 0.0& 0.0& 0.0& 0.0& 0.0& 0.0& 0.0& 0.0& 0.0& 0.0  &   & -0.693& 0.51& 0.192& 0.0& -0.106& -0.118& -0.07& 0.0& -0.071& -0.23& 0.356& 0.084& 0.0& 0.0& 0.0  &   \\ \hline
 0.344& -0.479& -0.48& 0.023& -0.317& 0.16& -0.387& -0.053& -0.164& 0.14& -0.257& 0.038& -0.04& 0.101& -0.129& $\Lambda_{2}$  &  0.964& -0.068& 0.103& 0.0& 0.0& 0.15& 0.1& -0.022& -0.005& 0.0& 0.046& 0.007& -0.015& 0.0& 0.142 & $\Lambda_{3}$ \\
-0.559& -0.769& 0.309& 0.0& 0.0& 0.0& 0.0& 0.0& 0.0& 0.0& 0.0& 0.0& 0.0& 0.0& 0.0  &   & -0.203& 0.294& 0.934& 0.0& 0.0& 0.0& 0.0& 0.0& 0.0& 0.0& 0.0& 0.0& 0.0& 0.0& 0.0  &   \\ \hline
-0.312& 0.429& -0.172& 0.091& 0.186& 0.311& 0.437& -0.085& 0.203& 0.317& 0.423& 0.023& 0.014& -0.056& -0.164& $\Lambda_{5}$ &  0.731& -0.205& 0.03& 0.023& -0.048& -0.172& -0.014& 0.008& 0.131& 0.516& 0.317& -0.065& -0.001& 0.003& 0.035 & $\Lambda_{13}$ \\
-0.203& 0.294& 0.934& 0.0& 0.0& 0.0& 0.0& 0.0& 0.0& 0.0& 0.0& 0.0& 0.0& 0.0& 0.0  & & 0.0& 0.588& -0.809& 0.0& 0.0& 0.0& 0.0& 0.0& 0.0& 0.0& 0.0& 0.0& 0.0& 0.0& 0.0  &   \\ \hline
-0.203& 0.294& 0.934& 0.0& 0.0& 0.0& 0.0& 0.0& 0.0& 0.0& 0.0& 0.0& 0.0& 0.0& 0.0& $\Lambda_{6}$  & 0.833& -0.075& 0.548& 0.0& 0.0& 0.0& 0.0& 0.0& 0.0& 0.0& 0.0& 0.0& 0.0& 0.0& 0.0 & $\Lambda_{5}$ \\
-0.184& -0.461& -0.306& 0.0& 0.0& 0.447& -0.673& 0.0& 0.047& 0.0& -0.029& 0.069& 0.0& 0.0& 0.0  &  & 0.325& 0.156& -0.508& 0.037& -0.026& 0.499& 0.137& 0.013& 0.038& 0.376& 0.14& 0.057& 0.009& -0.403& -0.119  &   \\ \hline
0.275& 0.166& -0.255& 0.003& -0.096& 0.52& 0.236& 0.227& -0.103& -0.334& -0.328& -0.334& -0.071& -0.221& -0.22& $\Lambda_{15}$  &  -0.105& -0.754& -0.649& 0.0& 0.0& 0.0& 0.0& 0.0& 0.0& 0.0& 0.0& 0.0& 0.0& 0.0& 0.0 & $\Lambda_{11}$  \\
-0.305& 0.952& -0.031& 0.0& 0.0& 0.0& 0.0& 0.0& 0.0& 0.0& 0.0& 0.0& 0.0& 0.0& 0.0  &  &  -0.747& 0.512& 0.424& 0.0& 0.0& 0.0& 0.0& 0.0& 0.0& 0.0& 0.0& 0.0& 0.0& 0.0& 0.0  &   \\ \hline 
0.0& 0.0& 1.0& 0.0& 0.0& 0.0& 0.0& 0.0& 0.0& 0.0& 0.0& 0.0& 0.0& 0.0& 0.0& $\Lambda_{7}$  &   0.206& -0.077& 0.976& 0.0& 0.0& 0.0& 0.0& 0.0& 0.0& 0.0& 0.0& 0.0& 0.0& 0.0& 0.0 & $\Lambda_{4}$  \\
-0.294& -0.103& 0.815& -0.091& -0.098& 0.057& -0.028& -0.03& 0.179& -0.365& 0.018& 0.031& -0.116& -0.18& -0.065  & &  0.072& 0.094& -0.31& 0.0& -0.281& 0.112& -0.031& 0.0& 0.208& 0.583& 0.225& 0.603& 0.0& 0.0& 0.0  &   \\ \hline
0.217& 0.284& 0.934& 0.0& 0.0& 0.0& 0.0& 0.0& 0.0& 0.0& 0.0& 0.0& 0.0& 0.0& 0.0& $\Lambda_{13}$ &  0.217& 0.284& 0.934& 0.0& 0.0& 0.0& 0.0& 0.0& 0.0& 0.0& 0.0& 0.0& 0.0& 0.0& 0.0 & $\Lambda_{11}$ \\
-0.466& -0.533& 0.127& 0.036& 0.264& -0.106& 0.522& -0.085& 0.182& -0.103& 0.238& -0.061& -0.098& 0.08& -0.021  &  &  -0.466& -0.533& 0.127& 0.036& 0.264& -0.106& 0.522& -0.085& 0.182& -0.103& 0.238& -0.061& -0.098& 0.08& -0.021  &   \\ \hline
-0.093& 0.129& 0.052& 0.549& 0.033& 0.117& -0.231& -0.049& -0.247& -0.308& 0.32& 0.361& -0.084& -0.22& 0.395& $\Lambda_{7}$  &  0.186& 0.965& -0.184& 0.0& 0.0& 0.0& 0.0& 0.0& 0.0& 0.0& 0.0& 0.0& 0.0& 0.0& 0.0 & $\Lambda_{14}$ \\
-0.559& -0.769& 0.309& 0.0& 0.0& 0.0& 0.0& 0.0& 0.0& 0.0& 0.0& 0.0& 0.0& 0.0& 0.0  &  & 0.203& -0.294& 0.934& 0.0& 0.0& 0.0& 0.0& 0.0& 0.0& 0.0& 0.0& 0.0& 0.0& 0.0& 0.0  &   \\ \hline
0.203& -0.294& 0.934& 0.0& 0.0& 0.0& 0.0& 0.0& 0.0& 0.0& 0.0& 0.0& 0.0& 0.0& 0.0& $\Lambda_{14}$  &  -0.031& 0.99& 0.139& 0.0& 0.0& 0.0& 0.0& 0.0& 0.0& 0.0& 0.0& 0.0& 0.0& 0.0& 0.0 & $\Lambda_{6}$  \\
 0.372& -0.135& 0.133& 0.0& -0.246& -0.736& -0.064& 0.237& 0.082& 0.068& -0.159& 0.327& -0.059& -0.049& 0.116  &  &  0.056& 0.02& 0.156& 0.35& 0.368& 0.232& -0.162& -0.12& -0.394& -0.017& 0.144& 0.545& -0.323& -0.153& 0.136  &   \\ \hline 
 -0.242& -0.175& 0.092& 0.119& -0.047& 0.299& 0.446& 0.278& -0.026& -0.22& -0.488& 0.198& 0.065& 0.129& 0.417& $\Lambda_{2}$ &  0.488& -0.671& 0.269& 0.035& 0.313& 0.161& -0.165& -0.01& 0.06& -0.05& -0.234& -0.094& 0.108& 0.065& -0.035 & $\Lambda_{15}$  \\
 -0.201& -0.297& -0.933& 0.0& 0.0& 0.0& 0.0& 0.0& 0.0& 0.0& 0.0& 0.0& 0.0& 0.0& 0.0  & &  -0.057& -0.017& 0.156& 0.0& -0.635& 0.264& -0.204& 0.0& -0.048& -0.142& -0.033& 0.659& 0.0& 0.0& 0.0  &   \\ \hline 
 0.068& -0.098& -0.31& 0.0& 0.275& -0.094& 0.089& 0.0& 0.247& 0.595& -0.134& -0.603& 0.0& 0.0& 0.0& $\Lambda_{2}$ &  0.337& 0.45& -0.827& 0.0& 0.0& 0.0& 0.0& 0.0& 0.0& 0.0& 0.0& 0.0& 0.0& 0.0& 0.0 & $\Lambda_{14}$ \\
 -0.008& -0.562& 0.827& 0.0& 0.0& 0.0& 0.0& 0.0& 0.0& 0.0& 0.0& 0.0& 0.0& 0.0& 0.0  &  &  0.0& 0.0& -0.602& 0.28& -0.146& -0.334& 0.0& 0.167& -0.198& 0.427& 0.0& 0.213& 0.347& 0.103& 0.0  &   \\ \hline 
 0.414& 0.22& 0.455& -0.15& 0.014& -0.263& 0.114& 0.049& -0.395& 0.387& 0.172& -0.216& -0.098& 0.269& -0.041& $\Lambda_{4}$ &  0.0& 0.433& 0.487& -0.176& 0.113& 0.34& -0.302& -0.243& -0.189& -0.147& 0.13& 0.097& -0.265& 0.25& -0.222 & $\Lambda_{10}$ \\
-0.003& -0.123& -0.323& 0.255& -0.32& 0.091& -0.032& -0.083& -0.205& -0.03& 0.013& -0.087& -0.741& 0.295& -0.106  &  &  -0.369& -0.408& -0.097& -0.107& 0.067& 0.011& -0.299& -0.035& 0.011& -0.051& 0.174& 0.036& 0.207& -0.018& -0.71  &   \\ \hline 
0.008& 0.562& 0.827& 0.0& 0.0& 0.0& 0.0& 0.0& 0.0& 0.0& 0.0& 0.0& 0.0& 0.0& 0.0& $\Lambda_{5}$ & 0.731& -0.205& 0.03& 0.023& -0.048& -0.172& -0.014& 0.008& 0.131& 0.516& 0.317& -0.065& -0.001& 0.003& 0.035 & $\Lambda_{8}$ \\
-0.924& -0.156& -0.349& 0.0& 0.0& 0.0& 0.0& 0.0& 0.0& 0.0& 0.0& 0.0& 0.0& 0.0& 0.0  &  & 0.0& 0.588& -0.809& 0.0& 0.0& 0.0& 0.0& 0.0& 0.0& 0.0& 0.0& 0.0& 0.0& 0.0& 0.0  &   \\ \hline
-0.284& -0.406& -0.182& -0.003& -0.22& -0.132& 0.638& -0.186& -0.235& 0.085& 0.178& 0.273& -0.162& 0.056& 0.128& $\Lambda_{13}$  & 0.488& 0.778& -0.395& 0.0& 0.0& 0.0& 0.0& 0.0& 0.0& 0.0& 0.0& 0.0& 0.0& 0.0& 0.0 & $\Lambda_{5}$  \\
-0.5& -0.505& 0.703& 0.0& 0.0& 0.0& 0.0& 0.0& 0.0& 0.0& 0.0& 0.0& 0.0& 0.0& 0.0  &    & -0.709& 0.59& 0.386& 0.0& 0.0& 0.0& 0.0& 0.0& 0.0& 0.0& 0.0& 0.0& 0.0& 0.0& 0.0  &   \\ \hline 
0.0& 0.0& 1.0& 0.0& 0.0& 0.0& 0.0& 0.0& 0.0& 0.0& 0.0& 0.0& 0.0& 0.0& 0.0& $\Lambda_{15}$ & -0.95& 0.254& -0.181& 0.0& 0.0& 0.0& 0.0& 0.0& 0.0& 0.0& 0.0& 0.0& 0.0& 0.0& 0.0 & $\Lambda_{10}$ \\
0.345& -0.476& -0.809& 0.0& 0.0& 0.0& 0.0& 0.0& 0.0& 0.0& 0.0& 0.0& 0.0& 0.0& 0.0  &   & 0.0& 0.0& 1.0& 0.0& 0.0& 0.0& 0.0& 0.0& 0.0& 0.0& 0.0& 0.0& 0.0& 0.0& 0.0  &   \\ \hline 
-0.21& 0.682& 0.7& 0.0& 0.0& 0.0& 0.0& 0.0& 0.0& 0.0& 0.0& 0.0& 0.0& 0.0& 0.0& $\Lambda_{14}$  & 0.0& -0.951& 0.309& 0.0& 0.0& 0.0& 0.0& 0.0& 0.0& 0.0& 0.0& 0.0& 0.0& 0.0& 0.0 & $\Lambda_{13}$  \\
0.364& 0.223& -0.237& 0.0& -0.189& 0.275& -0.031& 0.0& 0.485& -0.292& 0.47& 0.33& 0.0& 0.0& 0.0  &   & 0.148& -0.567& -0.182& 0.0& 0.0& 0.167& -0.52& 0.0& -0.361& 0.0& -0.293& -0.33& 0.0& 0.0& 0.0  &   \\ \hline 
-0.008& 0.357& 0.934& 0.0& 0.0& 0.0& 0.0& 0.0& 0.0& 0.0& 0.0& 0.0& 0.0& 0.0& 0.0& $\Lambda_{6}$ & -0.008& -0.032& -0.999& 0.0& 0.0& 0.0& 0.0& 0.0& 0.0& 0.0& 0.0& 0.0& 0.0& 0.0& 0.0 & $\Lambda_{10}$ \\
-0.285& -0.111& -0.309& 0.0& -0.048& -0.553& 0.243& 0.0& 0.279& -0.223& -0.177& -0.533& 0.0& 0.0& 0.0  &  & -0.14& -0.759& -0.636& 0.0& 0.0& 0.0& 0.0& 0.0& 0.0& 0.0& 0.0& 0.0& 0.0& 0.0& 0.0  &   \\ \hline
-0.21& 0.682& 0.7& 0.0& 0.0& 0.0& 0.0& 0.0& 0.0& 0.0& 0.0& 0.0& 0.0& 0.0& 0.0& $\Lambda_{1}$ &  0.71& -0.379& 0.103& 0.0& 0.0& 0.15& 0.554& 0.0& 0.016& 0.0& -0.11& -0.023& 0.0& 0.0& 0.0 & $\Lambda_{14}$  \\
 0.364& 0.223& -0.237& 0.0& -0.189& 0.275& -0.031& 0.0& 0.485& -0.292& 0.47& 0.33& 0.0& 0.0& 0.0  &   & 0.0& 0.0& 1.0& 0.0& 0.0& 0.0& 0.0& 0.0& 0.0& 0.0& 0.0& 0.0& 0.0& 0.0& 0.0  &   \\ \hline 
 0.0& 0.0& 1.0& 0.0& 0.0& 0.0& 0.0& 0.0& 0.0& 0.0& 0.0& 0.0& 0.0& 0.0& 0.0& $\Lambda_{5}$  & -0.559& -0.182& -0.809& 0.0& 0.0& 0.0& 0.0& 0.0& 0.0& 0.0& 0.0& 0.0& 0.0& 0.0& 0.0 & $\Lambda_{6}$ \\
 -0.217& 0.801& -0.559& 0.0& 0.0& 0.0& 0.0& 0.0& 0.0& 0.0& 0.0& 0.0& 0.0& 0.0& 0.0  &  & 0.532& 0.182& 0.827& 0.0& 0.0& 0.0& 0.0& 0.0& 0.0& 0.0& 0.0& 0.0& 0.0& 0.0& 0.0  &   \\ \hline 
 -0.186& -0.965& -0.184& 0.0& 0.0& 0.0& 0.0& 0.0& 0.0& 0.0& 0.0& 0.0& 0.0& 0.0& 0.0& $\Lambda_9$  & -0.37& -0.006& -0.18& -0.328& -0.136& 0.34& 0.267& -0.196& 0.094& -0.568& -0.173& -0.25& 0.104& -0.001& -0.214 & $\Lambda_{15}$ \\
 0.823& -0.042& -0.306& 0.013& -0.037& -0.27& -0.061& 0.056& 0.022& -0.154& 0.081& -0.038& 0.02& -0.321& 0.099  & & -0.476& 0.431& -0.767& 0.0& 0.0& 0.0& 0.0& 0.0& 0.0& 0.0& 0.0& 0.0& 0.0& 0.0& 0.0  &   \\ \hline
 -0.905& 0.294& 0.309& 0.0& 0.0& 0.0& 0.0& 0.0& 0.0& 0.0& 0.0& 0.0& 0.0& 0.0& 0.0& $\Lambda_{13}$ & 0.0& 0.294& -0.095& -0.407& 0.081& 0.086& 0.266& 0.037& 0.136& -0.17& -0.522& 0.268& 0.104& 0.155& 0.476 & $\Lambda_{9}$  \\
 0.0& 0.0& 1.0& 0.0& 0.0& 0.0& 0.0& 0.0& 0.0& 0.0& 0.0& 0.0& 0.0& 0.0& 0.0  &   & 0.105& -0.301& -0.331& 0.081& -0.089& -0.221& 0.173& 0.111& 0.478& 0.482& -0.288& -0.189& 0.243& 0.189& -0.095  &   \\ \hline
 0.0& -0.588& -0.809& 0.0& 0.0& 0.0& 0.0& 0.0& 0.0& 0.0& 0.0& 0.0& 0.0& 0.0& 0.0& $\Lambda_{9}$  & 0.689& 0.052& 0.626& 0.037& -0.064& 0.105& 0.062& 0.035& 0.053& -0.18& -0.044& 0.01& 0.057& 0.247& -0.083 & $\Lambda_{14}$ \\
-0.284& -0.67& -0.359& 0.008& -0.196& -0.084& 0.312& -0.018& -0.161& -0.093& 0.301& -0.013& 0.109& 0.08& -0.235  & &-0.412& 0.041& -0.135& 0.116& 0.115& 0.183& -0.294& 0.151& -0.046& 0.513& 0.296& -0.498& 0.013& 0.199& 0.022  &   \\ \hline
0.689& 0.052& 0.626& 0.037& -0.064& 0.105& 0.062& 0.035& 0.053& -0.18& -0.044& 0.01& 0.057& 0.247& -0.083& $\Lambda_{2}$  & -0.261& 0.126& -0.429& 0.0& -0.145& -0.323& -0.007& 0.418& -0.151& 0.065& 0.111& 0.136& 0.465& -0.201& -0.342 & $\Lambda_{5}$ \\
-0.412& 0.041& -0.135& 0.116& 0.115& 0.183& -0.294& 0.151& -0.046& 0.513& 0.296& -0.498& 0.013& 0.199& 0.022  & & -0.006& -0.128& 0.138& -0.102& -0.088& 0.413& 0.377& 0.526& -0.033& -0.001& -0.003& 0.1& 0.085& 0.426& 0.397  &   \\ \hline
0.537& -0.166& 0.827& 0.0& 0.0& 0.0& 0.0& 0.0& 0.0& 0.0& 0.0& 0.0& 0.0& 0.0& 0.0& $\Lambda_{10}$ & 0.532& -0.182& 0.827& 0.0& 0.0& 0.0& 0.0& 0.0& 0.0& 0.0& 0.0& 0.0& 0.0& 0.0& 0.0 & $\Lambda_{10}$ \\
0.851& 0.311& -0.423& 0.0& 0.0& 0.0& 0.0& 0.0& 0.0& 0.0& 0.0& 0.0& 0.0& 0.0& 0.0  &  & 0.342& 0.102& 0.934& 0.0& 0.0& 0.0& 0.0& 0.0& 0.0& 0.0& 0.0& 0.0& 0.0& 0.0& 0.0  &   \\ \hline
-0.266& 0.201& 0.719& 0.0& 0.13& -0.283& 0.127& 0.0& -0.038& 0.479& -0.148& 0.072& 0.0& 0.0& 0.0& $\Lambda_{6}$  & 0.718& -0.429& 0.548& 0.0& 0.0& 0.0& 0.0& 0.0& 0.0& 0.0& 0.0& 0.0& 0.0& 0.0& 0.0 & $\Lambda_{7}$ \\
0.203& 0.797& -0.568& 0.0& 0.0& 0.0& 0.0& 0.0& 0.0& 0.0& 0.0& 0.0& 0.0& 0.0& 0.0  &  & 0.975& -0.121& -0.184& 0.0& 0.0& 0.0& 0.0& 0.0& 0.0& 0.0& 0.0& 0.0& 0.0& 0.0& 0.0  &   \\ \hline
0.368& 0.749& -0.339& 0.0& -0.013& -0.032& -0.085& 0.0& 0.234& -0.235& -0.265& -0.031& 0.0& 0.0& 0.0& $\Lambda_{1}$ & -0.311& -0.259& -0.625& 0.0& -0.181& -0.396& 0.254& -0.095& -0.006& -0.121& 0.053& -0.031& 0.019& 0.371& -0.163 & $\Lambda_{12}$ \\
0.0& 0.0& 1.0& 0.0& 0.0& 0.0& 0.0& 0.0& 0.0& 0.0& 0.0& 0.0& 0.0& 0.0& 0.0  &  & 0.07& 0.817& -0.572& 0.0& 0.0& 0.0& 0.0& 0.0& 0.0& 0.0& 0.0& 0.0& 0.0& 0.0& 0.0  &   \\ \hline
-0.37& -0.006& -0.18& -0.328& -0.136& 0.34& 0.267& -0.196& 0.094& -0.568& -0.173& -0.25& 0.104& -0.001& -0.214& $\Lambda_{8}$  & 0.167& 0.794& 0.478& 0.029& 0.133& -0.071& 0.072& -0.01& 0.192& 0.062& -0.171& 0.046& -0.043& 0.059& -0.083 & $\Lambda_{10}$  \\
-0.476& 0.431& -0.767& 0.0& 0.0& 0.0& 0.0& 0.0& 0.0& 0.0& 0.0& 0.0& 0.0& 0.0& 0.0  & & 0.549& -0.728& -0.412& 0.0& 0.0& 0.0& 0.0& 0.0& 0.0& 0.0& 0.0& 0.0& 0.0& 0.0& 0.0  &   \\ \hline
0.0& 0.0& 1.0& 0.0& 0.0& 0.0& 0.0& 0.0& 0.0& 0.0& 0.0& 0.0& 0.0& 0.0& 0.0& $\Lambda_{11}$ & 0.269& -0.119& 0.139& -0.167& 0.115& 0.448& 0.161& -0.39& 0.132& 0.179& -0.103& -0.278& -0.255& -0.52& 0.048 & $\Lambda_{9}$  \\
-0.294& -0.103& 0.815& -0.091& -0.098& 0.057& -0.028& -0.03& 0.179& -0.365& 0.018& 0.031& -0.116& -0.18& -0.065  & & -0.281& -0.661& -0.696& 0.0& 0.0& 0.0& 0.0& 0.0& 0.0& 0.0& 0.0& 0.0& 0.0& 0.0& 0.0  &   \\ \hline
0.0& 0.433& 0.487& -0.176& 0.113& 0.34& -0.302& -0.243& -0.189& -0.147& 0.13& 0.097& -0.265& 0.25& -0.222& $\Lambda_{1}$ & 0.75& -0.131& -0.649& 0.0& 0.0& 0.0& 0.0& 0.0& 0.0& 0.0& 0.0& 0.0& 0.0& 0.0& 0.0 & $\Lambda_{7}$  \\
-0.369& -0.408& -0.097& -0.107& 0.067& 0.011& -0.299& -0.035& 0.011& -0.051& 0.174& 0.036& 0.207& -0.018& -0.71  &   & 0.07& 0.817& -0.572& 0.0& 0.0& 0.0& 0.0& 0.0& 0.0& 0.0& 0.0& 0.0& 0.0& 0.0& 0.0  &   \\ \hline
0.905& 0.294& -0.309& 0.0& 0.0& 0.0& 0.0& 0.0& 0.0& 0.0& 0.0& 0.0& 0.0& 0.0& 0.0& $\Lambda_{4}$  & 0.008& -0.317& -0.462& -0.056& -0.034& -0.47& 0.322& 0.017& 0.054& -0.433& 0.298& -0.087& 0.033& 0.218& -0.149 & $\Lambda_{8}$  \\
 0.264& -0.293& -0.234& 0.0& 0.068& -0.122& 0.231& -0.095& -0.105& 0.083& -0.223& 0.031& -0.324& 0.255& -0.685  & & 0.193& 0.107& -0.041& 0.0& -0.053& 0.124& 0.073& 0.0& 0.298& -0.274& 0.683& 0.543& 0.0& 0.0& 0.0  &   \\ \hline
 -0.104& -0.037& -0.288& 0.0& 0.693& 0.263& -0.283& -0.287& -0.054& 0.144& 0.001& -0.395& 0.039& -0.104& -0.001& $\Lambda_{9}$  & 0.004& -0.352& 0.936& 0.0& 0.0& 0.0& 0.0& 0.0& 0.0& 0.0& 0.0& 0.0& 0.0& 0.0& 0.0 & $\Lambda_{15}$ \\
 0.0& 0.0& 1.0& 0.0& 0.0& 0.0& 0.0& 0.0& 0.0& 0.0& 0.0& 0.0& 0.0& 0.0& 0.0  &   & 0.0& 0.0& 1.0& 0.0& 0.0& 0.0& 0.0& 0.0& 0.0& 0.0& 0.0& 0.0& 0.0& 0.0& 0.0  &   \\ \hline
\bottomrule
\end{tabular}%
}
\caption{\textbf{Two-qubit fundamental digital QNS protocol} The $L = 2$ protocol comprises 80 control sequences and corresponding initial states $\widetilde{\rho}_Q$. These initial states can be derived from combinations of trace-1 physical states. In all experiments, the observable is $O = \sigma_z \otimes \sigma_z$.}
\label{tab:2qb_fundamental_QNS}
\end{table}

\clearpage
\newpage
\bibliography{main.bib}

\begin{thebibliography}{106}
\bibitem [{mul()}]{multi_time_opt_comment}%
  \BibitemOpen
  \href@noop {} {\ }\bibinfo {note} {The digital CA spectra QNS measured at
  time $T$ can also be used to design optimal controls for other times $\{nT/L
  \mid 1 \leq n < L, n \in \mathbb{Z}\}$, where $n$ is an integer. Simultaneous
  optimization beyond a single time can be done through a non-Markovian
  Choi-chanel representation (e.g., as process tensor) where the $\chi$ matrix
  is extended to higher dimensions.}\BibitemShut {Stop}%
\end{thebibliography}
\end{document}